\documentclass[american,aps,pra,reprint,nofootinbib,superscriptaddress,onecolumn,notitlepage]{revtex4-1}
 \usepackage{mathrsfs}
\usepackage[unicode=true,pdfusetitle, bookmarks=true,bookmarksnumbered=false,bookmarksopen=false, breaklinks=false,pdfborder={0 0 0},backref=false,colorlinks=false]{hyperref}
\hypersetup{colorlinks,linkcolor=myurlcolor,citecolor=myurlcolor,urlcolor=myurlcolor}
\usepackage{graphics,epstopdf,graphicx,amsthm,amsmath,amssymb,mathptmx,braket,colortbl,color,bm,framed,mathrsfs}
\usepackage{silence}
\usepackage[T1]{fontenc}
\usepackage[up]{subfigure}
\usepackage{tikz}
\usepackage{algorithm}
\usepackage{algorithmic}
\usepackage{enumerate}
\usepackage{tcolorbox}

\usepackage{booktabs} 
\let\hbar\undefined
\usepackage{fouriernc} 

\definecolor{myurlcolor}{rgb}{0,0,0.9}

\newcommand{\proj}[1]{\left| #1\right\rangle\!\left\langle #1 \right|}

\newcommand{\inner}[2]{\langle #1 , #2\rangle}
\newcommand{\iinner}[2]{\langle #1 | #2\rangle}

\DeclareMathOperator{\trace}{Tr}
\newcommand{\Ptr}[2]{\trace_{#1}\Pa{#2}}
\newcommand{\Tr}[1]{\Ptr{}{#1}}

\newcommand{\Pa}[1]{\left[#1\right]}

\newcommand{\norm}[1]{\left\lVert #1 \right\rVert}

\theoremstyle{plain}
\newtheorem{thm}{Theorem}
\newtheorem{lem}[thm]{Lemma}
\newtheorem{prop}[thm]{Proposition}
\newtheorem{cor}[thm]{Corollary}

\newtheorem{Def}[thm]{Definition}
\newtheorem{Rem}[thm]{Remark}

\newtheorem{Exam}[thm]{Example}

\newcommand*{\myproofname}{Proof}
\newenvironment{mproof}[1][\myproofname]{\begin{proof}[#1]}{\end{proof}}
\def\ot{\otimes}
\def\complex{\mathbb{C}}
\def\real{\mathbb{R}}

\def\CMM{\mathcal{M}}

\def\Z{\mathbb{Z}}

\DeclareMathAlphabet{\mathcal}{OMS}{cmsy}{m}{n}

\makeatother
\begin{document}

  \author{Kaifeng Bu}
 \email{bu.115@osu.edu}
 \affiliation{Department of Mathematics, The Ohio State University, Columbus, Ohio 43210, USA}
\affiliation{Department of Physics, Harvard University, Cambridge, Massachusetts 02138, USA}

  \author{Weichen Gu}
  \email{gu.1213@osu.edu}
   \affiliation{Department of Mathematics, The Ohio State University, Columbus, Ohio 43210, USA}
 \affiliation{Department of Mathematics and Statistics, University of New Hampshire, Durham, New Hampshire  03824, USA}

  \author{Arthur Jaffe}
  \email{Arthur\_Jaffe@harvard.edu}
  \affiliation{Department of Physics, Harvard University, Cambridge, Massachusetts 02138, USA}
 \affiliation{Department of Mathematics, Harvard University, Cambridge, Massachusetts 02138, USA}

 \title{ Stabilizer Testing and Magic Entropy via Quantum Fourier Analysis }

\begin{abstract}
\begin{center}
\textbf{\textit{Dedicated to Huzihiro Araki, whose extraordinary insights into\\ physics and mathematics inspired many works, including this one.}}
\end{center}

Quantum Fourier analysis is an important topic in mathematical physics. We introduce a systematic protocol for testing and measuring ``magic'' in quantum states and gates, using a quantum Fourier approach. Magic, as a quantum resource, is necessary to achieve a
quantum advantage in computation. Our protocols are based on quantum convolutions and swap tests, implemented via quantum circuits. We describe this for both qubit and qudit systems.
Our quantum Fourier approach offers a unified method to quantify magic, in stabilizer circuits, as well as in matchgate and bosonic Gaussian circuits.
\end{abstract}

\maketitle

\tableofcontents
\section{Introduction}

Fourier analysis is a powerful mathematical framework with diverse applications, that  span from information theory to computer science. A key example is its use in linearity testing for classical Boolean functions. This involves property testing, where the goal is to determine if a given function (or state or circuit) possesses a certain property, or if it is within an $\epsilon$-distance of having that property~\cite{goldreich2017introduction}.

In this work, we present a systematic approach for testing and measuring quantum magic by using quantum Fourier analysis. This method leverages the observation that stabilizer states can be viewed as quantum Gaussian states in the quantum Fourier approach, serving as fixed points under quantum convolution. Due to its universality, our approach can also be applied to test and quantify magic in other classically simulable circuits, such as matchgate circuits and bosonic Gaussian circuits.

The key idea in this work stems from the
stability of stabilizer states under quantum convolution.
The protocols we introduce here are inspired by linearity testing in computer science and the study of entanglement entropy in quantum physics.  (We summarize this diagrammatically in Figure \ref{Fig:sum}).
Quantum convolutions have various potential applications that are based on their robust properties and intrinsic relations to Fourier transforms. In this regard, we propose several applications of quantum convolutions, including stabilizer testing for states, Clifford testing for gates, and
magic entropy.

\begin{figure}[h]

\centering
\begin{minipage}{0.45\textwidth}
\begin{tcolorbox}[title=Separability Testing, colback=blue!10]
The reduced state $\rho_A$ of a bipartite pure state $\ket{\phi}_{AB}$
is pure iff $\ket{\phi}_{AB}$ is a separable state.

\begin{tcolorbox}[boxrule=0.3pt,colback=blue!5]
 $\inner{\rho_A}{\rho_A}=1$ iff $\ket{\phi}_{AB}$ is  separable.
\end{tcolorbox}
\end{tcolorbox}
\end{minipage}
\hfill
\begin{minipage}{0.45\textwidth}
\begin{tcolorbox}[title=Linearity Testing,colback=red!10]
The self-convolution $f*f$ of a Boolean function $f$  is a Boolean function
iff $f$ is an affine linear function.

 \begin{tcolorbox}[boxrule=0.3pt,colback=red!5]
$\inner{f*f}{f*f}=1$ iff $f$ is affine linear.
\end{tcolorbox}
\end{tcolorbox}
\end{minipage}

\vskip .1in
\hskip .45in
\begin{tikzpicture}[>->,>=stealth]
    \draw [line width=4pt, color=blue](-1.8,2) -- (0.2,0) node[midway, left] {Inner product $\inner{\ }{\ }$};
       \draw [line width=4pt,color=red](2.2,2) -- (0.2,0) node[midway, right] {~Quantum convolution $\boxtimes$};     
\end{tikzpicture}

\begin{tcolorbox}[title=Stabilizer Testing, colback=green!10]
 The self-convolution $\boxtimes\psi$ of a pure state $\psi$ is  pure, iff 
$\psi$ is a stabilizer state

\begin{tcolorbox}[boxrule=0.3pt,width=8cm,colback=green!5]
$\inner{\boxtimes\psi}{\,\boxtimes\psi}=1$,  iff $\psi$ is stabilizer state.
\end{tcolorbox}

\end{tcolorbox}

\caption{The stabilizer testing in this work is based on the purity invariance of stabilizer states under quantum convolution.
We use that convention that self-convolution $\boxtimes\psi$ means the 2-fold  $\psi\boxtimes\psi$ for qudits and 3-fold self-convolution $\boxtimes_3(\psi,\psi,\psi)$ for qubits.}
\label{Fig:sum}
\end{figure}

\begin{figure}[h]

\begin{minipage}{0.45\textwidth}
\begin{tcolorbox}
[width=10cm,height=3.2cm,title= Stabilizer Test \label{prop:test}]

\begin{enumerate}
    \item 
Perform the convolution  for the given $\psi$, yielding 2 copies of output states;
 
\item
Perform the swap test on  the 2 copies.
If the output is $0$, it passes the test; otherwise, it fails.
\end{enumerate}
\end{tcolorbox}
\end{minipage}
\hfill
\begin{minipage}{0.45\textwidth}
\center{\includegraphics[width=5.5cm]  {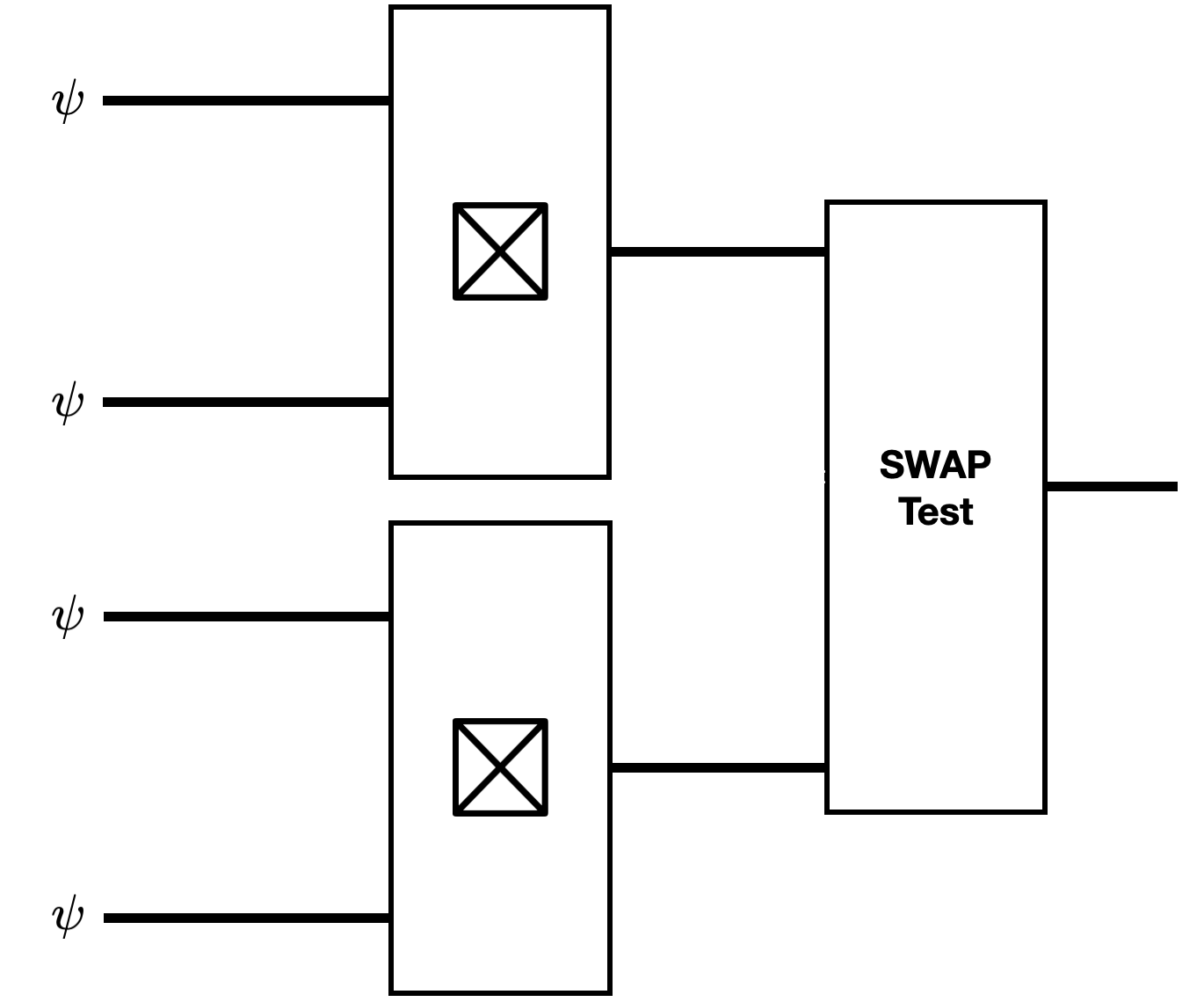}}     
\end{minipage}
\caption{The protocol for the stabilizer test, which we  explain in detail in \S\ref{sec:stab_T}.}

\label{fig:stab_test}
\end{figure}

\subsection{Background}

Blum-Luby-Rubinfeld~(BLR) testing~\cite{blum1990self} is a crucial tool in the field of theoretical computer science. Its importance stems from the use of this test to determine whether or not a function is linear with high probability, using a relatively small number of queries of the function.
BLR testing and its variants have many applications in code testing and cryptography.
For example, the BLR test plays a crucial role in determining whether a code (such as the Hadamard code~\cite{blum1990self} or Reed-Muller code~\cite{Alon05,bhattacharyya2010optimal}) is locally testable. In that case there exists an efficient method to determine with high probability whether a given vector is far from any codeword,  by only checking a small number of bits in the vector.
Locally testable codes can be used in the construction of probabilistically checkable proofs (PCP)~\cite{babai1991checking,babai1991non,feige1991approximating,arora1998proof,Arora98prob,Ben04STOC,Goldreich06,Moshkovitz08,DinurSIAM13}. This is a theoretical framework that deals with the verification of proofs using probabilistic tests. 
We provide a brief review of the BLR linearity testing of Boolean functions in~\S\ref{subsec:pre}.

In quantum physics, quantum state properties, such as entanglement and magic, are crucial for achieving a quantum advantage. This makes  their testing and measurement of significant interest.
For example, 
a pure bipartite state $\phi_{AB}$ is separable if it is a tensor product state  $\phi_{A}\otimes\phi_B$.
A state  which is not separable is called ``entangled.'' 
Entanglement plays a crucial role in quantum information processing and computation, so separability testing is  a key task.
It is well-known that  the reduced state $\rho_A$  of  a pure  bipartite  state $\phi_{AB}$ is pure, if and only if $\phi_{AB}$ is  separable.
Based on this fact, several protocols of separability testing have been proposed~\cite{Harrow13,v011a003,Beckey21,gs007, Buhrman01}. For example, Harrow and Montanaro gave a separability testing method using partial trace and swap tests, and also  discussed the connection with the linearity testing of Boolean functions~\cite{Harrow13}.
Pauli braiding testing and its variants, as a quantum generalization of BLR linearity testing, was proposed by 
Natarajan and Vidick~\cite{natarajan2017quantum,Natarajanfocs18} for robustly verifying entanglement, and plays a fundamental role in their analysis of $\text{MIP}^*=\text{RE}$~\cite{ji2022mipre}.

The relation between separability and purity of a reduced state  led to the use of entanglement entropy as a measure of entanglement in a quantum system. This  entropy is the von Neumann entropy of the reduced state $\rho_A$ of a bipartite pure state $\rho_{AB}$. Entanglement entropy also provides insights into the properties of many-body systems, such as  quantum phase transitions~\cite{VidalPRL03}, quantum field theory~\cite{Calabrese09}, and quantum gravity~\cite{Nishioka09}. Moreover, the entanglement entropy has been experimentally measured~\cite{Islam15} and has become an increasingly important tool for understanding and controlling these systems.

Besides entanglement, stabilizerness is also an important property of 
quantum states and circuits. 
Stabilizer states are the common eigenstates of an abelian subgroup of the qubit Pauli group; they were introduced by Gottesman to study quantum error correction~\cite{Gottesman97}.
From the Gottesman-Knill theorem one knows  that stabilizer circuits comprising Clifford
unitaries with stabilizer states and measurements can be simulated efficiently on a classical computer~\cite{gottesman1998heisenberg}. 
A state which is not a stabilizer is called magic. 
To test whether an unknown state is a stabilizer, several protocols have been 
proposed~\cite{LowPRA09,WangPRA11,Rocchetto18,montanaro2017learning,Gross21}, including the  Bell sampling (first introduced by Montanaro~\cite{montanaro2017learning}), and the Bell difference sampling introduced by Gross, Nezami and Walter~\cite{Gross21}. Further studies and applications have been developed, based on these protocols \cite{LaiIEEE22,grewalLIPICs23,grewal2023improved,Haug23}. Along with stabilizer testing, 
many other measures have been proposed to quantify the amount of magic in a given quantum state~\cite{Veitch12mag,Veitch14,HowardPRL17,BeverlandQST20,SeddonPRXQ21,BravyiPRL16,BravyiPRX16,bravyi2019simulation, Bu19,Bucomplexity22,BuPRA19_stat,RallPRA19,WangNJP19,LeonePRL22,haug2023efficient,Haug23,HaugPRB23, HL2023stabilizer,WangPRA23}. They have been applied in the classical simulation 
of quantum circuits \cite{SeddonPRXQ21,BravyiPRL16,BravyiPRX16,bravyi2019simulation,Bu19,Bucomplexity22,RallPRA19} and unitary  synthesis \cite{HowardPRL17,BeverlandQST20}.

\subsection{Summary of main results}

\begin{enumerate}[(1)] 
\item{}
We design quantum convolutions for 2-qubit-systems. We prove  mathematical properties of these convolutions on $n$-qubit systems, including commutativity with Clifford unitaries, majorization of the spectrum under convolution, and the purity invariance of stabilizer states under convolution.

\item{}
In \S\ref{sec:stab_T}, we propose a systematic framework for  stabilizer and Clifford testing  based on quantum convolutions and swap tests on both qudits and qubits (See Figure~\ref{fig:stab_test}  as an example).  
We use the Hadamard convolution for qudits from our previous work~\cite{BGJ23a,BGJ23b,BJ24a,BGJ24a}, and the quantum convolution constructed for qubits  in \S\ref{sec:con_qubit}.
 If the maximal overlap between the given state $\psi$ and stabilizer states is $1-\epsilon$, the
 probability of acceptance in the protocol is $1-\Theta(\epsilon)$. This bound is independent of the number $n$ of qudits/qubits 
and of  the local dimension $d$.  
We provide both an upper bound and a lower bound on the probability 
of acceptance; these bounds are close to each other, up to an error of order $\epsilon^2$.
By the Choi-Jamiołkowski isomorphism, we can also use this protocol to perform Clifford testing of quantum gates.

\item{}
Inspired by entanglement entropy, 
we introduce  "magic entropy" on both qudits and qubits in~\S\ref{sec:mag_en}. This is the von Neumann entropy (or quantum R\'enyi entropy) of
 the self-convolution $\boxtimes\psi$ for the given state $\psi$.  By the Choi-Jamiołkowski isomorphism, 
we can  generalize the concept of the magic entropy to quantum gates. 
We study the properties of magic entropy and show that it can be used as a measure of magic. 
\end{enumerate}

Aside from  these main results,  we also discuss some possible future directions to extend this work in~\S\ref{sec:disc}, including 
 experimental realization, the concept of magic spectrum, and potential connections with pseudo-random states, quantum error-correction codes, among other topics. We also have begun to  investigate how quantum higher-order Fourier analysis can quantify quantum complexity~\cite{BGJ25a}.

\begin{Rem}
In addition to stabilizer states and circuits, other classically simulable families exist, including fermionic Gaussian states/circuits (also known as matchgates)~\cite{valiant2001quantum,valiant2002quantum,bravyi2002fermionic,terhal2002classical,divincenzo2004fermionic} and bosonic Gaussian states/circuits~\cite{Bartlett02,Mari12,Veitch_2013}. It is worth noting that the method proposed here for testing stabilizer states can be extended to both fermionic and bosonic Gaussian states  through alternative choices  of quantum convolutions. (See follow-up works for details on fermionic Gaussian states~\cite{lyu2024fermionic,lyu2024fermionic_G} and bosonic Gaussian states~\cite{Bu2025CV}.) Therefore, the present work provides a universal quantum Fourier theoretical framework for characterizing Gaussianity in both discrete-variable and continuous-variable quantum systems.
\end{Rem}

\section{Preliminary}
\subsection{Pauli operators,  stabilizer states and Clifford unitaries}
An $n$-qudit system is a Hilbert space $\mathcal{H}^{\ot n}$, where $\mathcal{H}\simeq \complex^d$, and $d$ is prime. 
Let $L(\mathcal{H}^{\ot n})$ denote the set of all linear operators on $\mathcal{H}^{\ot n}$,
and let 
$D(\mathcal{H}^{\ot n})$ denote the set of  all quantum states on $\mathcal{H}^{\ot n}$.
We define one  orthonormal set in  the Hilbert space $\mathcal{H}$,  to be the computational basis and denote it by the  Dirac notation
$
\set{\ket{k}}_{k\in \mathbb{Z}_d}\;.
$
 The Pauli matrices $X$ and $Z$ are defined as
\[ X |k\rangle = |k+1\rangle\;, 
\quad 
Z |k\rangle = \chi(k) |k\rangle,\;\;\;\forall k\in \Z_d\;,
\]
where $\chi(k)=\omega^k_d$ and $\omega_d=\exp(2\pi i /d)$ is a $d$-th root of unity.
 If the local dimension $d$ is an odd prime number,  the Pauli operators (or Weyl operators)
are defined as 
\begin{eqnarray}\label{eqn:wpq}
w(p,q)=\chi(-2^{-1}pq)Z^pX^q\;.
\end{eqnarray}
Here $ 2^{-1}$ denotes the inverse of 2 in $\mathbb{Z}_d$.
If $d=2$, the Pauli operators are Hermitian and  defined as
\begin{eqnarray}
w(p,q)=i^{-pq}Z^pX^q \;.
\end{eqnarray}

The Pauli operators satisfy, 
\begin{eqnarray}\label{0106shi2}
w(p,q)\,w(p',q')=
 \chi(2^{-1}\inner{(p,q)}{(p',q')}_s)
 w(p+p',q+q') \; ,
\end{eqnarray}
if  $d>2$. When $d=2$, 
\begin{eqnarray}\label{0106shi2_2}
w(p,q)\,w(p',q')=
i^{\inner{(p,q)}{(p',q')}_s}\,
 w(p+p',q+q') \;.
\end{eqnarray}
In both cases, the symplectic inner product $\inner{(p,q)}{(p',q')}_s$ is defined as 
\begin{eqnarray}\label{SymplecticInnerProduct}
\inner{(p,q)}{(p',q')}_s
=pq'-qp' \; .
\end{eqnarray}
Let us denote 
$
V^n=\mathbb{Z}^n_d\times\mathbb{Z}^n_d \; 
$,
and for any $(\vec p, \vec q)\in V^n$, the Pauli operator $w(\vec{p}, \vec q)$ is defined as 
\begin{eqnarray*}
w(\vec p, \vec q)
=w(p_1, q_1)\ot...\ot w(p_n, q_n) \; ,
\end{eqnarray*}
with $\vec p=(p_1, p_2,..., p_n)\in \mathbb{Z}^n_d,\ \vec q=(q_1,..., q_n)\in \mathbb{Z}^n_d $.
The set of the Pauli operators $\set{w(\vec{p},\vec q)}_{(\vec p, \vec q)\in V^n}$ forms an orthonormal basis 
in $L(\mathcal{H}^{\ot n})$ with respect to the inner product 
$\inner{A}{B}=\frac{1}{d^n}\Tr{A^\dag B}$.

\begin{Def}[\bf Characteristic function]\label{def:charfn}
For any $n$-qudit state $\rho\in \mathcal{D}(\mathcal{H}^{\ot n})$, the characteristic function $\Xi_{\rho}:V^n\to\complex$ is defined
as 
\begin{eqnarray}\label{eq:charfn}
\Xi_{\rho}(\vec{p},\vec q)
=\Tr{\rho w(-\vec{p},-\vec q)} \; .
\end{eqnarray}
\end{Def}
Hence, the state $\rho$ can be written as a linear combination of the Pauli operators with characteristic function
\begin{eqnarray}\label{0109shi5}
\rho=\frac{1}{d^n}
\sum_{(\vec{p},\vec q)\in V^n}
\Xi_{\rho}(\vec{p},\vec q)w(\vec{p},\vec q) \; .
\end{eqnarray}
The process of taking characteristic functions is the quantum Fourier transform that we consider. 
The characteristic function has been used to study quantum Boolean functions~\cite{montanaro2010quantum}, and was later  applied to study quantum circuit complexity~\cite{Bucomplexity22}, and quantum scrambling~\cite{GBJPNAS23}.
(See also a more general framework of quantum Fourier analysis \cite{JaffePNAS20}.)

\begin{Def}[\bf Stabilizer state~\cite{Gottesman96, Gottesman97}]
A unit vector  $\ket{\psi}$ is a stabilizer vector if  there exists a maximal  abelian subgroup $S$ of the 
Pauli operators with $n$ generators $\set{w(\vec{p}_i, \vec q_i)}_{i\in[n]}$ such that $w(\vec{p}_i, \vec q_i)\ket{\psi}=\chi(x_i)\ket{\psi}$ with 
$x_i\in \mathbb{Z}_d$ for every $i\in[n]$. 
The corresponding state with density operator $\proj{\psi}$ is 
the projection onto the eigenstate. A general mixed stabilizer state $\rho$ is a  convex linear combination of pure stabilizer states.
\end{Def}

In general, every abelian subgroup $S$ of the Pauli operators has size $d^r$ with $r\in[n]$.
The operators in $S$ generate an abelian $C^*$-algebra $C^*(S)$.  The projections in $C^*(S)$ are called the \emph{stabilizer projections} associated with $S$.

\begin{Def}[\bf Minimal stabilizer-projection state]\label{def:MSPS}
Given an abelian subgroup of Pauli operators $S$,
a minimal projection in $C^*(S)$ is called a \emph{minimal stabilizer projection} associated with $S$.
A minimal stabilizer-projection state (MSPS) is a minimal stabilizer projection normalized by  dividing  by its dimension. 
\end{Def}

It is clear that if $P$ is a stabilizer projection associated with a subgroup $S$ of an abelian group $S'$,
then $P$ is also associated with $S'$.
In addition, when some stabilizer projection $P$ is given,
there is a unique minimal abelian subgroup $S$ associated with $P$, in the sense that for every  $S'$ associated with $P$, we have $S\subseteq S'$.
For example, let us consider the abelian group $S=\set{Z_1,...,Z_{n-1}}$ for an $n$-qudit system.
The states $\left\{\frac{1}{d}\proj{\vec j}\ot I \right\}_{\vec j\in\mathbb{Z}^{n-1}_d}$ are MSPS. 

\begin{Def}[\bf Clifford unitary]
An  $n$-qudit unitary $U$ is a Clifford unitary if conjugation by $U$ maps every Pauli operator to another Pauli operator, up to a phase.
\end{Def}
Clifford unitaries map stabilizer states to stabilizer states. We have the following definition of a stabilizer channel:
\begin{Def}[\bf Stabilizer channel]
A  quantum  channel is a stabilizer channel if it maps stabilizer states 
to stabilizer states.
\end{Def}

\subsection{The BLR linearity test for Boolean functions}\label{subsec:pre}
Functions $f:\set{0,1}^n\to \set{+1,-1}$ are called Boolean functions. Such a function $f$ is  linear, if $f(\vec x)f(\vec y)=f(\vec x+\vec y)$ for any $\vec x, \vec y\in\set{0,1}^n$. 
A famous test to decide
whether $f$ is linear was given by Blum-Luby-Rubinfeld~\cite{blum1990self}:

\begin{center}
  \begin{tcolorbox}[width=10cm,height=3.2cm,title=BLR Linearity Test]
 \begin{enumerate}
    \item Choose~~$\vec x, \vec y\in \set{0,1}^n $ ~to be uniformly random;
    \item Query $f(\vec x)$, $f(\vec y)$, and $f(\vec x+\vec y)$;
    \item Accept  if $f(\vec x)f(\vec y)=f(\vec x+ \vec y)$. Reject, otherwise.
\end{enumerate}
    \end{tcolorbox}
   \end{center}
The probability of acceptance is \begin{eqnarray}
\text{Pr}_{\text{accep}}(f)
=\frac{1}{2}\left[1+\mathbb{E}_{\vec x, \vec y}f(\vec x)f(\vec y)f(\vec x+\vec y)
\right],
\end{eqnarray}
where both $\mathbb{E}_{\vec x}$ and $\mathbb{E}_{\vec y}$   denote the expectation 
taken over the uniform distribution on $\set{0,1}^n$. 
This expression  can be rewritten  as
\begin{eqnarray}
\text{Pr}_{\text{accep}}(f)
=\frac{1}{2}\left[1+\inner{f}{f*f}
\right],
\end{eqnarray}
where the convolution $f*g$ of  Boolean functions $f$ and $g$ is  $f*g(\vec x)=\mathbb{E}_{\vec y}f(\vec y)g(\vec x+\vec y)$. The inner product between  Boolean functions $f$ and $g$ is  $\inner{f}{g}=\mathbb{E}_{\vec x}f(\vec x)g(\vec x)$.
\medskip

Related to this is the affine linearity test:
\begin{center}
  \begin{tcolorbox}[width=10cm,height=3.2cm,title=Affine Linearity Test]
 \begin{enumerate}
    \item Choose~~$\vec x, \vec y, \vec z\in \set{0,1}^n $ ~to be uniformly random;
    \item Query $f(\vec x)$, $f(\vec y)$, $f(\vec z)$, and $f(\vec x+\vec y+\vec z)$;
    \item Accept  if $f(\vec x)f(\vec y)f(\vec z)=f(\vec x+ \vec y+\vec z)$. Reject, otherwise.
\end{enumerate}
    \end{tcolorbox}
   \end{center}
The probability of acceptance 
is \begin{eqnarray}
\text{Pr}_{\text{accep}}(f)
=\frac{1}{2}\left[1+\mathbb{E}_{\vec x, \vec y,\vec z}f(\vec x)f(\vec y)f(\vec z)f(\vec x+\vec y+\vec z)
\right].
\end{eqnarray}
This probability  can be  rewritten as 
\begin{eqnarray}
\text{Pr}_{\text{accep}}(f)
=\frac{1}{2}\left[1+\inner{f*f}{f*f}
\right].
\end{eqnarray}
Thus $\text{Pr}_{\text{accep}}(f)=1$,   iff $f$ is an affine linear function. This implies that
$f*f$ is a Boolean function, iff $f$ is an affine linear function.
These tests provide one inspiration for the quantum tests that we explore here.

\section{Quantum convolutions on qubits}\label{sec:con_qubit}
We introduce the concept of quantum convolutions  for $K$ quantum systems, where each system contains $n$ qubits and $K=2N+1$ is odd. 
It is important to note that this differs from the convolution between two $n$-qubit systems introduced in \cite{BGJ23a,BGJ23b}. 
We choose $K$ to be odd, in order to make the  characteristic function of 
the convolution become multiplication of the characteristic functions of the input states.

In particular, we study the basic properties of quantum convolutions, including the multiplicative behavior 
characteristic functions under convolution, the commutativity of Clifford unitaries with 
convolutions,  and purity invariance of stabilizer states under convolution.
These results follow from the definition of the key unitary $V$ and its action on Pauli matrices.

\begin{Def}[\bf Key Unitary]\label{def:key_U}
The key unitary $V$ for  $K$ quantum systems, with each system containing $n$ qubits, is 
\begin{eqnarray}\label{eq:con_cir}
V:=U^{\ot n}
=U_{1,n+1,...,(K-1)n+1}\ot U_{2,n+2,...,(K-1)n+2}\ot ...\ot U_{n, 2n,..., Kn}.
\end{eqnarray}
Here  $U$ is a $K$-qubit unitary constructed using CNOT gates:
\begin{eqnarray}
U:=\left(\prod^K_{j=1}CNOT_{j\to 1}\right)\left(\prod^K_{i=1}CNOT_{1\to i}\right),
\end{eqnarray}
and 
$
CNOT_{2\to 1}\ket{x}\ket{y}=\ket{x+y}\ket{y}
$ for any $x,y\in\mathbb{Z}_2$. 
\end{Def}

\begin{Def}[\bf Convolution of multiple states]\label{Def:conv_qubit}
Given $K$ states $\rho_1,\rho_2,..., \rho_K$, each with $n$-qubits, the multiple convolution $\boxtimes_K$ 
of $\rho_1,\rho_2,..., \rho_K$ maps to an $n$-qubit state: 
\begin{eqnarray}
\boxtimes_{K}(\rho_1,\rho_2,...,\rho_K)=\boxtimes_K(\ot^K_{i=1}\rho_i)=\Ptr{1^c}{V\ot^K_{i=1}\rho_i V^\dag}\;.
\end{eqnarray}
Here $V$ is the key unitary in Definition~\ref{def:key_U}, and
$\Ptr{1^c}{\cdot}$ denotes the partial trace taken on the subsystem $2, 3..., K$, i.e., 
 $\Ptr{1^c}{\cdot}=\Ptr{2,3,...,K}{\cdot}$.
\end{Def}

This convolution $\boxtimes_{K}$ gives a quantum channel, which we also denote  as $\boxtimes_K$, i.e.,
$
\boxtimes_K(\cdot)
=\Ptr{1^c}{V \cdot V^\dag}
$.
Therefore, when we refer to $\boxtimes_K(\rho_1,\rho_2,...,\rho_K)$ or $\boxtimes_K(\rho_1\ot \cdots \ot \rho_K)$, we are referring to the action of the convolutional channel on the input states. In this work, we will use the terms "convolution" and "convolutional channel" interchangeably without distinction.
If the given $K$ states $\set{\rho_i}^K_{i=1}$ are the same, we denote their convolution as $\boxtimes_K\rho$ for simplicity. 

We use $\boxtimes_K$ to denote the convolution on $K$ input states, while  in our previous work~\cite{BGJ23a,BGJ23b} we used $\boxtimes_K$ to indicate the repeated $2$-fold convolution $K$ times. 
In fact, these two concepts are almost the same.  In Proposition~\ref{prop:3gen_all}, we show this  by repeating the $3$-fold convolution.
This is the reason that we use the notation $\boxtimes_K$ in this paper.

Since the key unitary consists of CNOT gates, its action on the computational basis 
can be expressed directly as follows.
\begin{lem}\label{230609lem1}
The action of the key unitary $V$ on the computational basis $\set{\ket{\vec x}}$ is
\begin{eqnarray}
V \left(\ot^{K}_{i=1}\ket{\vec x_i}\right)=\ket{\sum^K_{i=1}\vec x_i}\ot^K_{i=2}\ket{ \vec x_i+ \vec x_1},
\end{eqnarray}
and 
\begin{eqnarray}
V^\dag \left(\ot^{K}_{i=1}\ket{\vec x_i}\right)
=\ket{\sum^K_{i=1}\vec x_i}\ot^K_{i=2}
\ket{\sum_{j\neq i} 
\vec x_j}.
\end{eqnarray}
\end{lem}

\begin{prop}\label{prop:action_key}
The action of the  key unitary $V$ acting on the Pauli operators satisfies:
\begin{eqnarray}
V\ot^K_{k=1}w(\vec p_k, \vec q_k)V^\dag
=(-1)^{N\sum^K_{j=1}\vec p_j\cdot\vec q_1}\,
w\left(\sum^K_{j=1}\vec p_j, \sum^K_{j=1}\vec q_j\right)
\ot^{K}_{k=2}
w\left(\sum^K_{j=1}\vec p_j-\vec p_k, \vec q_1-\vec q_k\right), \label{230609shi1}
\end{eqnarray}
and thus 
\begin{eqnarray}
V^\dag\ot^K_{k=1}w(\vec p_k, \vec q_k)V
=(-1)^{N\sum^K_{j=1}\vec p_1\vec q_j}w\left(\sum^K_{j=1}\vec p_j, \sum^K_{j=1}\vec q_j\right)
\ot^{K}_{k=2}
\,
w\left(\vec p_1-\vec p_k, \sum_{j=1}^K\vec q_j-\vec q_k\right),
\end{eqnarray}
for any $(\vec p_k, \vec q_k)\in V^n$.
\end{prop}
\begin{proof}
Based on Lemma~\ref{lem:weylact},  the left hand side of \eqref{230609shi1} is equal to 
\begin{eqnarray*}
V\ot^K_{k=1}w(\vec p_k, \vec q_k)V^\dag
=i^{-\sum_j\vec p_j\cdot \vec q_j}
\left( 
Z^{\sum^K_{j=1}\vec p_j}
\ot^K_{k=2}
Z^{\sum^K_{j=1}\vec p_j-\vec p_k}
\right)
\left(
X^{\sum^K_{j=1}\vec q_j}
\ot^K_{k=2}
X^{\vec q_1+\vec q_k}
\right).
\end{eqnarray*}
The right hand side of \eqref{230609shi1} has
\begin{align*}
&w\left(\sum^K_{j=1}\vec p_j, \sum^K_{j=1}\vec q_j\right)
\ot^{K}_{k=2}
w\left(\sum^K_{j=1}\vec p_j-\vec p_k, \vec q_1-\vec q_k\right)\\
=&i^{-(\sum^K_{j=1}\vec p_j)\cdot (\sum^K_{j=1}\vec q_j)-\sum^K_{k=2}(\sum^K_{j=1}\vec p_j-\vec p_k)\cdot (\vec q_1-\vec q_k)}
\left( 
Z^{\sum^K_{j=1}\vec p_j}
\ot^K_{k=2}
Z^{\sum^K_{j=1}\vec p_j-\vec p_k}
\right)
\left(
X^{\sum^K_{j=1}\vec q_j}
\ot^K_{k=2}
X^{\vec q_1+\vec q_k}
\right).
\end{align*}
It is easy to verify that 
\begin{eqnarray*}
\left(\sum^K_{j=1}\vec p_j\right)\cdot \left(\sum^K_{j=1}\vec q_j\right)+
\sum^K_{k=2}\left(\sum^K_{j=1}\vec p_j-\vec p_k\right)\cdot \left(\vec q_1-\vec q_k\right)
=\sum^K_{j=1}\vec p_j\cdot \vec q_j
+(K-1)\left(\sum^K_{j=1}\vec p_j\right)\cdot\vec q_1.
\end{eqnarray*}
Then we get the final result by $K=2N+1$.

\end{proof}

\begin{lem}\label{lem:weylact}
The following identities hold:
\begin{eqnarray}
UX_1U^\dag =\prod^K_{i=1} X_i\;,\qquad UX_iU^\dag =X_1  X_i,\quad \text{for}~ i\geq 2,
\end{eqnarray}
and 
\begin{eqnarray}
UZ_1U^\dag=\prod^K_{i=1}Z_i\;,\qquad
UZ_iU^\dag=Z_{i^c},\quad \text{for}~ i\geq 2,
\end{eqnarray}
where $Z_{i^c}:=Z_1  Z_2\cdots Z_{i-1} I_iZ_{i+1}\cdots Z_K$.
\end{lem}
\begin{proof}
\begin{align*}
UX_1 U^\dag
=&\sum_{x_1,..x_K}
U\ket{x_1+1}\bra{x_1}\ot^K_{i=2}\proj{x_i}U^\dag\\
=&\sum_{x_1,..x_K}
\ket{\sum_ix_i+1}\bra{\sum_ix_i}
\ot^K_{i=2}\ket{x_1+x_i+1}\bra{x_1+x_i}\\
=&\sum_{y_1,..y_K}
\ot^K_{i}\ket{y_i+1}\bra{y_i}\\
=&X\ot X...\ot X=\prod_iX_i.
\end{align*}
If $i\geq 2$, for example $i=2$, 
\begin{align*}
    UX_2 U^\dag
    =&\sum_{x_1,..x_K}U\proj{x_1}\ot\ket{x_2+1}\bra{x_2}\ot^{K}_{i=3}\proj{x_i}U^\dag\\
    =&\sum_{x_1,..x_K}\ket{\sum_{i}x_i+1}\bra{\sum_ix_i}\ot\ket{x_1+x_2+1}\bra{x_1+x_2}\ot^{K}_{i=3}\proj{x_1+x_i}\\
    =&\sum_{y_1,...,y_K}\ket{y_1+1}\bra{y_1}\ot \ket{y_2+1}\bra{y_2}\ot^{K}_{i=3} \proj{y_i}\\
    =&X_1  X_2.
\end{align*}

Also,
\begin{align*}
    UZ_1U^\dag=&\sum_{x_1,..x_K}
U(-1)^{x_1}\ot^K_{i=1}\proj{x_i}U^\dag\\
=&\sum_{x_1,..x_K}
(-1)^{x_1}\proj{\sum_ix_i}\ot^K_{i=2}\proj{x_1+x_i}\\
=&\sum_{y_1,...,y_K}(-1)^{\sum_iy_i}
\ot^K_{i=1}\proj{y_i}\\
=&Z\ot Z...\ot Z=\prod_iZ_i.
\end{align*}
For $i\geq 2$, for example $i=2$, we have 
\begin{align*}
    UZ_2U^\dag=&\sum_{x_1,..x_K}
U(-1)^{x_2}\ot^K_{i=1}\proj{x_i}U^\dag\\
=&\sum_{x_1,..x_K}
(-1)^{x_2}\proj{\sum_ix_i}\ot^K_{i=2}\proj{x_1+x_i}\\
=&\sum_{y_1,...,y_K}(-1)^{\sum_{i\neq 2}y_i}
\ot^K_{i=1}\proj{y_i}\\
=&Z\ot I\ot Z...\ot Z.
\end{align*}
\end{proof}

\begin{prop}[\bf Reduction to the classical convolution]
Given $K$ states $\rho_1,\rho_2,..., \rho_K$, where each $\rho_i=\sum_{\vec x}p_i(\vec x)\proj{\vec x}$ is  diagonal  in the computational basis, then the convolution $\boxtimes_K\left(\otimes^K_{i=1}\rho_i\right)=\sum_{\vec x}q(\vec x)\proj{\vec x}$ is a diagonal state with 
\begin{eqnarray}
q(\vec x)
=p_1*p_2*...*p_K(\vec x).
\end{eqnarray}
Here $p*q(\vec x):=\sum_{\vec y\in \set{0,1}^n}p(\vec y)q(\vec y+\vec x)$ is the classical convolution of  two probability distributions $p$ and 
$q$ on $ \set{0,1}^n$.
\end{prop}
\begin{proof}
Based on the definition of classical convolution $*$, $q(\vec x)
=p_1*p_2*...*p_K(\vec x)$ is equal to $q(\vec x)
=\sum_{\vec x_1,..., \vec x_K:\sum_i\vec x_i=\vec x}\prod_ip(\vec x_i)$.
And by Lemma \ref{230609lem1}, the output state of the quantum convolution $\boxtimes_K$ is
\begin{eqnarray}
\boxtimes_K(\otimes^K_{i=1}\rho_i)
=
\Ptr{1^c}{V\ot^K_{i=1}\rho_i V^\dag}
=\sum_{\vec x_1,...,\vec x_K}
\prod_ip_i(x_i)\proj{\sum_i\vec x_i}.
\end{eqnarray}
Hence, $\bra{\vec x}\boxtimes_K(\otimes^K_{i=1}\rho_i)\ket{\vec x}=q(\vec x)$.
\end{proof}

\begin{prop}[\bf Convolution-multiplication duality]\label{230504prop1}
Given $K$ states $\rho_1,...,\rho_K$, the characteristic function of 
their convolution $\boxtimes_K(\ot^K_{i=1}\rho_i)$
 satisfies
\begin{eqnarray}
\Xi_{ \boxtimes_K(\ot^K_{i=1}\rho_i)}(\vec p, \vec q)
=(-1)^{N \vec p\cdot \vec q} \prod_{j=1}^K \Xi_{\rho_j}(\vec p, \vec q)
, \quad \forall (\vec p, \vec q)\in V^n.
\end{eqnarray}
\end{prop}
\begin{proof}
This is because $\boxtimes_K^\dag(w(\vec p,\vec q))=V^\dag\left( w(\vec p,\vec q)\ot I...\ot I\right)V$, where the action of the key unitary on $w(\vec p, \vec q)$ is given in Proposition \ref{prop:action_key}.
Hence
\begin{eqnarray*}
\Xi_{ \boxtimes_K(\ot^K_{i=1}\rho_i)}(\vec p, \vec q)
&=&\Tr{ \boxtimes_K(\ot^K_{i=1}\rho_i)w(\vec p, \vec q)}
=\Tr{ \ot^K_{i=1}\rho_i \boxtimes^\dag_K(w(\vec p, \vec q))}\\
&=&(-1)^{N\vec p \cdot \vec q } \prod^K_{i=1}\Tr{\rho_iw(\vec p, \vec q)}
=(-1)^{N \vec p\cdot \vec q} \prod_{j=1}^K \Xi_{\rho_j}(\vec p, \vec q).
\end{eqnarray*}

\end{proof}

Based on the characteristic function of the convolution $\boxtimes_K$, we find that any 
$K$-fold convolution $\boxtimes_K$ can be generated by $\boxtimes_3$. 

\begin{prop}[\bf $\boxtimes_3$ generates all $\boxtimes_K$ ]\label{prop:3gen_all}
Given $K$ $n$-qubit states $\rho_1,...,\rho_K$ with $K=2N+1$, the convolution $\boxtimes_K$ can be generated by 
repeating the 3-fold convolution $\boxtimes_3$ $N$ times, that is,
\begin{eqnarray}
\boxtimes_K\left(\ot^K_{i=1}\rho_i\right)
=\boxtimes_3\left(\rho_1\ot \rho_2\ot \boxtimes_{K-2}\left(\ot^K_{i=3}\rho_i\right)\right).
\end{eqnarray}
\end{prop}

Since every convolution $\boxtimes_K$
can  be generated by $\boxtimes_3$, we focus on the properties of 
$\boxtimes_3$. Let us first introduce some basic concepts, including the magic gap, majorization, and 
quantum R\'enyi entropy.

\begin{Def}[\cite{BGJ23a,BGJ23b}]
Given an $n$-qudit state $\rho$ for any integer $d$,  the mean state of $\rho$ is the 
operator $\CMM(\rho)$  with the characteristic function: 
\begin{align}\label{0109shi6}
\Xi_{\CMM(\rho)}(\vec p, \vec q) =
\left\{
\begin{aligned}
&\Xi_\rho ( \vec p, \vec q) \; , && |\Xi_\rho ( \vec p, \vec q)|=1 \; ,\\
& 0 \; , && |\Xi_\rho (  \vec p, \vec q)|<1 \; .
\end{aligned}
\right.
\end{align}
\end{Def}

\begin{Def}[\cite{BGJ23a,BGJ23b}]
Given an $n$-qudit state $\rho$,  the  magic gap  of $\rho$ is 
\begin{eqnarray*}
MG(\rho)=1-\max_{(\vec{p}, \vec q)\in  \text{Supp}(\Xi_{\rho}): |\Xi_{\rho}(\vec{p},\vec q)|\neq 1}|\Xi_{\rho}(\vec{p},\vec q)| \;.
\end{eqnarray*}
If $\set{(\vec{p},\vec q)\in  \text{Supp}(\Xi_{\rho}): |\Xi_{\rho}(\vec{p}, \vec q)|\neq 1}=\emptyset$, 
then $MG(\rho):=0$, i.e., there is no gap on the support.

\end{Def}

\begin{Def}[\bf Majorization \cite{MOA79}]\label{def:major}
Given two probability vectors $\vec p=\set{p_i}_{i\in[n]}$ and $\vec q=\set{q_i}_{i\in[n]}$, $\vec p$ is said to
be majorized by $\vec q$, denoted as  $\vec p\prec \vec q$,  if 
\begin{eqnarray*}
\sum^k_{i=1}p^{\downarrow}_i
&\leq& \sum^k_{i=1} q^{\downarrow}_i \;, ~~\forall 1\leq k\leq n-1 \;,\\
\sum^n_{i=1}p_i&=&\sum^n_{i=1} q_i=1 \;,
\end{eqnarray*}
where $p^{\downarrow}$ is the vector obtained by rearranging the components of $p$ in decreasing order, i.e., 
$p^{\downarrow}=(p^{\downarrow}_1,p^{\downarrow}_2,...,p^{\downarrow}_n)$ and $p^{\downarrow}_1\geq p^{\downarrow}_2\geq...\geq p^{\downarrow}_n$.
\end{Def}

\begin{Def}[\bf Quantum R\'{e}nyi entropy \cite{BHOW15}]\label{def:gen_ren_en}
For any  $\alpha\in[0,+\infty]$,
 the R\'enyi entropy $H_\alpha(\rho)$ for a quantum state $\rho$ is
\begin{eqnarray*}
S_{\alpha}(\rho)
=\frac{1}{1-\alpha}\log \Tr{\rho^\alpha} \;,
\end{eqnarray*}
\end{Def}

After introducing these basic concepts, we investigate the properties of 
quantum convolution and summarize its main properties in the following theorem.  We prove these properties  in  Appendix~\ref{appen:quan_con}.

\begin{thm}[\bf Properties of the quantum convolution]\label{thm:main_conv}
The  quantum convolution $\boxtimes_3$ satisfies: 

(1) \textbf{$\boxtimes_3$ is symmetric:} $\boxtimes_3\left(\ot^3_{i=1}\rho_i\right)=\boxtimes_3\left(\ot^3_{i=1}\rho_{\pi(i)}\right)$, for any permutation $\pi$ on 3 elements.

(2) \textbf{Convolutional stability for states:} If $\rho_1, \rho_2, \rho_3$ are all stabilizer states, then 
$\boxtimes\left(\ot^3_{i=1}\rho_i\right)$ is a stabilizer state.

(3) \textbf{Majorization under convolution:} Given $3$ $n$-qubit states $\rho_1,\rho_2,\rho_3$, with 
$\vec{\lambda}_{1}, \vec{\lambda}_2, \vec{\lambda}_3$ the vectors of  eigenvalues of $\rho_1, \rho_2, \rho_3$ respectively, we have 
\begin{eqnarray}\label{inedq:major_21}
\vec{\lambda}_{\boxtimes_3\left(\ot^3_{i=1}\rho_i\right)} 
\prec \vec{\lambda}_{\rho_j} \;,\quad j=1,2, 3.
\end{eqnarray}
This ensures that the quantum R\'enyi  entropy satisfies the following bound: 
\begin{eqnarray}\label{ineq:entopy_con}
S_{\alpha}\left(\boxtimes_3\left(\ot^3_{i=1}\rho_i\right)
\right)\geq \max_iS_{\alpha}(\rho_i), \quad\forall \alpha\geq 0.
\end{eqnarray}

(4) \textbf{Purity invariance of stabilizer states:} given a pure state $\psi$, the self-convolution $\boxtimes_3\psi$   is pure iff 
$\psi$ is a stabilizer state.

(5) \textbf{Commutativity with Clifford unitaries:} 
for any Clifford unitary $U$, there exists some Clifford unitary $U_1$ such that 
\begin{eqnarray}\label{230520shi1}
U_1\left(\boxtimes_3\left(\ot^3_{i=1}\rho_i\right)\right) U_1^\dag
=\boxtimes_3\left(\ot^3_{i=1} \left(U\rho_i U^\dag\right)\right) \;,\quad \forall \rho_1, \rho_2,\rho_3\;.
\end{eqnarray}

(6) \textbf{Mean state properties:}
Let $\rho_1,\rho_2,\rho_3$ be three $n$-qudit states. 
Then  
\begin{eqnarray}\label{0204shi6-1}
\CMM\left(\boxtimes_3\left(\ot^3_{i=1}\rho_i\right)\right)=
\boxtimes_3\left(\ot^3_{i=1}\CMM(\rho_i)\right)\;.
\end{eqnarray} 

(7) \textbf{Quantum central limit theorem:} Let $\rho$ be an $n$-qubit state and $K=2N+1$,
then
\begin{eqnarray}\label{230509shi22}
\norm{\boxtimes_{K}\rho-\mathcal{M}(\rho^\#)}_2
\leq (1-MG(\rho))^{K-1} \norm{\rho- \mathcal{M}(\rho)}_2 \;,
\end{eqnarray}
where $\rho^\# = \rho$ for  even $N$, and $\rho^\# = \rho^T$, i.e., the transpose of $\rho$, for  odd $N$.
\end{thm}

\section{Stabilizer testing via convolution-swap tests on qudits and qubits}\label{sec:stab_T}
We introduce a systematic way to perform  stabilizer testing for states, and Clifford testing for gates. 
These tests are based on quantum convolutions and swap tests.
The key idea of our tests  is  the invariance  of purity under quantum convolution, a property true both for  stabilizer states and for Clifford unitaries. (See Propositions \ref{prop:pure_stab} and \ref{prop:pure_Clif}.)

Swap tests are widely-used to compare quantum states. 
Given two input states $\rho$, $\sigma$ and another ancilla state $\proj{0}$, the test is performed by applying a Hadamard gate on the ancilla qubit, and then a controlled-SWAP gate, controlled by the
ancilla qubit; finally one applies a Hadmard gate on the ancilla qubit. One  takes a  measurement on the ancilla
qubit in the computational basis. The probability of getting an outcome "0" is
\begin{eqnarray}
\text{Pr}[0]=\frac{1}{2}
\left[
1+\Tr{\rho\sigma}
\right].
\end{eqnarray}
If $\rho=\sigma$, then 
output probability 
is equal to $\left(1+\Tr{\rho^2}\right)/2$, which 
is what we need in our protocol.

\subsection{Stabilizer testing for states \label{Sec:StabilizerTestStates}}

Now, let us start with  the stabilizer testing for $n$-qubit states by using the convolution
$\boxtimes_3$.  We will discuss the qudit case later.
Based on purity invariance of stabilizer states under quantum convolution, we propose the following  stabilizer testing protocol for qubits.

\begin{center}
\begin{tcolorbox}[width=10cm,height=3.5cm, title=Protocol 1: Stabilizer Test for $n$-qubit states]

\begin{enumerate}
    \item 
 Prepare 6 copies of $\psi$, and perform the convolution 
 for each 3 copies of $\psi$, and get 2 copies of $\boxtimes_3\psi$;
 
\item
Perform the swap test for the 2 copies of $\boxtimes_3\psi$.
If the output is $0$, it passes the test; otherwise, it fails.
\end{enumerate}
\end{tcolorbox}
\end{center}

The probability of acceptance is 
\begin{eqnarray}\label{eq:prop_acc_2}
\text{Pr}_{\text{accep}}[\psi]
=\frac{1}{2}
\left[1+\Tr{(\boxtimes_3\psi)^2}
\right].
\end{eqnarray}

\begin{thm}\label{thm:stab_test_qub}
Given an $n$-qubit pure state $\psi$, let 
$\max_{\phi\in STAB}|\iinner{\psi}{\phi}|^2=1-\epsilon$. Then the probability of acceptance
is bounded by
\begin{eqnarray}
\frac{1}{2}\left[1+(1-\epsilon)^6\right]\leq \mathrm{Pr}_{\mathrm{accep}}(\proj{\psi})\leq 1-3\epsilon+O(\epsilon^2).
\end{eqnarray}
\end{thm}
\begin{proof}
For the qubit case, 
we need to prove that 
\begin{eqnarray*}
(1-\epsilon)^6\leq \Tr{(\boxtimes_3\psi)^2}
\leq 1-6\epsilon+O(\epsilon^2).
\end{eqnarray*}
Since $1-\epsilon=\max_{\phi\in STAB}|\iinner{\psi}{\phi}|^2$, then there exists a
Clifford unitary $U$ such that 
\begin{eqnarray*}
|\iinner{\vec 0}{U\psi}|^2=1-\epsilon,
\end{eqnarray*}
Let us take $\ket{\varphi}=U\ket{\psi}$, then  $\varphi(\vec 0)=1-\epsilon$ with $\varphi(\vec x)=|\iinner{\vec x}{\varphi}|^2$. 
Again we have $\varphi(\vec x)\le \delta=\min\{ \epsilon, 1-\epsilon  \}$ for every $\vec x$.
By the
commutativity of $\boxtimes_3$ with Clifford unitaries in Theorem~\ref{thm:main_conv},  we have
$
\Tr{(\boxtimes_3\psi)^2}
=\Tr{(\boxtimes_3\varphi)^2}
$. 

For the lower bound, we have
\begin{eqnarray*}
\Tr{(\boxtimes_3\varphi)^2}
\geq \bra{\vec 0}\boxtimes_3\varphi \ket{\vec 0}^2
=\left(\sum_{\vec x,\vec y}\varphi(\vec x)\varphi(\vec y)\varphi(\vec x+\vec y)\right)^2
\geq  \left(\varphi(\vec 0)\varphi(\vec 0)\varphi(\vec 0)\right)^2
\geq (1-\epsilon)^6,
\end{eqnarray*}
where the first inequality comes from the Cauchy-Schwarz inequality, and the second equality is because
\begin{eqnarray*}
\bra{\vec 0}\boxtimes_3\varphi \ket{\vec 0}
&=&\Tr{\boxtimes_3(\varphi\ot\varphi\ot\varphi )\proj{\vec 0}}\\
&=&\Tr{\varphi\ot\varphi\ot\varphi 
\boxtimes^\dag_3\left(\proj{\vec 0}\right)}\\
&=&\sum_{\vec x, \vec y}
\Tr{\varphi\ot\varphi\ot\varphi \proj{\vec x+\vec y}\ot\proj{\vec y}\ot\proj{\vec x}}\\
&=&\sum_{\vec x,\vec y}\varphi(\vec x)\varphi(\vec y)\varphi(\vec x+\vec y).
\end{eqnarray*}

For the upper bound, since 
$|\iinner{\vec 0}{\varphi}|^2=1-\epsilon$, then 
$\ket{\varphi}$ can be written as 
$\ket{\varphi}=\sqrt{1-\epsilon}\ket{\vec 0}+\sum_{\vec x\neq \vec 0}\varphi_{\vec x}\ket{\vec x}$,
where $\varphi_{\vec x}=\iinner{\vec x}{\varphi}$, $\varphi(\vec x)=|\varphi_{\vec x}|^2$ and 
$\sum_{\vec x\neq \vec 0}\varphi(\vec x)=\epsilon$. 
Let us take the stabilizer group $G$ of $\ket{\vec 0}$, which is
$G=\set{(\vec p, \vec 0):\vec p\in\mathbb{Z}^n_2}$. 
Hence, $(\vec p, \vec q)\notin G$ iff  $\vec q\neq \vec 0$.
And for any $\vec q\neq \vec 0$, we have 
\begin{eqnarray*}
|\Xi_{\varphi}(\vec p, \vec q)|
\leq \sum_{\vec x}|\varphi_{\vec x}||\varphi_{\vec x+\vec q}|
=2|\varphi_{\vec 0}||\varphi_{\vec q}|
+\sum_{\vec x\neq \vec 0, \vec q}
|\varphi_{\vec x}||\varphi_{\vec x+\vec q}|
\leq 2\sqrt{(1-\epsilon)\epsilon}
+\epsilon.
\end{eqnarray*}
That is, $\max_{(\vec p, \vec q)\notin G}| \Xi_{\varphi}(\vec p, \vec q)|\leq 2\sqrt{(1-\epsilon)\epsilon}
+\epsilon$.
Moreover, 
\begin{eqnarray*}
1=\Tr{\varphi^2}
=\frac{1}{2^n}
\sum_{(\vec p, \vec q)\in V^n}
|\Xi_{\varphi}(\vec p, \vec q)|^2
=\frac{1}{2^n}
\sum_{(\vec p, \vec q)\in G}
|\Xi_{\varphi}(\vec p, \vec q)|^2
+\frac{1}{2^n}
\sum_{(\vec p, \vec q)\notin G}
|\Xi_{\varphi}(\vec p, \vec q)|^2.
\end{eqnarray*}
Since
\begin{eqnarray*}
\frac{1}{2^n}
\sum_{(\vec p, \vec q)\in G}
|\Xi_{\varphi}(\vec p, \vec q)|^2
=\frac{1}{2^n}
\sum_{\vec p\in \mathbb{Z}^n_d}
\left|(1-\epsilon)+\sum_{\vec x\neq\vec 0}\varphi(\vec x)w^{\inner{\vec x}{\vec p}}_d\right|^2
=(1-\epsilon)^2+\sum_{\vec x\neq \vec 0}
\varphi(\vec x)^2,
\end{eqnarray*}
we have
\begin{eqnarray*}
\frac{1}{2^n}
\sum_{(\vec p, \vec q)\notin G}
|\Xi_{\varphi}(\vec p, \vec q)|^2
\leq 1-(1-\epsilon)^2.
\end{eqnarray*}
Now,
\begin{eqnarray*}
\Tr{(\boxtimes_3\varphi)^2}
=\frac{1}{2^n}
\sum_{(\vec p, \vec q)\in V^n}
|\Xi_{\varphi}(\vec p, \vec q)|^6
=\frac{1}{2^n}
\sum_{(\vec p, \vec q)\in G}
|\Xi_{\varphi}(\vec p, \vec q)|^6
+\frac{1}{2^n}
\sum_{(\vec p, \vec q)\notin G}
|\Xi_{\varphi}(\vec p, \vec q)|^6,
\end{eqnarray*}
where 
\begin{eqnarray*}
\frac{1}{2^n}
\sum_{(\vec p, \vec q)\in G}
|\Xi_{\varphi}(\vec p, \vec q)|^6
\leq 1-6\epsilon(1-\epsilon)^5,
\end{eqnarray*}
and 
\begin{eqnarray*}
\frac{1}{2^n}
\sum_{(\vec p, \vec q)\notin G}
|\Xi_{\varphi}(\vec p, \vec q)|^6
\leq \left[2\sqrt{(1-\epsilon)\epsilon}
+\epsilon\right]^4
\frac{1}{2^n}
\sum_{(\vec p, \vec q)\notin G}
|\Xi_{\varphi}(\vec p, \vec q)|^2
\leq \left[2\sqrt{(1-\epsilon)\epsilon}
+\epsilon\right]^4\cdot
\left[ 1-(1-\epsilon)^2\right]=32\epsilon^3+o(\epsilon^3).
\end{eqnarray*}
Hence 
\begin{eqnarray*}
\Tr{(\boxtimes_3\varphi)^2}
\leq 1-6\epsilon+30\epsilon^2+o(\epsilon^2)
=1-6\epsilon+O(\epsilon^2).
\end{eqnarray*}

\end{proof}

Besides the qubit case, let us also consider the stabilizer testing for states in $n$-qudit systems with $d$ an odd prime.   For qudit testing, we introduce the Hadamard convolution $\boxtimes_H$
of two $n$-qudit states. Note that, we can also use other convolutions proposed in~\cite{BGJ23b} to implement the stabilizer testing if the local dimension $d$ satisfies some special requirement. For example, discrete beam splitter convolution can be applied to perform stabilizer testing in case $d\geq 7$.

\begin{Def}[\bf Hadamard Convolution, \cite{BGJ23b}]\label{Def:Had_Conv}
The Hadamard convolution  of two $n$-qudit states $\rho$ and $\sigma$ is
\begin{eqnarray}
\rho\boxtimes_H\sigma=\Ptr{B}{V_H\rho\ot\sigma V^\dag_H},
\end{eqnarray}
where  $V_H=U^{\ot n}_H:=U^{(1,n+1)}_H\ot U^{(2,n+2)}_H\ot...\ot U^{(n,2n)}_H$, where  
$U_H$ is the $2$-qudit unitary
\begin{eqnarray}
U_H=\sum_{x,y\in\mathbb{Z}_d}
\ket{x}\bra{ x+ y}
\ot \ket{y}\bra{ x- y}\;,
\end{eqnarray}
and where $U^{(i,n+j)}_H$ denotes the action of $U_H$ on the $i^{\rm th}$ qudit of the first state and the $j^{\rm th}$ qudit of the second one. 
The corresponding convolutional channel $\mathcal{E}_H$ is 
\begin{eqnarray}
\mathcal{E}_H(\ \cdot\ )
=\Ptr{B}{V_H \ \cdot\  V^\dag_H}\;.
\end{eqnarray}
\end{Def}

\begin{lem}[\cite{BGJ23b}]\label{lem:prop_HCon}
The Hadamard convolution  for any odd, prime d, satisfies the following properties: 

(1) \textbf{Hadamard convolution is abelian:} $\rho\boxtimes_H\sigma=\sigma\boxtimes_H\rho$, for any $n$-qudit states $\rho$ and $\sigma$.

(2) \textbf{Convolution-multiplication duality:} $\Xi_{ \rho \boxtimes_H \sigma} (\vec p, \vec q) = \Xi_\rho (2^{-1}\vec p, \vec q) \Xi_\sigma (2^{-1}\vec p, \vec q)$,
 $\forall (\vec p, \vec q)\in V^n$.

(3) \textbf{Purity invariance of stabilizer states:} given a pure state $\psi$, the self-convolution $\psi\boxtimes_H\psi$   is pure iff 
$\psi$ is a stabilizer state.

(4) \textbf{Commutativity with Clifford unitaries:} 
for any Clifford unitary $U$, there exists some Clifford unitary $U_1$ such that 
\begin{eqnarray}\label{230520shi1-2}
U_1\left(\rho\boxtimes_{H}\sigma\right) U_1^\dag
=(U\rho U^\dag) \boxtimes_H(U\sigma U^\dag) \;,\quad \forall \rho,\sigma\;.
\end{eqnarray}

(5) \textbf{Mean state properties:}
Let $\rho$ and $\sigma$ be two $n$-qudit states with the same mean state $\CMM(\rho) = \CMM(\sigma)$. 
Then  
\begin{eqnarray}\label{0204shi6-2}
\CMM(\rho\boxtimes_H\sigma)=\CMM(\rho)\boxtimes_H\sigma
=\rho\boxtimes_H\CMM(\sigma)=\CMM(\rho)\boxtimes_H\CMM(\sigma)\;.
\end{eqnarray} 
\end{lem}

 Based on purity invariance of stabilizer states under quantum convolution, we propose the following  stabilizer testing protocol for qudits.

\begin{center}
\begin{tcolorbox}
[width=10cm,height=3.5cm,title=Protocol \texttt{2}: Stabilizer Test for $n$-qudit states \label{prop:test2}]

\begin{enumerate}
    \item 

 Prepare 4 copies of the state $\psi$, and perform the convolution 
 to obtain 2 copies of  $\psi\boxtimes_H \psi$;

\item
Perform the swap test for the 2 copies of $\psi\boxtimes_H \psi$.
If the output is $0$, it passes the test; otherwise, it fails.
\end{enumerate}
\end{tcolorbox}
\end{center}

The probability of acceptance is 
\begin{eqnarray}
\text{Pr}_{\text{accep}}[\psi]
=\frac{1}{2}
\left[1+\Tr{(\psi\boxtimes_H\psi)^2}
\right].
\end{eqnarray}

\begin{thm}\label{230613thm1}
Given an $n$-qudit pure state $\psi$ with $d$ an odd prime, let 
$\max_{\phi\in STAB}|\iinner{\psi}{\phi}|^2=1-\epsilon$. Then the probability of acceptance
is bounded by:
\begin{eqnarray}
\frac{1}{2}\left[1+(1-\epsilon)^4\right]\leq \mathrm{Pr}_{\mathrm{accep}}[\psi]
\leq 1-2\epsilon+O(\epsilon^2).
\end{eqnarray}
\end{thm}
\begin{proof}
We only need to prove that 
\begin{eqnarray*}
(1-\epsilon)^4\leq \Tr{(\psi\boxtimes_H\psi)^2}
\leq 1-4\epsilon+O(\epsilon^2).
\end{eqnarray*}
Since $1-\epsilon=\max_{\phi\in STAB}|\iinner{\psi}{\phi}|^2$,  there exists a 
Clifford unitary $U$ such that 
\begin{eqnarray*}
|\iinner{\vec 0}{U\psi}|^2=1-\epsilon.
\end{eqnarray*}
Let us take $\ket{\varphi}=U\ket{\psi}$, then  $\varphi(\vec 0)=1-\epsilon$ with $\varphi(\vec x)=|\iinner{\vec x}{\varphi}|^2$. By the
commutativity of $\boxtimes_H$ with Clifford unitaries in Lemma~\ref{lem:prop_HCon},  we have
$
\Tr{(\psi\boxtimes_H\psi)^2}
=\Tr{(\varphi\boxtimes_H\varphi)^2}
$. 

For the lower bound, we have
\begin{eqnarray*}
\Tr{(\varphi\boxtimes_H\varphi)^2}
\geq \bra{\vec 0}\varphi\boxtimes_H\varphi \ket{\vec 0}^2
= \left(\sum_{\vec y}\varphi(\vec y)\varphi(-\vec y)\right)^2
\geq  \left(\varphi(\vec 0)\varphi(\vec 0)\right)^2
\geq (1-\epsilon)^4,
\end{eqnarray*}
where the first inequality comes from the Cauchy-Schwarz inequality, and the equality is because for every $\vec x$ we have
\begin{eqnarray*}
\bra{\vec x}\varphi\boxtimes_H\varphi\ket{\vec x}
&=&\Tr{\mathcal{E}_H(\varphi\ot\varphi)\proj{\vec x}}
=\Tr{\varphi\ot\varphi \mathcal{E}^\dag_H(\proj{\vec x})}
=\sum_{\vec y}\Tr{\varphi\ot\varphi \proj{\vec x+\vec y}\ot \proj{\vec x-\vec y}}\\
&=&\sum_{\vec y}\varphi(\vec x+\vec y)\varphi(\vec x-\vec y).
\end{eqnarray*}

For the upper bound, 
since 
$|\iinner{\vec 0}{\varphi}|^2=1-\epsilon$,  
$\ket{\varphi}$ can be written as 
$$\ket{\varphi}=\sqrt{1-\epsilon}\ket{\vec 0}+\sum_{\vec x\neq \vec 0}\varphi_{\vec x}\ket{\vec x},$$
where $\varphi_{\vec x}=\iinner{\vec x}{\varphi}$, $\varphi(\vec x)=|\varphi_{\vec x}|^2$ and 
$\sum_{\vec x\neq \vec 0}\varphi(\vec x)=\epsilon$. 
Since $1-\epsilon=\max_{\phi\in STAB}|\iinner{\varphi}{\phi}|^2$,
we have $$\varphi(\vec x)\le \delta=\min\{ \epsilon, 1-\epsilon  \}$$ for any $\vec x$.
Let us take the stabilizer group $G$ of $\ket{\vec 0}$, which is
$G=\set{(\vec p, \vec 0):\vec p\in\mathbb{Z}^n_d}$. 
Hence $(\vec p, \vec q)\notin G$ iff  $\vec q\neq \vec 0$.
And for any $\vec q\neq \vec 0$, we have 
\begin{eqnarray*}
\left|\Xi_{\varphi}(\vec p, \vec q)\right|
&=&\left|\Tr{\proj{\varphi}Z^{-\vec p}X^{-\vec q}}\right|
=\sum_{\vec x}\left|\bra{\varphi}Z^{-\vec p}X^{-\vec q}\proj{\vec x} \ket{\varphi}\right|\\
&\leq& \sum_{\vec x}|\varphi_{\vec x}||\varphi_{\vec x+\vec q}|
=|\varphi_{\vec 0}||\varphi_{\vec q}|+|\varphi_{\vec 0}||\varphi_{-\vec q}|
+\sum_{\vec x\neq \vec 0, -\vec q}
|\varphi_{\vec x}||\varphi_{\vec x+\vec q}|\\
&\leq& 2\sqrt{(1-\epsilon)\epsilon}
+\epsilon.
\end{eqnarray*}
That is, $$\max_{(\vec p,\vec q)\notin G}|\Xi_{\varphi}(\vec p,\vec q)|\leq 2\sqrt{(1-\epsilon)\epsilon}
+\epsilon.$$
Moreover, 
\begin{eqnarray*}
1=\Tr{\varphi^2}
=\frac{1}{d^n}
\sum_{(\vec p,\vec q)\in V^n}
\left|\Xi_{\varphi}(\vec p,\vec q)\right|^2
=\frac{1}{d^n}
\sum_{(\vec p,\vec q)\in G}
\left|\Xi_{\varphi}(\vec p,\vec q)\right|^2
+\frac{1}{d^n}
\sum_{(\vec p,\vec q)\notin G}
\left|\Xi_{\varphi}(\vec p,\vec q)\right|^2.
\end{eqnarray*}
Since
\begin{eqnarray*}
\frac{1}{d^n}
\sum_{(\vec p,\vec q)\in G}
\left|\Xi_{\varphi}(\vec p,\vec q)\right|^2
=\frac{1}{d^n}
\sum_{\vec p\in \mathbb{Z}^n_d}
\left|(1-\epsilon)+ \sum_{\vec x\neq\vec 0}\varphi(\vec x)w^{\inner{\vec x}{\vec p}}_d\right|^2
=(1-\epsilon)^2+\sum_{\vec x\neq \vec 0}
\varphi(\vec x)^2,
\end{eqnarray*}
we have
\begin{eqnarray*}
\frac{1}{d^n}
\sum_{(\vec p,\vec q)\notin G}
\left|\Xi_{\varphi}(\vec p,\vec q)\right|^2
\leq 1-(1-\epsilon)^2.
\end{eqnarray*}
Then 
\begin{eqnarray*}
\Tr{(\varphi\boxtimes_H\varphi)^2}
=\frac{1}{d^n}\sum_{\vec p, \vec q}
|\Xi_{\varphi\boxtimes_H\varphi}(\vec p, \vec q)|^2
=\frac{1}{d^n}\sum_{\vec p, \vec q}
|\Xi_{\varphi}(2^{-1}\vec p, \vec q)|^4
=\frac{1}{d^n}
\sum_{(\vec p, \vec q)\in G}
\left|\Xi_{\varphi}(\vec p, \vec q)\right|^4
+\frac{1}{d^n}
\sum_{(\vec p, \vec q)\notin G}
\left|\Xi_{\varphi}(\vec p, \vec q)\right|^4.
\end{eqnarray*}
For the first part, 
\begin{eqnarray*}
&&\frac{1}{d^n}
\sum_{(\vec p, \vec q)\in G}
\left|\Xi_{\varphi}(\vec p, \vec q)\right|^4\\
&=&\frac{1}{d^n}
\sum_{\vec p\in \mathbb{Z}^n_d}
\left|(1-\epsilon)+\sum_{\vec x\neq\vec 0}\varphi(\vec x)w^{\inner{\vec x}{\vec p}}_d\right|^4\\
&= &(1-\epsilon)^4
+4(1-\epsilon)^3\mathbb{E}_{\vec p\in \mathbb{Z}^n_d}\left(\sum_{\vec x\neq\vec 0}\varphi(\vec x)w^{\inner{\vec x}{\vec p}}_d\right) 
\\
&&+4(1-\epsilon)^2\mathbb{E}_{\vec p\in \mathbb{Z}^n_d}\left|\sum_{\vec x\neq\vec 0}\varphi(\vec x)w^{\inner{\vec x}{\vec p}}_d\right|^2
+ (1-\epsilon)^2\mathbb{E}_{\vec p\in \mathbb{Z}^n_d}\left(\sum_{\vec x\neq\vec 0}\varphi(\vec x)w^{\inner{\vec x}{\vec p}}_d\right)^2 +(1-\epsilon)^2\mathbb{E}_{\vec p\in \mathbb{Z}^n_d}\left(\sum_{\vec x\neq\vec 0}\varphi(\vec x)w^{\inner{-\vec x}{\vec p}}_d\right)^2
\\
&&+2(1-\epsilon) \mathbb{E}_{\vec p\in \mathbb{Z}^n_d} \left(\sum_{\vec x\neq\vec 0}\varphi(\vec x)w^{\inner{\vec x}{\vec p}}_d\right)^2 \left(\sum_{\vec x\neq\vec 0}\varphi(\vec x)w^{\inner{-\vec x}{\vec p}}_d\right)
+
2(1-\epsilon) \mathbb{E}_{\vec p\in \mathbb{Z}^n_d} \left(\sum_{\vec x\neq\vec 0}\varphi(\vec x)w^{\inner{\vec x}{\vec p}}_d\right) \left(\sum_{\vec x\neq\vec 0}\varphi(\vec x)w^{\inner{-\vec x}{\vec p}}_d\right)^2
\\
&&+ 
\mathbb{E}_{\vec p\in \mathbb{Z}^n_d} \left|\sum_{\vec y\neq\vec 0}\varphi(\vec x)w^{\inner{\vec x}{\vec p}}_d\right|^4
\\
&\le& (1-\epsilon)^4+0
+6(1-\epsilon)^2 \delta\epsilon
+4(1-\epsilon) \delta \epsilon^2 +\delta\epsilon^3\\
&\le&  1-4\epsilon(1-\epsilon)^3 ,
 \end{eqnarray*}
where the second to the last inequality comes from the fact that 
$ \sum_{\vec p\in \mathbb{Z}^n_d}\left(\sum_{\vec x\neq\vec 0}\varphi(\vec x)w^{\inner{\vec x}{\vec p}}_d\right)=0$, $\left|\sum_{\vec x\neq\vec 0}\varphi(\vec x)w^{\inner{\vec x}{\vec p}}_d\right|\leq \epsilon$,
and
\begin{eqnarray*}
\mathbb{E}_{\vec p\in \mathbb{Z}^n_d}\left|\sum_{\vec x\neq\vec 0}\varphi(\vec x)w^{\inner{\vec x}{\vec p}}_d\right|^2 = 
\mathbb{E}_{\vec p\in \mathbb{Z}^n_d} \sum_{\vec x,\vec y\neq\vec 0}\varphi(\vec x) \varphi(\vec y)w^{\inner{\vec x-\vec y}{\vec p}}_d  =  \sum_{\vec x \neq\vec 0}\varphi(\vec x)^2 \le \delta \sum_{\vec x \neq\vec 0}\varphi(\vec x) \le \delta \epsilon,
\end{eqnarray*}
and the similar bounds on higher order terms.

For the second part,
\begin{eqnarray*}
\frac{1}{d^n}
\sum_{(\vec p, \vec q)\notin G}
\left|\Xi_{\varphi}(\vec p, \vec q)\right|^4
\leq \left[2\sqrt{(1-\epsilon)\epsilon}
+\epsilon\right]^2
\frac{1}{d^n}
\sum_{(\vec p, \vec q)\notin G}
\left|\Xi_{\varphi}(\vec p, \vec q)\right|^2
\leq \left[2\sqrt{(1-\epsilon)\epsilon}
+\epsilon\right]^2\cdot
\left[ 1-(1-\epsilon)^2\right]=8\epsilon^2+o(\epsilon^2).
\end{eqnarray*}
Hence 
\begin{eqnarray*}
\Tr{(\varphi\boxtimes_H\varphi)^2}
\leq 1-4\epsilon+20\epsilon^2+o(\epsilon^2)
=1-4\epsilon+O(\epsilon^2).
\end{eqnarray*}

\end{proof}
 The bound obtained in our analysis is almost optimal, as the lower bound  and upper bound
 are both equal to $1-2\epsilon$ up to some higher-order terms of $\epsilon$.

\begin{Rem}
Note that the upper bound on the probability of acceptance in the Bell difference sampling is dimension-dependent, given by $1-\frac{\epsilon}{8d^2}$ for any odd prime $d$~\cite{Gross21}. However, using the method proposed in our work, we can also provide a lower bound of $1-2\epsilon+O(\epsilon^2)$ and a dimension-independent upper bound  of
$1-2\epsilon+O(\epsilon^2)$
on the  probability of acceptance  in Bell difference sampling.
 This improved upper bound is dimension-independent and close to the lower bound up to some higher-order terms of $\epsilon$.

\end{Rem}

\subsection{Clifford testing for gates}
To test whether the given unitary is a Clifford  or is $\epsilon$ far away  from the Clifford group, 
we can generate the Choi state $J_U$ of the unitary and apply the 
stabilizer testing for the Choi state.
Hence, the  probability of acceptance is 
\begin{eqnarray}
\text{Pr}_{\text{accep}}[U]=
\Tr{(J_U\boxtimes_H J_U)^2},
\end{eqnarray}
for any odd prime d, and 
\begin{eqnarray}
\text{Pr}_{\text{accep}}[U]=
\Tr{(\boxtimes_3 J_U)^2},
\end{eqnarray}
for $d=2$.

Here, let us consider the maximal overlap
between a given unitary $U$ and the set of Clifford unitaries,
$
\max_{V\in Cl_n}
|\inner{V}{U}|^2
$,
where $\inner{V}{U}=\frac{1}{d^n}\Tr{V^\dag U}$.
Then we have the following results on the success probability 
on the testing for both qubits and qudits.

\begin{thm}\label{thm:Clif_test_2}
Given an $n$-qubit gate $U$, let 
$\max_{V\in Cl_n}|\inner{V}{U}|^2=1-\epsilon$. Then the  probability of acceptance
is bounded by:
\begin{eqnarray}
\frac{1}{2}\left[1+(1-\epsilon)^6\right]\leq \text{Pr}_{\text{accep}}[U]\leq 
1-3\epsilon+O(\epsilon^2).
\end{eqnarray}
\end{thm}

\begin{thm}\label{thm:Clif_test_d}
Given an $n$-qudit gate $U$ with $d$ an odd prime number, let 
$\max_{V\in Cl_n}|\inner{V}{U}|^2=1-\epsilon$. Then the probability of acceptance
is bounded by:
\begin{eqnarray}
\frac{1}{2}\left[1+(1-\epsilon)^4\right]\leq \text{Pr}_{\text{accep}}[U]
\leq 1-2\epsilon+O(\epsilon^2).
\end{eqnarray}
\end{thm}

The proof of Theorems \ref{thm:Clif_test_2}  and \ref{thm:Clif_test_d}  follows a similar approach as the state testing, but with some 
additional technique lemmas.  For the sake of completeness, we present the detailed proof in Appendix~\ref{appen:stab_test}.

\begin{Rem}
We briefly discuss the sample complexity and circuit depth that one would require to realize the convolution-swap test for qubit systems. The qudit case can be analyzed similarly.

Based on the results in the swap tests~\cite{barenco1997stabilization,BuhrmanPRL01,de2019quantum},  
estimating the success probability in \eqref{eq:prop_acc_2} within additive error $\varepsilon$ requires $O(1/\varepsilon^2)$ copies of the state $\boxtimes_3\psi$. For each state $\boxtimes_3\psi$, we need $3$ copies of the input state $\psi$. Hence,  estimation within additive error $\varepsilon$ requires a total number of state copies  $O(1/\varepsilon^2)$.

The circuit depth to realize the swap test is $O(1)$,  if one uses  parallelized controlled-SWAP gates.
Similarly,   the quantum convolution $\boxtimes_3$ also achieves
$O(1)$ circuit depth by employing parallelized CNOT gates, as detailed in Definition \ref{def:key_U}.
Hence, for each run of the convolution-swap test, the circuit complexity is $O(1)$ for using the CNOT gates and controlled-SWAP gates.

Finally, both the experiments for using swap-tests to 
measure entanglement entropy by Greiner and his group~\cite{Islam15},  and the high-fidelity realization of entangling gates by the Harvard-MIT-QuEra collaboration~\cite{evered2023highfidelity} suggest promising avenues for the potential experimental realization of our convolution-swap test. This could be a compelling direction for future work.
\end{Rem}

\section{Magic entropy on qudits and qubits}\label{sec:mag_en}
We introduce "magic entropy" as a measure of magic based on quantum convolutions.

\begin{Def}[\bf Magic entropy on qubits]\label{Def:mag_ent}
Given an $n$-qubit pure state $\psi$, the magic entropy is 
\begin{eqnarray}
\text{ME}(\psi)
=S\left(\boxtimes_{3}\psi\right).
\end{eqnarray}
\end{Def}

\begin{Def}[\bf Higher-order, R\'enyi magic entropy on qubits]
Given an $n$-qubit pure state $\psi$ and an integer $N\geq 1$ and $\alpha\in [0, +\infty]$, the order-$N$, $\alpha$-R\'enyi magic entropy is 
\begin{eqnarray}
\text{ME}^{(N)}_{\alpha}(\psi)
=S_{\alpha}\left(\boxtimes_{2N+1}\psi\right).
\end{eqnarray}
\end{Def}
For the case of $N=1$ and $\alpha=1$, it reduces to magic entropy in Definition~\ref{Def:mag_ent}.

\begin{prop}\label{prop:magi_ent}
The order-$N$, $\alpha$-R\'enyi magic entropy $\text{ME}^{(N)}_{\alpha}(\psi)$
satisfies:

(1) $0\leq \text{ME}^{(N)}_{\alpha}(\psi)\leq n$, and $\text{ME}^{(N)}_{\alpha}(\psi)=0$ iff $\psi$ is a stabilizer state.

(2) The $\text{ME}^{(N)}_{\alpha}$ is invariant under a Clifford unitary acting on states.

(3) $ME^{(N)}_{\alpha}(\psi)\leq ME^{(N+1)}_{\alpha}(\psi)$ for any $\alpha\in [0,\infty]$, and integer $N\geq 1$.

(4) For any integer $N\geq 1$, 
$
ME^{(N)}_{\alpha}(\psi)\geq ME^{(N)}_{\beta}(\psi), \forall \alpha\leq \beta.
$

(5) $\text{ME}^{(N)}_{\alpha}(\psi_1\ot \psi_2)=\text{ME}^{(N)}_{\alpha}(\psi_1)+\text{ME}^{(N)}_{\alpha}(\psi_2)$.
\end{prop}
\begin{proof}

(1)  $0\leq \text{ME}^{(N)}_{\alpha}(\psi)\leq n$ comes directly from the definition. $\text{ME}^{(N)}_{\alpha}(\psi)=0$ iff $\boxtimes_3\psi$ is a pure state,
iff $\psi$ is stabilizer state by the purity invariance of stabilizer states in Theorem~\ref{thm:main_conv}. 

(2) comes directly the commutativity of convolution with Clifford unitaries in Theorem~\ref{thm:main_conv}.

(3) comes directly from  the entropy inequality~\eqref{ineq:entopy_con} in Theorem~\ref{thm:main_conv}.

(4) comes directly from the monotonicity of quantum R\'enyi entropy.

(5) holds because 
$
\boxtimes_3(\psi_1\ot \psi_2,\psi_1\ot \psi_2,\psi_1\ot \psi_2) 
=\boxtimes_3(\psi_1\ot\psi_1\ot\psi_1)
\ot \boxtimes_3(\psi_2\ot\psi_2\ot\psi_2).
$
\end{proof}

\begin{Rem}
Similar to quantum R\'enyi entropy,
 we can also 
use the quantum $\alpha$-Tsallis entropy $$T_{\alpha}(\rho)=\frac{\Tr{\rho^\alpha}-1}{1-\alpha}$$
(or other Schur concave functions\footnote{A function $f:\real^n_+\to \real$ is 
Schur concave if $\vec p\prec \vec q$ implies that $f(\vec p)\geq f(\vec q)$.}) to quantify magic. For example, 
we define the order-$N$ $\alpha$-Tsallis magic entropy as 
$MT^{(N)}_{\alpha}(\psi)=T_{\alpha}(\boxtimes_{2N+1}\psi)$. It is easy to verify that 
it also satisfies $(1)-(4)$ for $\alpha\geq 1$ in Proposition~\ref{prop:magi_ent}.
Moreover, for $N=1$ and $\alpha=2$, the corresponding Tsallis magic entropy
$MT^{(1)}_2(\psi)$
is equivalent to $2(1-\text{Pr}_{\text{accep}}[\psi])$, where
$\text{Pr}_{\text{accep}}[\psi]$ is the probability 
of acceptance in the stabilizer testing  as given
in \eqref{eq:prop_acc_2}.

\end{Rem}

\begin{Exam}
Let us consider the $T$ state $\ket{T}=T\ket{+}=\frac{1}{\sqrt{2}}(\ket{0}+e^{i\pi/4}\ket{1})$.
Then the magic entropy of $\ket{T}$ is
\begin{eqnarray}
ME(\proj{T})=h(1/4),
\end{eqnarray}
where $h(x)=-x\log_2x-(1-x)\log_2(1-x)$ is the binary entropy.
In general, the order-$N$, $\alpha$-R\'enyi magic entropy is
\begin{eqnarray*}
ME^{(N)}_{\alpha}(\proj{T})=h_{\alpha}\left(\frac{1}{2}\left(1-2^{-N}\right)\right),
\end{eqnarray*}
where $h_{\alpha}(x)=\frac{1}{1-\alpha}\log[x^{\alpha}+(1-x)^{\alpha}]$ is the binary $\alpha$-R\'enyi entropy.
Hence, for $m$ copies of the $T$ state
\begin{eqnarray*}
ME^{(N)}_{\alpha}(\proj{T}^{\ot m})
=mh_{\alpha}\left(\frac{1}{2}\left(1-2^{-N}\right)\right).
\end{eqnarray*}
\end{Exam}

\begin{Exam}
Let us consider another magic state $\ket{H}$ with 
$\proj{H}=\frac{1}{2}\left(I+\frac{1}{\sqrt{3}}X+\frac{1}{\sqrt{3}}Y+\frac{1}{\sqrt{3}}Z\right)$. 
Then the magic entropy of $\ket{H}$ is 
\begin{eqnarray}
ME(\proj{H})=h(1/3).
\end{eqnarray}
In general, the order-$N$, $\alpha$-R\'enyi magic entropy is
\begin{eqnarray*}
ME^{(N)}_{\alpha}(\proj{H})=h_{\alpha}\left(\frac{1}{2}\left(1-3^{-N}\right)\right).
\end{eqnarray*}
Hence, for $m$ copies of the $H$ state
\begin{eqnarray*}
ME^{(N)}_{\alpha}(\proj{H}^{\ot m})
=mh_{\alpha}\left(\frac{1}{2}\left(1-3^{-N}\right)\right).
\end{eqnarray*}

\end{Exam}

\begin{Rem}
Note that in the theory of entanglement, the entanglement entropy depends on the spectrum of the reduced state, which is also known as the entanglement spectrum. Similarly, we can introduce the concept of the magic spectrum, which is defined as the spectrum of the self-convolution $\boxtimes\psi$. For instance, the magic spectrum of the state $\ket{T}$ is given by $(3/4, 1/4)$.  Further exploration of the properties and applications of the magic spectrum is left for future work.
\end{Rem}

\begin{prop}
Given an  $n$-qubit input state $\ket{\psi}$ and a quantum circuit $C_t$ consisting of Clifford unitaries and $t$ 1-qubit non-Clifford gates, 
the magic entropy of the output state $C_t\ket{\psi}$ satisfies
\begin{eqnarray}
ME\left(C_t\proj{\psi} C^\dag_t\right)\leq ME\left(\proj{\psi}\right)+2t.
\end{eqnarray}
\end{prop}
\begin{proof}
Since magic entropy is invariant under the action of a Clifford unitary, we only need to prove the effect on entropy of the action of a single  1-qubit,
non-Clifford gate $g$. Without loss
of generality, let us assume that the 1-qubit gate $g$ acts on the first qubit, denoted as $g_1$.
We only need to show
\begin{eqnarray*}
ME(g_1\proj{\psi} g^\dag_1)\leq
ME(\proj{\psi})+2.
\end{eqnarray*}
Since the reduced states of $\boxtimes_3g_1\proj{\psi} g^\dag_1$ and 
$\proj{\psi}$ are the same on the $(n-1)$-qubit system, then 
$S\left(\Ptr{1}{\boxtimes_3g_1\proj{\psi} g^\dag_1}\right)=S\left(\Ptr{1}{\boxtimes_3\proj{\psi}}\right)$. 
By the Araki-Lieb triangle inequality~\cite{araki1970entropy}, i.e., $S(\rho_{AB})\geq |S(\rho_A)-S(\rho_B)|$ and subadditivity of entropy, i.e., $S(\rho_{AB})\leq S(\rho_A)+S(\rho_B)$ , we have 
\begin{eqnarray*}
\left|S\left(\boxtimes_3g_1\proj{\psi} g^\dag_1\right)-S\left(\Ptr{1}{\boxtimes_3g_1\proj{\psi} g^\dag_1}\right)\right|
\leq S\left(\Ptr{1^c}{\boxtimes_3g_1\proj{\psi} g^\dag_1}\right)
\leq 1,
\end{eqnarray*}
and
\begin{eqnarray*}
\left|S\left(\boxtimes_3\proj{\psi} \right)-S\left(\Ptr{1}{\boxtimes_3\proj{\psi} }\right)\right|
\leq S\left(\Ptr{1^c}{\boxtimes_3\proj{\psi} }\right)
\leq 1.
\end{eqnarray*}
Hence, we have
$\left|S\left(\boxtimes_3g_1\proj{\psi} g^\dag_1\right)-S\left(\boxtimes_3\proj{\psi} \right)\right|\leq 2$.
\end{proof}

The relative entropy of magic was introduced 
in \cite{Veitch14} as the $\min_{\sigma\in \text{STAB}} D(\rho||\sigma)$
where $STAB$ is the set of quantum states, which can be written as the convex combination 
of pure stabilizer states.
In this subsection, we provide a modification of (R\'enyi) relative entropy by 
replacing the set $\text{STAB}$ by the set $\text{MSPS}$ in Defintion~\ref{def:MSPS}.

\begin{Def}[\bf Quantum  R\'enyi relative entropy \cite{Hiai11,Martin13}]
Given two quantum states $\rho$ and $\sigma$, the quantum R\'enyi relative entropy of $\rho$ with respect to $\sigma$ is 
\begin{eqnarray*}
    D_{\alpha}(\rho||\sigma)=\frac{1}{\alpha-1}\log\Tr{\left(\sigma^{\frac{1-\alpha}{2\alpha}}\rho\sigma^{\frac{1-\alpha}{2\alpha}}\right)^{\alpha}} \;,
\end{eqnarray*}
where $\alpha\in[0,+\infty]$.
\end{Def}
For example, $\lim_{\alpha\to 1}D_{\alpha}(\rho||\sigma)=D(\rho||\sigma)=\Tr{\rho\log \rho}-\Tr{\rho\log \sigma}$\footnote{A nice formula for the quantum relative entropy on von Neuman algebra was obtained in \cite{araki1976relative} using modular operator.}, and 
$\lim_{\alpha\to\infty}D_{\alpha}(\rho||\sigma)=D_{\infty}(\rho||\sigma)=\min\set{\lambda:\rho\leq 2^{\lambda}\sigma}$.
Note that quantum R\'enyi entropy $D_{\alpha}$ is additive under tensor product, and is monotone under quantum channels for $\alpha\geq 1/2$ \cite{Tomamichel2015quantum}. 

\begin{Def}[Modified R\'enyi relative entropy of magic]\label{Def:MRM}
Given a quantum state $\rho$ and $\alpha\in [0,+\infty]$, the modified R\'enyi relative entropy of magic is
\begin{eqnarray}
MRM_{\alpha}(\rho)
=\min_{\sigma\in MSPS}
D_{\alpha}(\rho||\sigma).
\end{eqnarray}

\end{Def}

\begin{prop}
The modified R\'enyi relative entropy of magic
satisfies the following properties:

(1) $0\leq MRM_{\alpha}(\rho)\leq n$, $MRM_{\alpha}(\rho)=0$ iff $\rho$ is an MSPS;

(2) $MRM_{\alpha}$ is invariant under Clifford unitaries;

(3) For $\alpha=1$ or $+\infty$, $MRM_{\alpha}$ is nonincreasing under stabilizer measurement $\set{I\ot \proj{\vec x}}_{\vec x}$, that is, 
\begin{eqnarray}
\sum_{\vec x}p_{\vec x}MRM_{\alpha}(\rho_{\vec x}) 
\leq MRM_{\alpha}(\rho),
\end{eqnarray}
where $p_{\vec x}=\Tr{\rho I\ot \proj{\vec x}}$, and 
$\rho_{\vec x}=I\ot \proj{\vec x}\rho I\ot \proj{\vec x}/p_{\vec x}$.

\end{prop}

\begin{proof} 
One infers 
(1) and (2) directly from Definition~\ref{Def:MRM}. 
To prove (3), note that from Lemma \ref{lem:rel_mag} one has, for any $\alpha= 1 ~\text{or} ~+\infty$, that
$
MRM_{\alpha}(\rho)
=D_{\alpha}(\rho||\mathcal{M}(\rho))
=S_{\alpha}(\mathcal{M}(\rho))-S_{\alpha}(\rho)
$.
Hence
\begin{eqnarray*}
\sum_{\vec x}p_{\vec x}MRM_{\alpha}(\rho_{\vec x})
\leq \sum_{\vec x}
p_{\vec x}D_{\alpha}(\rho_{\vec x}||\mathcal{M}(\rho)_{\vec x})
\leq D_{\alpha}(\rho||\mathcal{M}(\rho)),
\end{eqnarray*}
where the last inequality comes from the result in \cite{VedralPRA98,DattaIEEE09}  that 
$
\sum_{\vec x}p_{\vec x}D_{\alpha}(\rho_{\vec x}||\mathcal{M}(\rho)_{\vec x})
\leq D_\alpha(\rho||\mathcal{M}(\rho))
$ when $\alpha=1$ or $ +\infty$.
\end{proof}

\begin{lem}[\cite{BGJ23a,BGJ23b}]\label{lem:rel_mag}
Given an $n$-qudit state $\rho$ for any integer $d\geq 2$ and $\alpha\in [1,+\infty]$, one has  
\begin{align*}
MRM_{\alpha}(\rho)=\min_{\sigma\in MSPS}
D_{\alpha}(\rho||\sigma)
=D_{\alpha}(\rho||\mathcal{M}(\rho))=S_{\alpha}(\mathcal{M}(\rho))-S_{\alpha}(\rho).
\end{align*}
That is, $\mathcal{M}(\rho)$ is closest MSPS to the given state $\rho$ w.r.t. R\'enyi relative entropy $D_{\alpha}$
for any $\alpha\geq 1$. 
\end{lem}

\begin{prop}
Given an $n$-qubit pure state $\psi$ and $\alpha\in [1,+\infty]$, 
\begin{eqnarray}\label{ineq:mon}
ME^{(N)}_{\alpha}(\psi)\leq MRM_{\alpha}(\psi),
\end{eqnarray}
and 
\begin{eqnarray}\label{ineq:lim}
ME^{(N)}_{\alpha}(\psi)\xrightarrow{N\to\infty} MRM_{\alpha}(\psi).
\end{eqnarray}
\end{prop}
\begin{proof}
By Lemma \ref{lem:rel_mag}, we have $0\leq D_{\alpha}(\boxtimes_{2N+1}\psi||\mathcal{M}(\boxtimes_{2N+1}\psi))
=S_{\alpha}(\mathcal{M}(\boxtimes_{2N+1}\psi))-S_{\alpha}(\boxtimes_{2N+1}\psi)
$. By Lemma \ref{lem:mean_con}, we have $S_{\alpha}(\mathcal{M}(\boxtimes_{2N+1}\psi))=S_{\alpha}(\mathcal{M}(\psi))$.
Hence, we obtain the inequality \eqref{ineq:mon}.
 The limit \eqref{ineq:lim} comes from the quantum central limit theorem in Theorem~\ref{thm:main_conv} and the continuity of quantum R\'enyi divergence for $\alpha\geq 1$ \cite{Martin13}.
 \end{proof}

\begin{prop}
Given an  input state $\rho$ and a quantum circuit $C_t$,  consisting of Clifford unitaries and $t$ magic T gates, 
the $MRM_{\alpha}(\rho)$ with $\alpha\geq 1$ of the output state $C_t\rho C^\dag_t$ satisfies, 
\begin{eqnarray}
MRM_{\alpha}(C_t\rho C^\dag_t)\leq MRM_{\alpha}(\rho)+t.
\end{eqnarray}
\end{prop}
\begin{proof}
Without loss of generality, we may assume the single-qubit $T$ gate acts on the first qubit, i.e., $T_1$.
We only need to show
\begin{eqnarray*}
MRM_{\alpha}(T_1\rho T^\dag_1)\leq
MRM_{\alpha}(\rho)+1.
\end{eqnarray*}

By Lemma~\ref{lem:rel_mag},  $MRM_{\alpha}(\rho)=S_{\alpha}(\mathcal{M}(\rho))-S_{\alpha}(\rho)$, 
and  $MRM_{\alpha}(T_1\rho T^\dag_1)=S_{\alpha}(\mathcal{M}(T_1\rho T^\dag_1))-S_{\alpha}(T_1\rho T^\dag_1)$.
Hence, we only need to prove 
\begin{eqnarray}\label{230611shi1}
S_{\alpha}(\mathcal{M}(T_1\rho T^\dag_1))\leq
S_{\alpha}(\mathcal{M}(\rho))+1.
\end{eqnarray}
Let us denote the abelian groups $G_{\rho}:=\set{\vec x: |\Xi_{\rho}[\vec x]|=1}$ and 
$G_{T_1\rho T^\dag_1}:=\set{\vec x: |\Xi_{T\rho T^\dag}[\vec x]|=1}$.
Since $\mathcal{M}(\rho)$ and $\mathcal{M}(T_1\rho T^\dag_1)$ are projectional states,
to prove \eqref{230611shi1} one only need to show that
\begin{eqnarray*}
\left|G_{T_1\rho T^\dag_1}\right|
\geq \frac{1}{2}|G_{\rho}|.
\end{eqnarray*}
Consider
the subgroup $G_{\rho, 0}=\set{\vec x\in G: x_1=(0,0)}$,
then 
$G_\rho$ has the following coset decomposition 
\begin{eqnarray*}
G_{\rho}
=G_{\rho, 0}\cup (\vec v_X+G_{\rho, 0})
\cup (\vec v_Y+G_{\rho, 0})
\cup (\vec v_Z+G_{\rho, 0}),
\end{eqnarray*}
where each of $\vec v_X+G_{\rho, 0}$, $\vec v_Y+G_{\rho, 0}$ and $\vec v_Z+G_{\rho, 0}$ has size either  $0$ or $|G_{\rho, 0}|$,
and once $|\vec v_X+G_{\rho, 0}|=|\vec v_Y+G_{\rho, 0}| = |G_{\rho, 0}|$ then we have $|\vec v_Z+G_{\rho, 0}|= |G_{\rho, 0}|$.
The sets $G_{\rho,0}$ and $\vec v_X+G_{\rho, 0}$ are invariant under the action of 
$T_1$, that is, 
for any $\vec u\in G_{\rho,0}$ or $\vec u\in \vec v_Z+G_{\rho, 0}$ , we have $\vec u\in G_{T_1\rho T^\dag_1} $.
Hence
\begin{eqnarray*}
\left|G_{T_1\rho T^\dag_1}\right|\geq \frac{1}{2}\left|G_{\rho}\right|.
\end{eqnarray*}

\end{proof}
In \cite{arunachalam2022parameterized}, the minimal number of $T$ gates to generate the $W$ state
$\ket{W_n}=\frac{1}{\sqrt{n}}\sum_{\vec{x}:|\vec{x}|=1}\ket{\vec{x}}$ is $\Omega(n)$ under the assumption 
of the classical exponential time hypothesis (ETH) (i.e., solving the classical k-SAT problem requires exponential time \cite{impagliazzo2001complexity}).
Here, we can provide an unconditional proof of this statement by the above lemma, as 
$MRM_{\alpha}(W_n)=n$.
\begin{cor}
Any Clifford +T circuit, which starts from the all-zero state and  outputs the state $\ket{W_n}\ot\ket{\text{Junk}}$, 
must contain $\Omega(n)$ T gates.
\end{cor}

\begin{Rem}
Note that for a pure state $\psi$, $MRM_{\alpha}(\psi)$ for $\alpha \geq 1$ is equivalent to  stabilizer nullity \cite{BeverlandQST20} or
stabilizer dimension \cite{grewal2023improved}.
Besides, for $\alpha=1/2$, $MRM_{\alpha}(\psi)$ is equivalent to stabilizer fidelity~\cite{bravyi2019simulation}.
Moreover, $ME^{(N)}_{2}(\psi)$  is equivalent to stabilizer R\'enyi entropy \cite{LeonePRL22} or quantum R\'enyi Fourier  entropy~\cite{Bucomplexity22}.

\end{Rem}

Let us also consider the $n$-qudit system, with $d$ being odd prime. 

 \begin{Def}[Magic entropy on qudit-systems]
Given an $n$-qudit pure state $\psi$, the magic 
entropy $ME(\psi)$ is 
\begin{eqnarray}
\text{ME}(\psi)
:=S(\psi\boxtimes_H\psi).
\end{eqnarray}
The $\alpha$-R\'enyi magic entropy $ME_{\alpha}(\psi)$ is 
\begin{eqnarray}
\text{ME}_{\alpha}(\psi)
:=S_{\alpha}(\psi\boxtimes_H\psi).
\end{eqnarray}

\end{Def}

\begin{prop}
The  $\alpha$-R\'enyi magic entropy $\text{ME}_{\alpha}(\psi)$
satisfies:

(1) $0\leq \text{ME}_{\alpha}(\psi)\leq n\log d$, and $\text{ME}_{\alpha}(\psi)=0$ iff $\psi$ is a stabilizer state;

(2) The $\text{ME}_{\alpha}$ is invariant under a Clifford unitary acting on states;

(3) $\text{ME}_{\alpha}(\psi_1\ot \psi_2)=\text{ME}_{\alpha}(\psi_1)+\text{ME}_{\alpha}(\psi_2)$.

(4) For any  $\alpha\in [1,+\infty]$, 
\begin{eqnarray}
ME_{\alpha}(\psi)\leq MRM_{\alpha}(\psi).
\end{eqnarray}

\end{prop}

\begin{proof}
These properties come directly from the properties of Hadamard convolution in Lemma \ref{lem:prop_HCon}.
\end{proof}

\begin{Rem}

The magic entropy can be also defined on quantum gates by the  Choi-Jamiołkowski isomorphism \cite{Choi75,Jamio72}.
That is, we can consider the magic entropy of Choi states
\begin{eqnarray}
\text{ME}^{(N)}_{\alpha}(\Lambda)
:=\text{ME}^{(N)}_{\alpha}(J_\Lambda),
\end{eqnarray}
where $J_\Lambda$ is the Choi state of a given quantum channel $\Lambda$. 
A similar method also works for $n$-qudit channels.
Based on the results on states, $\text{ME}^{(N)}_{\alpha}(\Lambda)$ can also serve as a measure of magic for 
quantum gates. 
For example, for a given unitary $U$,
the magic entropy of $U$  vanishes iff $U$ is a Clifford unitary.
For $\alpha=2$, $\text{ME}^{(N)}_{\alpha}(\Lambda)$ is equivalent to the
$(p, q)$ group norm defined in~\cite{BuPRA19_stat}.

\end{Rem}
\section{Discussion and future directions}\label{sec:disc}

In this work, we have introduced a series of quantum convolutions on qubits/qudits. 
Based on these quantum convolutions, we provide a systematic framework to implement the
stabilizer testing for states, and Clifford testing for gates.
We also introduce a magic measure named ``magic entropy''
by quantum convolutions on qubits/qudits.
Moreover, due to the universality of 
Gaussian states and quantum Fourier analysis, the results in this work
can also extended to 
other framework such as matchgate and bosonic Gaussian circuits.
In addition,
 there are still several open questions that require further investigation and research, and we outline four such  problems:

(1) The quantum convolution in this work can be realized by CNOT gates and swap tests, both of which can be implemented experimentally. Thus it is possible
to perform the stabilizer testing and measure the magic entropy of many-body quantum systems. This would proceed  in a fashion similar to the successful measurement of entanglement entropy by Greiner and his group~\cite{Islam15} and high-fidelity realization of entangling gates by the Harvard-MIT-QuEra collaboration~\cite{evered2023highfidelity}.

(2)
Inspired by the concept of entanglement spectrum, we  introduce the notion of magic spectrum. We define this as 
the spectrum of the self-convolution $\boxtimes\psi$ for a given state $\psi$. The properties and characteristics of the magic spectrum, especially in the context of many-body quantum systems, warrant further investigation and study.

(3) Linearity testing of Boolean functions
plays an important role in classical error correction, cryptography, and complexity theory.
It would be interesting to find the application of stabilizer testing and its generalization
in quantum error correction codes, and in quantum complexity theory. For example, can such testing give insight to obtain  a better understanding of locally-testable quantum codes~\cite{AharonovQLTC15,Eldar17},  or the quantum-PCP conjecture~\cite{AharonovQPCP13}.

(4) One should investigate the  central limit theorem with random quantum states, and to study  its use in this case.  If each quantum state $\rho_i$ is chosen randomly from some given ensemble $\mathcal{E}$, 
the $\boxtimes_K(\ot_i \rho_i)$ is a random state. A natural question arises: what is the distribution of the output state
$\boxtimes_K(\ot_i \rho_i)$? Is it close to a Haar random state?

\section{Acknowledgments}
We thank Chi-Ning Chou, Roy Garcia, Markus Greiner, Yichen Hu and Yves Hon Kwan for helpful discussion.
This work was supported in part by ARO Grant W911NF-19-1-0302, ARO MURI Grant W911NF-20-1-0082, and NSF Eager Grant 2037687.

\section{Appendix}

\subsection{Quantum convolution of states}\label{appen:quan_con}

\begin{prop}[\bf $\boxtimes_K$ is symmetric]
Given $K$  states $\rho_1,\rho_2,..., \rho_K$,  the convolution $\boxtimes_K$ of 
$\rho_1,\rho_2,..., \rho_K$ is  invariant under permutation, that is, 
\begin{eqnarray}
\boxtimes_K\left(\otimes^K_{i=1}\rho_i\right)
=\boxtimes_K\left(\otimes^K_{i=1}\rho_{\pi (i)}\right),
\end{eqnarray}
for any permutation $\pi$ on $K$ elements. 

\end{prop}
\begin{proof}
This comes directly from Proposition~\ref{230504prop1}.
\end{proof}

\begin{prop}[\bf Convolutional stability for states]
If $\rho_1, \rho_2, \rho_3$ are all stabilizer states, then 
$\boxtimes_3(\ot^3_{i=1}\rho_i)$ is a stabilizer state.
\end{prop}
\begin{proof}
This comes from the fact that the $\boxtimes_3$ is a stabilizer channel, as 
the key unitary only consists of CNOT gates.
\end{proof}

\begin{prop}[\bf{Mean state property}]
Let $\rho_1,\rho_2,\rho_3$ be three $n$-qudit states. 
Then  
\begin{eqnarray}\label{0204shi6-3}
\CMM\left(\boxtimes_3\left(\ot^3_{i=1}\rho_i\right)\right)=
\boxtimes_3\left(\ot^3_{i=1}\CMM(\rho_i)\right)\;.
\end{eqnarray} 
\end{prop}
\begin{proof}
This comes from the definition of mean states and Lemma~\ref{230504prop1}.
\end{proof}

\begin{lem}[ \cite{Appleby05,Gross06,ZhuPRA17,DB13}]\label{230522lem1}
For any prime $d$ and any integer $n$, the following holds:

(1) For each $n$-qudit Clifford unitary $U$, there is a sympletic matrix $M$ and 
a function $f:\mathbb{Z}^{2n}_d\to \mathbb{Z}_d$ such that 
\begin{eqnarray}
Uw(\vec x)U^\dag 
=\omega^{f(\vec x)}_dw(M\vec x).
\end{eqnarray}

(2) 
Conversely, for each symplectic matrix $M$, there is a Clifford unitary $U$ and
a phase function 
$f:\mathbb{Z}^{2n}_d\to \mathbb{Z}_d$  such that the above equation holds. 
If $d$ is odd, one can choose $U$ such that $f\equiv 0 \mod d$.

\end{lem}

\begin{thm}[\bf  Commutativity with Clifford Unitary]\label{thm:commu_clif}
Let $U$ be any Clifford unitary,
then there exists another Clifford unitary $U_1$ such that 
\begin{align}\label{230507shi1}
\boxtimes_K\left(\ot^K_{i=1}\left(U\rho_i U^\dag\right)\right)= U_1\boxtimes_K\left(\ot^K_{i=1}\rho_i\right) U^\dag_1\;, \quad \forall \rho_1,...,\rho_K \in D(\mathcal{H}^{\ot n}).
\end{align}
Moreover, if $K=4N+1$, then $U_1=U$.
\end{thm}
\begin{proof}
It follows from Lemma \ref{230522lem1} that 
there exists a function $f(\vec p,\vec q)$ and a symplectic matrix $M$ such that
$U^\dag w(\vec p,\vec q) U = (-1)^{f(\vec p,\vec q)}w(M (\vec p,\vec q))$ for any $(\vec p,\vec q)\in V^n$.
Hence, 
\begin{align*}
(U^\dag)^{\otimes K} \boxtimes_K^\dag(w(\vec p,\vec q)) U^{\otimes K} 
=&  (-1)^{N \vec p\cdot \vec q} (U^\dag)^{\otimes K} w(\vec p,\vec q)^{\otimes K}  U ^{\otimes K}\\
=& (-1)^{N \vec p\cdot \vec q+ Kf(\vec p, \vec q)}  w(M (\vec p,\vec q))^{\otimes K}\\
=& (-1)^{N \vec p\cdot \vec q+ f(\vec p, \vec q)}  w(M (\vec p,\vec q))^{\otimes K},
\end{align*}
where the last equality comes from the fact that $K=2N+1$ is odd. Let us consider the problem in two cases:
(a) $N$ is even; (b) $N$ is odd.

(a) $N$ is even, then $(U^\dag)^{\otimes K} \boxtimes_K^\dag(w(\vec p,\vec q)) U^{\otimes K} = (-1)^{ f(\vec p, \vec q)}  w(M (\vec p,\vec q))^{\otimes K}$.
Let us take $U_1=U$ and denote $(\vec p' ,\vec q')=M(\vec p, \vec q)$, we have 
\begin{align*}
\boxtimes_K^\dag\left( U^\dag  w(\vec p,\vec q) U \right)
=&\boxtimes_K^\dag\left( (-1)^{f(\vec p, \vec q)}  w(M (\vec p,\vec q))\right) \\
=& (-1)^{N\vec p'\cdot\vec q'} (-1)^{ f(\vec p, \vec q) }  w(M (\vec p,\vec q))^{\otimes K} \\
=&(-1)^{ f(\vec p, \vec q) }  w(M (\vec p,\vec q))^{\otimes K}.
\end{align*}
Therefore, \eqref{230507shi1} holds.

(b) $N$ is odd,  $(U^\dag)^{\otimes K} \boxtimes_K^\dag(w(\vec p,\vec q)) U^{\otimes K} = (-1)^{\vec p\cdot \vec q+ f(\vec p, \vec q)}  w(M (\vec p,\vec q))^{\otimes K}$.
Since $M=\left[
\begin{array}{cc}
A&B\\
C&D
\end{array}\right]
$ 
is symplectic, that is $M^T\left[
\begin{array}{cc}
0&I\\
I&0
\end{array}\right]
M=\left[
\begin{array}{cc}
0&I\\
I&0
\end{array}\right]
$,
 then the matrices $A,B,C,D$ satisfies the following properties
$
A^TC+C^TA\equiv 0$,
$B^TD+D^TB\equiv 0$,
$A^TD+C^TB\equiv I
\mod 2$. 
Hence, both $A^TC$ and $B^TD$ are symmetric, i.e., $[A^TC]_{ij}=[A^TC]_{ji}$ and $[B^TD]_{ij}=[B^TD]_{ji}$. 
For $(\vec p' ,\vec q')=M(\vec p, \vec q)=(A\vec p+B\vec q, C\vec p+D\vec q)$, we have 
\begin{eqnarray*}
\vec p'\cdot \vec q'&=&
(A\vec p+B\vec q)^T
(C\vec p+D\vec q)\\
&=&\vec p^T A^TC\vec p
+\vec q^T B^TD\vec q
+\vec p^T(A^TD+C^TB)\vec q\\
&\equiv&
\sum_i[A^TC]_{ii}p_i
+\sum_i[B^TD]_{ii}q_i
+\vec p\cdot \vec q \mod 2,
\end{eqnarray*}
where  the last $\equiv$ comes from the 
fact that  both $A^TC$ and $B^TD$ are symmetric, $p^2_i=p_i$ for $p_i\in\set{0,1}$ and
 $A^TD+C^TB\equiv I$.
 Let us denote $\vec a=(a_1,..,a_n)$ with $a_i=[A^TC]_{ii}$, and 
 $\vec b=(b_1,..,b_n)$ with $b_i=[B^TD]_{ii}$. Then 
 $\vec p'\cdot \vec q'\equiv \vec a\cdot\vec p +\vec b\cdot\vec q+\vec p\cdot \vec q \mod 2$.
Let us define $U_1=X^{\vec a}Z^{\vec b}U$, then 
\begin{align*}
&\boxtimes_K^\dag\left( U^\dag_1  w(\vec p,\vec q) U_1 \right)\\
=&\boxtimes_K^\dag\left( (-1)^{f(\vec p, \vec q)+\vec a\cdot\vec p+\vec b\cdot\vec q}  w(M (\vec p,\vec q))\right) \\
=& (-1)^{\vec p'\cdot\vec q'} (-1)^{ \vec a\cdot\vec p+\vec b\cdot\vec q+f(\vec p, \vec q) }  w(M (\vec p,\vec q))^{\otimes K}\\
=&(-1)^{\vec p\cdot\vec q+ f(\vec p, \vec q) }  w(M (\vec p,\vec q))^{\otimes K}.
\end{align*}
Therefore, \eqref{230507shi1} holds.

\end{proof}

 Based on the characteristic function of the convolution 
$\boxtimes_3$ in Lemma~\ref{230504prop1}, we have the following statement directly.

\begin{cor}\label{lem:conv_iden}
Let  $\rho_1,\rho_2,\rho_3$ be $n$-qubit states where at least one of them is $I_n/2^n$, then
\begin{align*}
  \boxtimes_3\left( \rho_1, \rho_2,\rho_3\right) = \frac{I_n}{2^n} \;.
\end{align*}
\end{cor}

\begin{Def}[\bf Schur concavity \cite{MOA79}]
A function $f:\real^n_+\to \real$ is 
Schur concave if $\vec p\prec \vec q$ implies that $f(\vec p)\geq f(\vec q)$.
A function is 
strictly Schur concave if $\vec p\prec \vec q$ implies that $f(\vec p)> f(\vec q)$ except for $\vec p= \vec q$.
\end{Def}

Let us consider two well-known examples of Schur-concave functions: the subentropy and the generalized quantum R\'enyi entropy.

\begin{thm}[\bf Majorization under quantum convolution]\label{thm:entropy}
Given $3$ $n$-qubit states $\rho_1,\rho_2,\rho_3$, with 
$\vec{\lambda}_{1}, \vec{\lambda}_2, \vec{\lambda}_3$ the vectors of  eigenvalues of $\rho_1, \rho_2, \rho_3$. 
Then we have 

\begin{eqnarray}\label{inedq:major_2}
\vec{\lambda}_{\boxtimes_3(\rho_1, \rho_2,\rho_3)} 
\prec \vec{\lambda}_{\rho_i} \;,\quad i=1,2, 3.
\end{eqnarray}
Thus for any Schur-concave function $f$, 
\begin{eqnarray}
f(\boxtimes_3(\rho_1, \rho_2, \rho_3))
\ge f(\rho_i)\;,\quad i=1,2, 3.
\end{eqnarray}
\end{thm}
\begin{proof}
Here, we  prove $\vec{\lambda}_{\boxtimes_3(\rho_1, \rho_2,\rho_3)} 
\prec \vec{\lambda}_{\rho_1}$; the other cases can be proved in the same way.
Consider the spectral decompositions of the states $\rho_1$ and $\rho=\boxtimes_3(\rho_1,\rho_2, \rho_3)$ as 
$
\rho_1=  \sum_{j=1}^{2^n} \lambda_j |\psi_j\rangle \langle \psi_j|$,
and $
\rho  = \sum_{j=1}^{2^n} \mu_j |\xi_j\rangle \langle \xi_ j| 
$.
We need to prove that
$
(\mu_1,...,\mu_{2^n}) \prec(\lambda_1,...,\lambda_{2^n})\;
$.
Let 
$
\tau_j
:=\boxtimes_3\left(
|\psi_j\rangle \langle \psi_j| ,\rho_{2}, \rho_3\right)\;.
$
Then   $\tau_j$ is a quantum state.
Moreover,
\begin{align}\label{0102shi1}
\sum_{j=1}^{2^n} \lambda_j \tau_j =  \sum_{j=1}^{2^n}  \lambda_j \; \boxtimes_3\left(
|\psi_j\rangle \langle \psi_j| ,\rho_{2}, \rho_3\right) = \boxtimes_3\left(\rho_1,\rho_2,\rho_3\right) =\rho \;,
\end{align}
and
\begin{align}\label{0102shi3}
\sum_{j=1}^{2^n} \tau_j = \sum_{j=1}^{2^n} \boxtimes_3\left(
|\psi_j\rangle \langle \psi_j| ,\rho_{2},\rho_3\right)=
\boxtimes_3\left(
\left(\sum^{2^n}_{j=1} \proj{\psi_j}\right) ,\rho_{2}, \rho_3\right) =   I \;,
\end{align}
where the last equality follows from Lemma \ref{lem:conv_iden}.
Consider the $2^n\times 2^n$ matrix $M=\left( m_{\imath j}\right)_{\imath,j=1}^{2^n}$, where each entry $m_{\imath j}$ is defined as
$
m_{\imath j} = \langle \xi_\imath| \tau_j  | \xi_\imath \rangle \;.
$
By definition, each $m_{\imath j}\ge 0$, and
\begin{align}
\sum_{\imath } m_{\imath j} =& \sum_{\imath }  \langle \xi_\imath| \tau_j  | \xi_\imath \rangle = \Tr{\tau_j}= 1 \;,\\
\sum_{j } m_{\imath j} =&\sum_{j } \langle \xi_k| \tau_j  | \xi_k \rangle = \langle \xi_k| I  | \xi_k \rangle =1 \;, \label{0102shi2}
\end{align}
where (\ref{0102shi2}) comes from (\ref{0102shi3}).
Thus $M$ is a doubly stochastic matrix.
Moreover, by (\ref{0102shi1}),
\begin{align}\label{0102shi4}
\mu_i= \langle \xi_\imath| \rho   | \xi_\imath \rangle = \sum_{j=1}^{2^n} \lambda_j \langle \xi_\imath| \tau_j  | \xi_\imath \rangle=  \sum_{j=1}^{2^n} \lambda_j m_{ij}\;.
\end{align}
That is, 
$
(\mu_1,...,\mu_{2^n})^T = M (\lambda_1,...,\lambda_{2^n}) ^T \;.
$
Based on Proposition 1.A.3 in \cite{MOA79}, $(\mu_1,...,\mu_{2^n}) \prec(\lambda_1,...,\lambda_{2^n})\;$.

\end{proof}

Hence we have the following corollary on the self-convolution $\boxtimes_3\psi$ of a pure state. 

\begin{prop}[\bf Stabilizer purity preservation under convolution]\label{prop:pure_stab}
Given a pure $n$-qubit state $\psi$, its self-convolution $\boxtimes_3\psi$ is pure iff 
$\psi$ is a stabilizer state.
\end{prop}

\begin{prop}\label{prop:entropy}
Given $3$ $n$-qubit states $\rho_1,\rho_2,\rho_3$. For 
$\alpha\in[0,+\infty]$, 
\begin{eqnarray}
S_{\alpha}\left(\boxtimes_3(\rho_1, \rho_2, \rho_3)\right)
\geq \max\set{S_{\alpha}(\rho_1),S_{\alpha}(\rho_2),S_{\alpha}(\rho_3) }.
\end{eqnarray}
\end{prop}
\begin{proof}
This is because $S_{\alpha}$ is  Schur concave.
\end{proof}

\begin{Def}[\bf Quantum Tsallis entropy~\cite{tsallis1988possible}]
Given a quantum state $\rho$,
the Tsallis entropy with parameter $\alpha\geq 1$ is 
\begin{eqnarray}
T_{\alpha}(\rho)
:=\frac{1}{1-\alpha}\left(\Tr{\rho^\alpha}-1\right).
\end{eqnarray}
\end{Def}

\begin{prop}
Given $3$ $n$-qubit states $\rho_1,\rho_2,\rho_3$, for any $\alpha\geq 1$ we have 
\begin{eqnarray}
T_{\alpha}(\boxtimes_3(\rho_1,\rho_2,\rho_3))\geq 
\max\set{T_{\alpha}(\rho_1), T_{\alpha}(\rho_2), T_{\alpha}(\rho_3)}\;.
\end{eqnarray}
\end{prop}
\begin{proof}
This is because the Tsallis entropy is  Schur concave for $\alpha\geq 1$.
\end{proof}

\begin{Def}[\bf Subentropy \cite{datta2014suben}]\label{def:suben}
Given a quantum state $\rho$ with eigenvalues $\set{\lambda_i}^d_{i=1}$, the subentropy of
$\rho$ is 
\begin{eqnarray}
Q(\rho):=\sum^d_{i=1}\frac{\lambda^d_i}{\Pi_{j\neq i}(\lambda_j-\lambda_i)}\log\lambda_i.
\end{eqnarray}
\end{Def}

\begin{prop}
Given $3$ $n$-qubit states $\rho_1,\rho_2,\rho_3$, 
we have 
\begin{eqnarray}
Q(\boxtimes_3(\rho_1,\rho_2,\rho_3))\geq 
\max\left\{Q(\rho_1), Q(\rho_2), Q(\rho_3)\right\}.
\end{eqnarray}
\end{prop}
\begin{proof}
This is because the subentropy $Q$ is  Schur concave.
\end{proof}

\begin{lem}\label{lem:eq_entr}
Let $\rho_1,\rho_2, \rho_3$ be  $n$-qubit states, 
and let
$\rho_k =  \sum_{j=1}^{T} \mu_j P_j $ be the spectral decomposition of the quantum state $\rho_k$,
where $P_j$ is the projection to the eigenspace corresponding to the eigenvalue $\mu_j$, and $\mu_\imath\neq\mu_j$ for any $\imath\neq j$.
For any $\alpha\not\in \{-\infty,0,\infty\}$,
the equality
$$S_\alpha\left( \boxtimes_3\left(\ot^3_{i=1}\rho_i\right)\right) = S_\alpha(\rho_k)$$
holds iff each $Q_j : =\boxtimes_K(P_j\ot_{i\neq k}\rho_i)$ is a projection of the same rank as $P_j$  and $Q_{j_1}\perp Q_{j_2}$ whenever $j_1\neq j_2$.
\end{lem}
\begin{proof}
Without loss of generality, let us assume that $k=1$.
We follow the notations in the proof of Theorem \ref{thm:entropy}.
We have  the spectral decompositions of the states $\rho_1$ and $\rho=\boxtimes_3(\rho_1,\rho_2, \rho_3)$ as 
$
\rho_1=  \sum_{j=1}^{2^n} \lambda_j |\psi_j\rangle \langle \psi_j|$,
and $
\rho  = \sum_{j=1}^{2^n} \mu_j |\xi_j\rangle \langle \xi_ j| 
$, and assume $\{\lambda_j\}$ and $\{\mu_j\}$ are listed in non-increasing order.

Since ("$\Leftarrow$") is trivial, we only need to prove 
("$\Rightarrow$"). Recall that we are assuming $\alpha\not\in \{-\infty,0,\infty\}$.
Since $
(\mu_1,...,\mu_{2^n}) \prec(\lambda_1,...,\lambda_{2^n})$
and we have
\begin{align}\label{0109shi2}
\sum_{j=1}^J \mu_j \le \sum_{j=1}^J \lambda_j,\quad J=1,2,3,..., 2^n \;.
\end{align}
Moreover, 
$S_\alpha(\boxtimes_3(\rho_1,\rho_2,\rho_3)) = S_\alpha(\rho_1)$ iff every equality in (\ref{0109shi2}) holds, which means $ (\mu_1,...,\mu_{2^n}) = (\lambda_1,...,\lambda_{2^n})$.

Assume $L$ is the largest number such that $\mu_1=\mu_L$.
For any $l\le L$,
\begin{align}\label{0116shi1}
\mu_l = \sum_{j}m_{lj} \lambda_j = \sum_{j}m_{lj} \mu_j = \left(\sum_{j=1}^L m_{lj}\right) \mu_l + \sum_{j>L} m_{lj}\mu_j\;.
\end{align}
If there exists $m_{lj}>0$ for some $l\le L$ and $j>L$, then
\begin{align*}
(\ref{0116shi1}) \,< \left(\sum_{j=1}^L m_{lj}\right) \mu_l + \sum_{j>L} m_{lj} \mu_L= \mu_L \;,
\end{align*}
a contradiction.
Hence  $m_{lj}=0$ whenever $l\le L$ and $j>L$,
and therefore for every $l\le L$ we have
\begin{align*}
\sum_{j=1}^L m_{lj} =1 \;.
\end{align*}
Note that $M$ is a doubly stochastic matrix, so we also have $m_{lj}=0$ when $l>L$ and $j\le L$.
By the definition of $m_{lj}$,
we have $\langle \xi_l| \tau_j  | \xi_l \rangle=0$ whenever $l\le L< j$ or $j\le L<l$.
Denote $\sigma_l= | \xi_l \rangle\langle \xi_l|$,
then $\tau_j \sigma_l=0$ when $l\le L< j$ or when $j\le L<l$.
Therefore $\tau_j \sum_{l=1}^L \sigma_l=0$ when $j>L$,
and $\tau_j \sum_{l=L+1}^{2^n} \sigma_l=0$ when $j\le L$.

Denote $\psi_j=|\psi_j\rangle \langle \psi_j|$,
then by definition $\tau_j = \boxtimes_3( \psi_j ,\rho_{2}, \rho_3)$.
When $l\le L$,
\begin{align*}
\langle \xi_l| \sum_{j=1}^L \tau_j | \xi_l \rangle = \sum_{j=1}^L m_{lj} =1 \;,
\end{align*}
that is
\begin{align}\label{0109shi4}
\langle \xi_l|\boxtimes_3\left(\left(\frac 1L\sum_{j=1}^L \psi_j\right) ,\rho_{2}, \rho_3\right)  | \xi_l \rangle =\frac 1L,\quad \forall l\le L \;.
\end{align}
Let's denote $\Phi=   \frac 1L\sum_{j=1}^L \tau_j = \boxtimes_3\left( \frac 1L\sum_{j=1}^L  \psi_j ,\rho_{2}, \rho_3\right)$.
Taking $\alpha=2$ in Theorem \ref{thm:entropy}, we have
\begin{align*}
 \left\|\Phi \right\|_2^2 \le \left\| \frac 1L \sum_{j=1}^L \psi_j \right\|_2^2 = \frac{1}{ L} \;.
\end{align*}
Expanding the matrix $\Phi$ under the basis $\{\xi_l\}$, we have
\begin{align*}
\sum_{l,l'}\left|\langle \xi_l|\Phi| \xi_{l'} \rangle\right|^2 \le \frac{1}{L} \;,
\end{align*}
and compared with (\ref{0109shi4}) we have that
\begin{align*}
\boxtimes_3\left(\frac 1L\sum_{j=1}^L  \psi_j ,\rho_{2}, \rho_3\right) = \Phi = \frac{1}{L}\sum_{l=1}^L | \xi_{l} \rangle\langle \xi_l| \;. 
\end{align*}
That is
\begin{eqnarray*}
\boxtimes_3( P_1,\rho_{2}, \rho_3)  =Q_1 \;,
\end{eqnarray*}
where $P_1$ is the spectral projection of $\rho_1$ corresponding to the eigenvalue $\lambda_1$, and
$Q_1$ is the spectral projection of $ \boxtimes_3(\rho_1,\rho_2, \rho_3)$ corresponding to $\mu_1$.
Repeat this process (by replacing $\rho_1$ by $(\rho_1 - \lambda_1 P_1)/\trace[\rho_1 -\lambda_1 P_1]$ )
and we obtain that, for each spectral projection $P_j$ of $\rho_k$, $Q_j:= \boxtimes_3(P_j,\rho_{2}, \rho_3)$ is a projection of same rank as $P_j$, and $Q_j$ 
is a spectral projection of $\boxtimes_3(\rho_1, \rho_2, \rho_3)$.
\end{proof}

\begin{thm}[\bf The case of equality]\label{0112thm1}
Let $\rho_1,\rho_2, \rho_3$ be  $n$-qubit states,
 $\alpha\notin  \{-\infty, 0, +\infty\}$,
and
\begin{align}\label{0214shi1}
G_k=\{ w( \vec p,  \vec q): |\Xi_{\rho_j} (\vec p, \vec q)|=1 \;\text{ for }\; j\neq k\}.
\end{align}
The equality
\begin{align}\label{0110shi8}
S_\alpha\left(\boxtimes_3\left(\ot^3_{i=1}\rho_i\right)\right) =  S_\alpha(\rho_k)
\end{align}
holds for some $1\le k\le 3$ iff $\rho_k$ is in the abelian C*-algebra generated by $G_k$, i.e.,
$\rho_k$ is a convex sum of MSPSs associated with $G_k$.
\end{thm}
\begin{proof}
We only prove the case where $k=1$; the proof is similar for the other cases.

First, since
\[G_1 = \bigcap_{j=2}^K \{ w( \vec p,  \vec q): |\Xi_{\rho_j} (\vec p, \vec q)|=1\} \]
is an intersection of abelian subgroups,
$G_1$ is an abelian subgroup.

Now assume $\rho_1$ is a state in the  C*-algebra generated by elements in $G_1$,
and we will show the equality \eqref{0110shi8} holds.
For each MSPS $\sigma_j$ associated with $G_1$ we have
$$|\Xi_{\sigma_j}(\vec p,\vec q)| =\left\{
\begin{aligned}
 &1 &&  w( \vec p,\vec q) \in G_1,\\
&0 && w( \vec p,\vec q) \not \in G_1.
\end{aligned}\right.
$$
Combined with the definition \eqref{0214shi1} of $G_1$, 
we have
$$|\Xi_{\boxtimes_3(\sigma_j ,\rho_2,\rho_3)} (\vec p,\vec q)| = |\Xi_{\sigma_j}( \vec p,  \vec q) \Xi_{\rho_2}( \vec p,  \vec q) \Xi_{\rho_3}( \vec p,  \vec q) | =\left\{
\begin{aligned}
 &1 &&  w( \vec p,\vec q) \in G_1,\\
&0 && w( \vec p,\vec q) \not \in G_1.
\end{aligned}\right.
$$
That is,
$\boxtimes_3(\sigma_j ,\rho_2,\rho_3)$ is a MSPS associated with the abelian group $G_1$.

Moreover,
$$\boxtimes_3(\sigma_j ,\rho_2, \rho_3) = \boxtimes_3(\sigma_\imath ,\rho_2,\rho_3),$$ if and only if $\sigma_j = \sigma_\imath$.
Hence the map $\sigma_j \mapsto \boxtimes_K(\sigma_j ,\rho_2, \rho_3)$ is a bijection from the set of MSPSs associated with $S$ to the set of MSPSs associated with $G_1$.
Therefore,
if $ \rho_1=\sum \mu_j \sigma_j
 $ is a linear sum of MSPSs associated with $G_1$,
 then 
 $$ \boxtimes_3(\rho_1 ,\rho_2,\rho_3) = \sum \mu_j \boxtimes_3(\sigma_j ,\rho_2,\rho_3),
 $$
is a linear sum (with the same coefficients) of MSPSs associated with $G_1$.
Hence, the equality \eqref{0110shi8} holds.

On the other hand, let us consider the case where  \eqref{0110shi8} holds.
Let $P$ be any spectral projection of $\rho_1$ and assume $\mathrm{rank}(P)=r$.
In the following, 
we show $P$ is a stabilizer projection associated with $G_1$.
By Lemma \ref{lem:eq_entr},
$\boxtimes_3(P ,\rho_2,\rho_3)$ is a projection $Q$ of rank $r$, 
hence $\|\frac1r P\|_2 = \|\frac1r Q\|_2$ and $\|\Xi_{\frac1r P}\|_2 = \|\Xi_{\frac1r Q}\|_2$.
While
\begin{align*}
\left|\Xi_{\frac1rQ}(\vec p, \vec q) \right|=  \left| \Xi_{ \frac1r P}( \vec p,  \vec q) \Xi_{\rho_2} ( \vec p,  \vec q)  \Xi_{\rho_3} ( \vec p,  \vec q) \right|\;,\quad \forall  (\vec p,\vec q) \in V^n \;.
\end{align*}
Thus, whenever $(\vec p,\vec q) \in V^n$ satisfies $\Xi_{ \frac1r P}( \vec p, \vec q)\neq 0$,
we  have $\left| \Xi_{ \frac1r P}( \vec p,  \vec q) \Xi_{\rho_2} ( \vec p,  \vec q)  \Xi_{\rho_3} ( \vec p,  \vec q) \right|=1$;
or equivalently,
\begin{align*}
\Xi_{ \frac1r P}( \vec p,  \vec q)\neq 0 \quad \Rightarrow \quad |\Xi_{\rho_2} ( \vec p, \vec q)|=|\Xi_{\rho_3} ( \vec p, \vec q)| =1 \;.
\end{align*}
That is, the characteristic function $ \Xi_{ \frac1r P}$ of $\frac1r P$ supports on $G_1$,
therefore
\begin{align*}
\frac1r P = \frac{1}{2^n} \sum_{ (\vec p, \vec q)\in S_1} \Xi_{ \frac1r P}( \vec p, \vec q) w(\vec p, \vec q) \in C^*(G_1) \;.
\end{align*}
Therefore $P$  is a stabilizer projection associated with $G_1$.
Hence $\rho_1$ is a linear sum of projections in the abelian C*-algebra generated by $G_1$.
\end{proof}

\begin{Rem}
The  entropic inequalities have also been studied in the other quantum case, including the free probability theory \cite{szarek1996volumes,shlyakhtenko2007shannon,shlyakhtenko2007free}, 
continuous-variable quantum systems~\cite{Konig14,Palma14}, subfactor theory~\cite{HuangLiuWu22}, and qudit systems~\cite{Audenaert16, Carlen16, BGJ23a,BGJ23b}.
\end{Rem}

\begin{lem}\label{lem:mean_con}
Given an $n$-qubit state $\rho$ and $K=2N+1$, then 
\begin{eqnarray}
\mathcal{M}(\boxtimes_K\rho)
=\mathcal{M}(\rho^\#),
\end{eqnarray}
where $\rho^\# = \rho$ for  even $N$, and $\rho^\# = \rho^T$ for  odd $N$.
\end{lem}
\begin{proof}
By the definition of mean states, 
$\Xi_{\mathcal{M}(\boxtimes_K\rho)}(\vec p, \vec q)=(-1)^{N\vec p\cdot \vec q}\Xi^K_\rho ( \vec p, \vec q)=(-1)^{N\vec p\cdot \vec q}\Xi_\rho ( \vec p, \vec q)$ if $|\Xi_\rho ( \vec p, \vec q)|=1$, otherwise $\mathcal{M}(\boxtimes_K\rho)(\vec p, \vec q)=0$.
Hence, $\mathcal{M}(\boxtimes_K\rho)
=\mathcal{M}(\rho^\#)$.
\end{proof}

\begin{thm}[\bf Quantum central limit theorem ]\label{thm:CLT_1}
Let $\rho$ be an $n$-qubit state and $K=2N+1$,
then
\begin{eqnarray}\label{230509shi2}
\norm{\boxtimes_{K}\rho-\mathcal{M}(\rho^\#)}_2
\leq (1-MG(\rho))^{K-1} \norm{\rho- \mathcal{M}(\rho)}_2 \;,
\end{eqnarray}
where $\rho^\# = \rho$ for  even $N$, and $\rho^\# = \rho^T$ for  odd $N$.
\end{thm}
\begin{proof}

Let $S$ be the abelian subgroup associated with $\mathcal{M}(\rho)$,
then 
$\Xi_{\mathcal{M}(\rho)}(\vec p, \vec q) =\pm 1$ for  $(\vec p, \vec q)\in S$, and 
$\Xi_{\mathcal{M}(\rho)}(\vec p, \vec q) =0$ for $(\vec p, \vec q)\notin S$.
By Proposition \ref{230504prop1}, we have 
$
\mathcal{M}(\boxtimes_K\rho)
=\mathcal{M}(\rho^\#) \;
$.
Hence
\begin{eqnarray*}
\boxtimes_{K}\rho-\mathcal{M}(\rho^\#)
=\frac{1}{d^n}
\sum_{(\vec p, \vec q)\notin S}
\Xi_{\boxtimes_{K}\rho}(\vec p, \vec q)
w(\vec p, \vec q) \;.
\end{eqnarray*}
Therefore,
\begin{align*}
\norm{\boxtimes_K\rho-\mathcal{M}(\rho^\#)}_2^2
=\frac{1}{d^n}
\sum_{(\vec p, \vec q)\notin S}
|\Xi_{\boxtimes_{K}\rho}(\vec p, \vec q)|^2
\leq \frac{1}{d^n}(1-MG(\rho))^{2K-2}
\sum_{(\vec p, \vec q)\notin S}|\Xi_{\rho}(\vec p, \vec q)^2|
= (1-MG(\rho))^{2K-2}\norm{\rho- \mathcal{M}(\rho)}^2_2 \;,
\end{align*}
where $|\Xi_{\boxtimes_K\rho}(\vec p, \vec q)|
=|\Xi_{\rho}(\vec p, \vec q)|^K
\leq (1-MG(\rho))^{K-1}
|\Xi_{\rho}(\vec p, \vec q)|
$
comes directly from Proposition \ref{230504prop1}.
\end{proof}
\begin{Rem}
Other versions of the quantum central limit theorem have been considered 
 \cite{Cushen71,Lieb73,Lieb1973,Giri78,Goderis89,Matsui02,Cramer10,Jaksic09,Arous13,Michoel04,GoderisPTRT89,JaksicJMP10,Accardi94,Liu16,JiangLiuWu19,Hayashi09,CampbellPRA13,BekerCMP21,Carbone22}, including 
 results in DV quantum information theory~\cite{BGJ23a,BGJ23b},  subfactor theory~\cite{Liu16,JiangLiuWu19}, quantum walks on a lattice ~\cite{Carbone22}, CV quantum information theory~\cite{CampbellPRA13,BekerCMP21}, and 
 the free probability theory~\cite{voiculescu1986addition,voiculescu1987multiplication}.
\end{Rem}
\subsection{Quantum convolution of channels}

We focus here  on the convolutions of $n$-qubit channels, i.e., the quantum channels
acting on $n$-qubit systems.
By the Choi-Jamiołkowski isomorphism \cite{Choi75,Jamio72}, any  quantum
channel $\Lambda$ from $\mathcal{H}_{A}$ to $\mathcal{H}_{A'}$ can be represented by its Choi state
\begin{eqnarray*}
J_{\Lambda}=id_{A}\ot \Lambda (\proj{\Phi}) \;,
\end{eqnarray*}
where $| \Phi \rangle = \frac{1}{\sqrt{2^n}} \sum_{\vec j \in \mathbb{Z}_2^n}\ket{\vec j}_{A}\ot\ket{\vec j}_{A'} $.
For any input state $\rho_A$, the output state of the quantum channel $\Lambda(\rho_A)$ can be represented via the Choi state $J_{\Lambda}$ as
\begin{eqnarray}\label{0127shi2}
\Lambda(\rho_A)
=2^n\Ptr{A}{J_{\Lambda} \cdot \rho^T_{A}\ot I_{A'}} \;.
\end{eqnarray}
On the other hand, for any operator $J$ on $\mathcal{H}_A\ot\mathcal{H}_{A'}$,
the map
\begin{eqnarray*}
\rho\to
2^n\Ptr{A}{J \cdot \rho^T_{A}\ot I_{A'}},
\end{eqnarray*}
is (1) completely positive if and only if $J$ is positive, (2) trace-preserving
if and only if $\Ptr{A'}{J}=I_n/2^n$.

\begin{lem}[\bf The convolution of Choi states is Choi]
Given 3 quantum states $\rho_{AA'},
\sigma_{AA'}$ and $\eta_{AA'}$ on $\mathcal{H}_A\ot\mathcal{H}_{A'}$ with
$\Ptr{A'}{\rho_{AA'}} =\Ptr{A'}{\sigma_{AA'} } =\Ptr{A'}{\eta_{AA'} }  =I_A/2^n$.
Then
$\boxtimes_3(\rho_{AA'},\sigma_{AA'},\eta_{AA'})$ also satisfies that $\Ptr{A'}{\boxtimes_3(\rho_{AA'}, \sigma_{AA'}, \eta_{AA'})}=I_A/2^n$.
\end{lem}

\begin{proof}
By Proposition \ref{230504prop1}, we have
\begin{eqnarray*}
&&\Xi_{\boxtimes^3(\rho_{AA'}, \sigma_{AA'}, \eta_{AA'})} (\vec{p}_A,\vec{p}_{A'}, \vec{q}_A, \vec{q}_{A'})\\
&=&(-1)^{ \vec{p}_A\cdot \vec{q}_A+ \vec{p}_{A'}\cdot\vec{q}_{A'} } \Xi_{\rho} (\vec{p}_A,\vec{p}_{A'}, \vec{q}_A, \vec{q}_{A'})
\Xi_{\sigma}(\vec{p}_A,\vec{p}_{A'}, \vec{q}_A, \vec{q}_{A'}) 
\Xi_{\eta}(\vec{p}_A,\vec{p}_{A'}, \vec{q}_A, \vec{q}_{A'})\;.
\end{eqnarray*}
Since $\Ptr{A'}{\rho_{AA'}} =\Ptr{A'}{\sigma_{AA'} } =\Ptr{A'}{\eta_{AA'} }  =I_A/2^n$,
we have
$$\Xi_{\rho}(\vec{p}_A,\vec{0}_{A'}, \vec{q}_A, \vec{0}_{A'})= \Xi_{\sigma}(\vec{p}_A,\vec{0}_{A'}, \vec{q}_A, \vec{0}_{A'})= \Xi_{\eta}(\vec{p}_A,\vec{0}_{A'}, \vec{q}_A, \vec{0}_{A'})=0\;,\quad \forall (\vec{p}_A,\vec{q}_{A})\neq (\vec 0, \vec 0)\;, $$
hence
\begin{eqnarray*}
\Xi_{\boxtimes^3(\rho_{AA'}, \sigma_{AA'}, \eta_{AA'})} (\vec{p}_A, \vec{0}_{A'}, \vec{q}_A, \vec{0}_{A'})
=0\;,\quad \forall (\vec{p}_A,\vec{q}_{A})\neq (\vec 0, \vec 0)\;.
\end{eqnarray*}
Therefore $$\Ptr{A'}{\boxtimes^3(\rho_{AA'}, \sigma_{AA'}, \eta_{AA'})} = \frac{1}{2^n} \sum_{\vec{p}_A, \vec{q}_A} \Xi_{\boxtimes^3(\rho_{AA'}, \sigma_{AA'}, \eta_{AA'})}(\vec{p}_A, \vec{0}_{A'}, \vec{q}_A, \vec{0}_{A'}) w( \vec p_A, \vec q_A) =I_A/2^n, $$
and the proof is complete.

\end{proof}

To distinguish from the convolution $\boxtimes_K$ of states , let us denote the convolution of channels as 
$\boxtimes^K$. Here let us first consider the  $K=3$ case as it can generate $\boxtimes^K$ for any odd $K$.

\begin{Def}[\bf Convolution of channels]\label{def:con_chan}
Given 3 $n$-qudit channels $\Lambda_1$, $\Lambda_2$ and $\Lambda_3$, the convolution $\boxtimes^3(\Lambda_1, \Lambda_2,\Lambda_3)$
is the quantum channel  with the Choi state
\begin{align*}
J_{\boxtimes^3(\Lambda_1, \Lambda_2,\Lambda_3)} :=\boxtimes_3\left( \ot^3_{i=1}J_{\Lambda_i}\right)\;,
\end{align*}
where the state $ \boxtimes_3( \ot^3_{i=1}J_{\Lambda_i})$ is the convolution of the Choi states  $J_{\Lambda_1}$, $J_{\Lambda_2}$ and $J_{\Lambda_3}$.
\end{Def}

Since the convolution of states is invariant under permutation, it follows that the convolution of channels is also invariant under permutation.
\begin{cor}
Given 3 $n$-qubit channels $\Lambda_1,\Lambda_2,\Lambda_3$, then 
\begin{eqnarray}
\boxtimes^3(\Lambda_1,\Lambda_2,\Lambda_3)
=\boxtimes^3(\Lambda_{\pi(1)},\Lambda_{\pi(2)},\Lambda_{\pi(3)}),
\end{eqnarray}
for any permutation $\pi$ on 3 elements.
\end{cor}

Denote the inverse of the convolution $\boxtimes^{-1}_3$ to be
\begin{eqnarray}\label{eq:inver_conv_chn}
\boxtimes^{-1}_3(\rho)
=V^\dag \left(\rho\ot \frac{I_n}{2^n}\ot \frac{I_n}{2^n} \right) V \;,
\end{eqnarray}
which satisfies that $\boxtimes_3\circ\boxtimes^{-1}_3=id$. Besides, we observe that
$\boxtimes^{-1}_3= \frac{1}{2^{2n}}\boxtimes^\dag_3$.

\begin{thm}[\bf Exact formula for convolution of channels]\label{thm:exact_conv_chan}
Given $3$ $n$-qubit channels $\Lambda_1,\Lambda_2, \Lambda_3$, their convolution $\boxtimes^3(\Lambda_1, \Lambda_2, \Lambda_3)$
is
\begin{eqnarray}\label{eq:exp_box_chan}
\boxtimes^3(\Lambda_1, \Lambda_2, \Lambda_3)(\cdot)
=\boxtimes_3\circ (\ot_i\Lambda_i) \circ \boxtimes^{-1}_3(\cdot) \;,
\end{eqnarray}
where  $\boxtimes^{-1}_3$ is the inverse of the convolutional channel in \eqref{eq:inver_conv_chn}.
\end{thm}
\begin{proof}
First, the Choi state of any channel $\Lambda$ can be rewritten in terms of Weyl operators as
\begin{eqnarray}\label{0127shi3}
J_\Lambda = \frac{1}{2^{2n}} \sum_{(\vec p,\vec q)\in V^n  } (-1)^{\vec p\cdot \vec q}w(\vec p,\vec q) \otimes \Lambda(w(\vec p,\vec q )) \;,
\end{eqnarray}
and thus
\begin{eqnarray*}
\boxtimes^3( J_{\Lambda_1}, J_{\Lambda_2}, J_{\Lambda_3})
=\frac{1}{2^{4n}}\sum_{\vec p,\vec q}  
w(\vec{p},\vec{q})
\ot (\boxtimes^3(
\Lambda_1(w( \vec p,  \vec q))
, \Lambda_2(w( \vec p,  \vec q)), \Lambda_3(w(\vec p, \vec q ))
) \;.
\end{eqnarray*}
Hence for any $n$-qubit state $\rho=\frac{1}{2^n} \sum_{\vec{p},\vec{q}}\Xi_{\rho}(\vec p, \vec q)w(\vec p, \vec q)$,
\begin{eqnarray*}
\boxtimes^3(\Lambda_1,\Lambda_2, \Lambda_3)(\rho)
&=&2^n\Ptr{1}{\boxtimes^K( J_{\Lambda_1}, J_{\Lambda_2}, J_{\Lambda_3}) \cdot \rho^T\ot I}\\
&=&
\sum_{\vec{p},\vec{q}} (-1)^{\vec p\cdot \vec q}\; \Xi_{\rho}(\vec p, \vec q)
\Ptr{1}{\boxtimes^3( J_{\Lambda_1}, J_{\Lambda_2}, J_{\Lambda_3}) \cdot w(\vec p, \vec q)\ot I}\\
&=&\frac{1}{2^{Kn}}
\sum_{\vec{p},\vec{q}} (-1)^{\vec p\cdot \vec q}\; \Xi_{\rho}(\vec p, \vec q)
\boxtimes^3(
\Lambda_1(w( \vec p,  \vec q))
, \Lambda_2(w( \vec p,  \vec q)), \Lambda_3(w(\vec p, \vec q )))\;.
\end{eqnarray*}
Moreover,
\begin{eqnarray*}
&&\boxtimes_3\circ (\Lambda_1\ot\Lambda_2\ot\Lambda_3) \circ\boxtimes_3^{-1}(\rho)\\
&=& \frac{1}{2^{3n}} \sum_{\vec{p},\vec{q}} \Xi_{\rho}(\vec p, \vec q)
\boxtimes_3\circ (\Lambda_1\ot\Lambda_2\ot\Lambda_3) \circ\boxtimes_3^\dag
(w(\vec p, \vec q))\\
&=&\frac{1}{2^{3n}} \sum_{\vec{p},\vec{q}} \Xi_{\rho}(\vec p, \vec q)
\boxtimes_3 \circ (\Lambda_1\ot\Lambda_2 \ot\Lambda_3) ((-1)^{N\vec p \cdot \vec q } w( \vec p, \vec q)^{\ot 3})\\
&=&\frac{1}{2^{3n}} \sum_{\vec{p},\vec{q}} (-1)^{N\vec p \cdot \vec q }\; \Xi_{\rho}(\vec p, \vec q)
\boxtimes_3(\Lambda_1(w(\vec p,\vec q))\ot \Lambda_2(w(\vec p,\vec q)) \ot \Lambda_3( w(\vec p, \vec q)))\\
&=&\frac{1}{2^{3n}} \sum_{\vec{p},\vec{q}} (-1)^{N\vec p \cdot \vec q } \Xi_{\rho}(\vec p, \vec q)
\boxtimes^3(
\Lambda_1(w(\vec p,  \vec q))
, \Lambda_2(w(\vec p,  \vec q)), \Lambda_3(w(\vec p, \vec q )) \;.
\end{eqnarray*}
Therefore, \eqref{eq:exp_box_chan} holds for any quantum state $\rho$.
\end{proof}

The completely depolarizing channel $\mathcal{R}$ on $n$-qubit systems is defined as follows:
\begin{eqnarray}\label{0212shi4}
\mathcal{R}(\rho) = \Tr{\rho} \frac{I_n}{2^n},\quad \forall \rho\;.
\end{eqnarray}
We can now present the following result.

\begin{prop}\label{prop:iden_pre}
Given two $n$-qubit channels $\Lambda_1, \Lambda_2$, 
\begin{align}
\boxtimes^3(\Lambda_1, \Lambda_2, \mathcal{R}) = \mathcal{R}\;.
\end{align}

\end{prop}

\begin{proof}
For any Pauli operator $w(\vec p,\vec q)$,
\begin{eqnarray*}
 \boxtimes^3(\Lambda_1 ,\Lambda_{2}, \mathcal{R})(w(\vec p, \vec q))
&=& \boxtimes_3\circ (\Lambda_1\ot  \Lambda_2 \ot \mathcal{R}) \circ \boxtimes_3^{-1}(w(\vec p, \vec q))
=\frac{1}{2^{2n}}\boxtimes_3\circ (\Lambda_1\ot  \Lambda_2  \ot  \mathcal{R})\circ\boxtimes_3^{\dag}(w(\vec p, \vec q))\\
&=&\frac{(-1)^{\vec p\cdot \vec q}}{2^{2n}}
\boxtimes_3\circ (\Lambda_1\ot  \Lambda_2 \ot \mathcal{R})(w( \vec p,  \vec q)\ot\cdots\ot w( \vec p, \vec q))\\
&=&\frac{(-1)^{\vec p\cdot \vec q}}{2^{2n}} \boxtimes_3 \left(\Lambda_1(w( \vec p, \vec q)) \ot\Lambda_{2}(w( \vec p, \vec q))
\ot I_n
\right)\delta_{\vec p,\vec 0}\delta_{\vec q, \vec 0}\\
&=& I_n
\delta_{\vec p,\vec 0}\delta_{\vec q, \vec 0}
=\mathcal{R}(w(\vec p, \vec q)) \;,
\end{eqnarray*}
where the fifth equality comes from Lemma \ref{lem:conv_iden}.
\end{proof}

\begin{prop}[\bf Convolutional stability for channels]
Given $3$ stabilizer channels $\Lambda_1,\Lambda_2, \Lambda_3$, the convolution
$ \boxtimes^3(\Lambda_1, \Lambda_2, \Lambda_3)$ is a stabilizer channel.
\end{prop}
\begin{proof}
Since  both $\boxtimes_3$ and $\boxtimes_3^{-1}$ are stabilizer channels, 
by Theorem
\ref{thm:exact_conv_chan},  $ \boxtimes^3(\Lambda_1, \Lambda_2, \Lambda_3)$
is a stabilizer channel.
\end{proof}

\begin{Def}
Let $\Lambda$ be an $n$-qubit channel,
the R\'{e}nyi entropy of $\Lambda$ is the R\'{e}nyi entropy of the  Choi state $J_{\Lambda}$,
that is
\begin{align*}
S_\alpha(\Lambda) = S_\alpha(J_\Lambda) \;.
\end{align*}
\end{Def}
As a consequence of Proposition \ref{prop:entropy} and the definition of the convolution of channels,
we have the following proposition.

\begin{prop}[\bf Convolution increases entropy of channels]\label{prop:entropy_chn}
Let $\Lambda_1, \Lambda_2, \Lambda_3$ be  three $n$-qubit channels and $\alpha\in [0,+\infty]$.
Then
\begin{align*}
S_\alpha(\boxtimes^3( \Lambda_1, \Lambda_2, \Lambda_3)) \ge  \max\set{S_\alpha(\Lambda_1), S_{\alpha}(\Lambda_2), S_\alpha(\Lambda_3)} \;.
\end{align*}
\end{prop}
Hence,   we have the following results on the self-convolution 
$\boxtimes^3U$ for a unitary $U$.

\begin{prop}\label{prop:pure_Clif}
Given a unitary channel $U$ , its self-convolution $\boxtimes^3U$
is  a unitary iff 
$U$ is Clifford.
\end{prop}

\begin{Def}[\bf Mean channel~\cite{BGJ23a,BGJ23b}]
Given a quantum channel  $\Lambda$, 
the mean channel 
 $\mathcal{M}(\Lambda)$  is  the quantum channel with the 
 Choi state $J_{\mathcal{M}(\Lambda)} = \mathcal{M}(J_\Lambda)$, where the $\mathcal{M}(J_\Lambda)$
 is the MS of the Choi state $J_{\Lambda}$.
\end{Def}

\begin{Exam}[Mean channel for the T gate]
Let us consider the $T$ gate,
$
    T=\left[
    \begin{array}{cc}
      1   &0  \\
       0  & e^{i\pi/4} 
    \end{array}
    \right]
$,
which is a 1-qubit non-Clifford gate. Then
$
J_{\mathcal{M}(T)}=\frac{1}{4}(I\ot I+Z\ot Z)
$.
Hence, the mean channel $\mathcal{M}(T)$ is the complete dephasing channel w.r.t. Pauli Z basis, i.e., 
\begin{eqnarray}
\mathcal{M}(T)(\rho)
=\sum_{x\in \set{0,1}}\bra{x}\rho\ket{x}\proj{ x}.
\end{eqnarray}
\end{Exam}

\begin{Def}[\cite{BGJ23a,BGJ23b}]
The magic gap of $\Lambda$ is  the magic gap of the Choi state $J_{\Lambda}$, i.e.,
\begin{eqnarray*}
MG(\Lambda)
:=MG(J_{\Lambda}) \;.
\end{eqnarray*}

\end{Def}

Given two $n$-qudit channels $\Lambda_1$ and $\Lambda_2$, the diamond distance \cite{Aharonov98} between $\Lambda_1$ and $\Lambda_2$ is
\begin{eqnarray*}
\norm{\Lambda_1-\Lambda_2}_{\diamond}
=\sup_{\rho\in \mathcal{D}(\mathcal{H}_S\ot \mathcal{H}_R)}\norm{(\Lambda_1-\Lambda_2)\ot id_R(\rho)}_1 \;,
\end{eqnarray*}
where $id_R$ is the identity mapping on the ancilla system.

\begin{thm}[\bf Central limit theorem for channels]\label{thm:CLT_chan}
Let $\Lambda$ be an  $n$-qubit channel with $K=2N+1$,
then
\begin{eqnarray}\label{230526shi1}
\norm{\boxtimes^{K}\Lambda-\mathcal{M}(\Lambda^\#)}_{\diamond}
\leq 2^{2n}(1-MG(\Lambda))^{K-1} \norm{J_{\Lambda}-J_{\CMM(\Lambda)}}_2 \;,
\end{eqnarray}
where $\Lambda^\# = \Lambda$ for even $N$, and
$\Lambda^\#$ is the channel whose Choi state is $(J_\Lambda)^T$
for odd $N$.
\end{thm}

\begin{proof}
We have the estimates
\begin{align*}
\norm{\boxtimes^{K} \Lambda-\CMM(\Lambda^\#)}_{\diamond}
\leq& 2^n \norm{J_{\boxtimes^{K}\left( \Lambda,...,\Lambda \right)}-J_{\CMM(\Lambda^\#)}}_1
=2^n \norm{\boxtimes^K (J_{\Lambda},..., J_{\Lambda})- \CMM(J_{\Lambda^\#})}_1\\
\leq& 2^{2n}\norm{\boxtimes_K (J_{\Lambda},..., J_{\Lambda}) -\CMM(J_{\Lambda^\#})}_2
\leq  2^{2n} (1-MG(\Lambda))^{K-1}\norm{J_{\Lambda} -J_{\CMM(\Lambda)}}_2 \;,
\end{align*}
where the first inequality comes from the fact that $\norm{ \Lambda_1-\Lambda_2}_{\diamond}\leq 2^n\norm{J_{\Lambda_1}-J_{\Lambda_2}}_1$ (see \cite{Watrous18}), 
the equality comes from the definition of mean channels,
and the next inequality comes from the fact that
$\norm{\cdot}_1\leq \sqrt{2^{2n}}\norm{\cdot}_2$.
The last inequality uses Theorem \ref{thm:CLT_1}.
\end{proof}

\subsection{Detailed proof of the results on stabilizer testing}\label{appen:stab_test}

\begin{mproof}[Proof of Theorem \ref{thm:Clif_test_d}]
We only need to prove that 
\begin{align*}
    (1-\epsilon)^4
    \leq \Tr{(J_U\boxtimes_H J_U)^2}
    \leq 1-4\epsilon+O(\epsilon^2).
\end{align*}
Since $\max_{V\in Cl_n}|\inner{V}{U}|^2=1-\epsilon$, then there exist some Clifford unitary $V$ such that 
$|\inner{V}{U}|^2=1-\epsilon$.  By the
commutativity of $\boxtimes_H$ with Clifford unitaries in Lemma~\ref{lem:prop_HCon}, we have
\begin{eqnarray*}
\Tr{(J_U\boxtimes_HJ_U)^2}
=\Tr{(J_{V^\dag U}\boxtimes_HJ_{V^\dag U})^2}.
\end{eqnarray*}

Let us define
\begin{eqnarray*}
\ket{w(\vec p, \vec q)}
=(w(\vec p, \vec q)\ot I)\ket{\text{Bell}}, \forall (\vec p, \vec q)\in V^n.
\end{eqnarray*}
This forms an orthonormal basis for  $2n$-qudit systems.
For any $2n$-qudit  $\rho$, let
\begin{eqnarray*}
\rho(\vec p, \vec q)
:=\bra{w(\vec p, \vec q)}\rho\ket{w(\vec p, \vec q)}.
\end{eqnarray*}
Then the Choi state $J_{V^\dag U}=V^\dag U\ot I\proj{\text{Bell}}U^\dag V\ot I$.
Let us denote $\ket{V^\dag U}:=V^\dag U\ot I\ket{\text{Bell}}$ for simplicity, then $J_{V^\dag U}=\proj{V^\dag U}$. 
Therefore,
\begin{eqnarray*}
|\iinner{w(\vec 0, \vec 0)}{V^\dag U}|^2
=\bra{\text{Bell}}J_{V^\dag U}\ket{\text{Bell}}
=\bra{\text{Bell}}V^\dag U\ot I\ket{\text{Bell}}
\bra{\text{Bell}}U^\dag V \ot I\ket{\text{Bell}}
=|\inner{V}{U}|^2=1-\epsilon.
\end{eqnarray*}

For the lower bound, we have
\begin{eqnarray*}
\Tr{(J_{V^\dag U}\boxtimes_HJ_{V^\dag U})^2}
&\geq& \bra{w(\vec 0, \vec 0)}J_{V^\dag U}\boxtimes_HJ_{V^\dag U} \ket{w(\vec 0, \vec 0)}^2\\
&=& \left(\sum_{\vec p, \vec q}J_{V^\dag U}(\vec p, \vec q)J_{V^\dag U}(-\vec p, -\vec q)\right)^2\\
&\geq& \left(J_{V^\dag U}(\vec 0, \vec 0)^2\right)^2\\
&\geq& (1-\epsilon)^4,
\end{eqnarray*}
where the first inequality comes from the Cauchy-Schwarz inequality, and the equality comes from Lemma~\ref{lem:unitray_Bell}.

For the upper bound, let us rewrite  $\ket{V^\dag U}$  as
\begin{align*}
    \ket{V^\dag U} =\sqrt{1-\epsilon}\ket{w(\vec 0,\vec 0)}
    +\sum_{(\vec p, \vec q)\neq (\vec 0, \vec 0)}
    \iinner{w(\vec p, \vec q)}{V^\dag U}\ket{w(\vec p, \vec q)},
\end{align*}
where $\sum_{(\vec p, \vec q)\neq (\vec 0, \vec 0)}| \iinner{w(\vec p, \vec q)}{V^\dag U}|^2=\epsilon$.
Let us take the stabilizer group of $\ket{w(\vec 0, \vec 0)}$, which is 
$G=\set{((\vec p, \vec q), (-\vec p, \vec q): \vec p, \vec q\in \mathbb{Z}^n_d)}$. 
Hence, $((\vec p, \vec q), (-\vec u, \vec v))\notin G$ iff 
$(\vec p, \vec q)\neq (\vec u, \vec v)$. 
In this case, it can be written as 
$((\vec p, \vec q), (-\vec u, \vec v))=((\vec u, \vec v), (-\vec u, \vec v))+((\vec p-\vec u, \vec q- \vec v), (\vec 0, \vec 0))$.
That is,
\begin{eqnarray*}
w(\vec p, \vec q)\ot w(-\vec u, \vec v)
=\chi(-2^{-1}\inner{(\vec u, \vec v)}{(\vec p-\vec u, \vec q-\vec v)}_s)w(\vec u, \vec v)\ot w(-\vec u, \vec v)(w(\vec p-\vec u, \vec q-\vec v)\ot w(\vec 0, \vec 0)).
\end{eqnarray*}
Since
\begin{eqnarray*}
w(\vec p, \vec q)\ot w(-\vec p, \vec q)
=\sum_{\vec x, \vec y}\proj{w(\vec x, \vec y)}\omega^{-\inner{(\vec x, \vec y)}{(\vec p, \vec q))}_s}_d,
\end{eqnarray*}
then 
\begin{eqnarray*}
&&w(\vec p, \vec q)\ot w(-\vec u, \vec v)\\
&=&\chi(-2^{-1}\inner{(\vec u, \vec v)}{(\vec p-\vec u, \vec q-\vec v)}_s)
\sum_{\vec x, \vec y}\ket{w(\vec x, \vec y)}\bra{w(\vec x-\vec p+\vec u, \vec y-\vec q+\vec v)}\omega^{-\inner{(\vec x, \vec y)}{(\vec u, \vec v))}_s}_d \chi(-2^{-1} \inner{(\vec x, \vec y)}{(\vec p-\vec u, \vec q-\vec v)}_s).
\end{eqnarray*}

Then, when $((\vec p, \vec q), (\vec u,\vec v))\notin G$, 
\begin{align*}
 &\left| \Tr{J_{V^\dag U}w(\vec p, \vec q)\ot w(-\vec u, \vec v)}\right|\\
 \leq& \sum_{\vec x, \vec y}\left|\Tr{J_{V^\dag U}\ket{w(\vec x, \vec y)}\bra{w(\vec x-\vec p+\vec u,  \vec y-\vec q+\vec v)}}\right|\\
 =&  \sum_{\vec x, \vec y}\left|\iinner{w(\vec x, \vec y)}{V^\dag U}\right|
 \left|\iinner{w(\vec x-\vec p+\vec u,  \vec y-\vec q+\vec v)}{V^\dag U}\right|\\
 =& \left|\iinner{w(\vec 0, \vec 0)}{V^\dag U}\right|
 \left|\iinner{w(-\vec p+\vec u,  -\vec q+ \vec v)}{V^\dag U}\right|+
 \left|\iinner{w(\vec 0, \vec 0)}{V^\dag U}\right|
 \left|\iinner{w(\vec p-\vec u,  \vec q- \vec v)}{V^\dag U}\right|\\
 &  +\sum_{(\vec x, \vec y)\neq (\vec 0,\vec 0), (\vec p-\vec u, \vec q-\vec v)}\left|\iinner{w(\vec x, \vec y)}{V^\dag U}\right|
 \left|\iinner{w(\vec x-\vec p+\vec u,  \vec y-\vec q+\vec v)}{V^\dag U}\right|
 \\
 \leq&2\sqrt{(1-\epsilon)\epsilon}+\epsilon.
\end{align*}
That is, 
\begin{eqnarray*}
\max_{((\vec p, \vec q), (\vec u,\vec v))\notin G}
\left|\Xi_{J_{V^\dag U}}((\vec p, \vec q), (\vec u, \vec v))\right|\leq 
2\sqrt{(1-\epsilon)\epsilon}+\epsilon.
\end{eqnarray*}

Moreover, 
\begin{align*}
1=&\Tr{J^2_{V^\dag U}}
=\frac{1}{d^{2n}}
\sum_{(\vec p, \vec q), (\vec u, \vec v)}
\left|\Xi_{J_{V^\dag U}}((\vec p, \vec q), (\vec u, \vec v))\right|^2\\
=&\frac{1}{d^{2n}}
\sum_{((\vec p, \vec q), (\vec u, \vec v))\in G}
\left|\Xi_{J_{V^\dag U}}((\vec p, \vec q), (\vec u, \vec v))\right|^2
+\frac{1}{d^{2n}}
\sum_{((\vec p, \vec q), (\vec u, \vec v))\notin G}
\left|\Xi_{J_{V^\dag U}}((\vec p, \vec q), (\vec u, \vec v))\right|^2,
\end{align*}
where 
\begin{eqnarray*}
&&\frac{1}{d^{2n}}
\sum_{((\vec p, \vec q), (\vec u, \vec v))\in G}
\left|\Xi_{J_{V^\dag U}}((\vec p, \vec q), (\vec u, \vec v))\right|^2\\
&=&\frac{1}{d^{2n}}
\sum_{(\vec p, \vec q)\in V^n}
\left|\Xi_{J_{V^\dag U}}((\vec p, \vec q), (-\vec p, \vec q))\right|^2\\
&=&\frac{1}{d^{2n}}
\sum_{(\vec p, \vec q)\in V^n}
\left|\sum_{\vec x, \vec y} J_{V^\dag U}(\vec x, \vec y)\omega^{-\inner{(\vec x, \vec y)}{(\vec p, \vec q))}_s}_d\right|^2\\
&=&(1-\epsilon)^2+\sum_{(\vec p, \vec q)\neq(\vec 0,\vec 0)}
J_{V^\dag U}(\vec p, \vec q)^2.
\end{eqnarray*}
Hence, we have 
\begin{eqnarray*}
\frac{1}{d^{2n}}
\sum_{((\vec p, \vec q), (\vec u, \vec v))\notin G}
\left|\Xi_{J_{V^\dag U}}((\vec p, \vec q), (\vec u, \vec v))\right|^2
\leq 1-(1-\epsilon)^2.
\end{eqnarray*}
Then 
\begin{eqnarray*}
&&\Tr{(J_{V^\dag U}\boxtimes_H J_{V^\dag U})^2}\\
&=&\frac{1}{d^{2n}}
\sum_{(\vec p, \vec q), (\vec u, \vec v)}
\left|\Xi_{J_{V^\dag U}}((\vec p, \vec q), (\vec u, \vec v))\right|^4\\
&=&\frac{1}{d^{2n}}
\sum_{((\vec p, \vec q), (\vec u, \vec v))\in G}
\left|\Xi_{J_{V^\dag U}}((\vec p, \vec q), (\vec u, \vec v))\right|^4
+\frac{1}{d^{2n}}
\sum_{((\vec p, \vec q), (\vec u, \vec v))\notin G}
\left|\Xi_{J_{V^\dag U}}((\vec p, \vec q), (\vec u, \vec v))\right|^4,
\end{eqnarray*}
where 
\begin{eqnarray*}
\frac{1}{d^{2n}}
\sum_{((\vec p, \vec q), (\vec u, \vec v))\in G}
\left|\Xi_{J_{V^\dag U}}((\vec p, \vec q), (\vec u, \vec v))\right|^4
\leq 1-4\epsilon(1-\epsilon)^3,
\end{eqnarray*}
and 
\begin{eqnarray*}
&&\frac{1}{d^{2n}}
\sum_{((\vec p, \vec q), (\vec u, \vec v))\notin G}
\left|\Xi_{J_{V^\dag U}}((\vec p, \vec q), (\vec u, \vec v))\right|^4\\
&\leq& \left[2\sqrt{(1-\epsilon)\epsilon}+\epsilon\right]^2
\frac{1}{d^{2n}}
\sum_{((\vec p, \vec q), (\vec u, \vec v))\notin G}
\left|\Xi_{J_{V^\dag U}}((\vec p, \vec q), (\vec u, \vec v))\right|^2\\
&\leq& \left[2\sqrt{(1-\epsilon)\epsilon}+\epsilon\right]^2
\cdot\left[1-(1-\epsilon)^2\right]\\
&=&8\epsilon^2+o(\epsilon^2).
\end{eqnarray*}
Hence 
\begin{eqnarray*}
\Tr{(J_{V^\dag U}\boxtimes_H J_{V^\dag U})^2}
\leq 1-4\epsilon+O(\epsilon^2).
\end{eqnarray*}

\end{mproof}

\begin{lem}\label{lem:unitray_Bell}
For any $2n$-qudit states $\rho,\sigma$, we have 
\begin{eqnarray}
\rho\boxtimes\sigma(\vec p, \vec q)
=\sum_{(\vec a,\vec b)+(\vec c, \vec d)=(\vec p,2\vec q)}
\rho(\vec a, \vec b)
\sigma(\vec c, \vec d), \quad \forall (\vec p, \vec q)\in V^n.
\end{eqnarray}
\end{lem}
\begin{proof}
First, 
\begin{eqnarray*}
V_H\ket{w(\vec a, \vec b)}\ot \ket{w(\vec p, \vec q)}
&=&V_H w(\vec a,\vec b)\ot w(\vec p, \vec q)\ket{\text{Bell}}\ot\ket{\text{Bell}}\\
&=&w(\vec a+\vec p, 2^{-1}(\vec b+\vec q))\ot w(\vec a-\vec p, 2^{-1}(\vec b-\vec q))V_H\ket{\text{Bell}}\ot \ket{\text{Bell}}\\
&=&w(\vec a+\vec p, 2^{-1}(\vec b+\vec q))\ot w(\vec a-\vec p, 2^{-1}(\vec b-\vec q))\ket{\text{Bell}}\ot \ket{\text{Bell}}\\
&=&\ket{w(\vec a+\vec p, 2^{-1}(\vec b+\vec q))}\ot \ket{w(\vec a-\vec p, 2^{-1}(\vec b-\vec q))},
\end{eqnarray*}
where the second equality comes from the following equality proved in \cite{BGJ23a,BGJ23b},
\begin{eqnarray}
V_H w(\vec{p}_1,\vec{q}_1)\ot w(\vec{p}_2, \vec{p}_2)V^\dag_H
=w(\vec{p}_1+\vec{p}_2,2^{-1}\vec{q}_1+2^{-1}\vec{q}_2)
\ot w(\vec{p}_1-\vec{p}_2, 2^{-1}\vec{q}_1-2^{-1}\vec{q}_2),
\end{eqnarray}
and the third
equality is because
\begin{eqnarray*}
V_H\ket{\text{Bell}}\ot \ket{\text{Bell}}&=&\frac{1}{d^n}V_H\sum_{\vec x,\vec y}\ket{\vec x}\ket{\vec x}\ket{\vec y}\ket{\vec y}\\
&=&\frac{1}{d^n}\sum_{\vec x, \vec y}\ket{2^{-1}(\vec x+\vec y)}\ket{2^{-1}(\vec x+\vec y)}\ket{2^{-1}(\vec x-\vec y)}\ket{2^{-1}(\vec x-\vec y)}\\
&=&\frac{1}{d^n}\sum_{\vec x,\vec y}\ket{\vec x}\ket{\vec x}\ket{\vec y}\ket{\vec y}\\
&=&\ket{\text{Bell}}\ot \ket{\text{Bell}}.
\end{eqnarray*}
Similarly, we can also prove that 
\begin{eqnarray*}
V^\dag_H\ket{w(\vec a, \vec b)}\ot \ket{w(\vec p, \vec q)}
=\ket{w(2^{-1}(\vec a+\vec p), \vec b+\vec q)}\ot \ket{w(2^{-1}(\vec a-\vec p), \vec b-\vec q}.
\end{eqnarray*}

Hence
\begin{eqnarray*}
\rho\boxtimes\sigma(\vec p, \vec q)
&=&\Tr{(V_H\rho\ot \sigma V^\dag_H)(\proj{w(\vec p, \vec q)}\ot I) }\\
&=&\Tr{\rho\ot \sigma (V^\dag_H \proj{w(\vec p, \vec q)}\ot I V_H)}\\
&=&\sum_{\vec c, \vec d}\Tr{\rho\ot \sigma (V^\dag_H \proj{w(\vec p, \vec q)}\ot \proj{w(\vec c, \vec d)}V_H)}\\
&=& \sum_{\vec c, \vec d}
\rho(2^{-1}(\vec p+\vec c), \vec q+\vec d)
\sigma(2^{-1}(\vec p- \vec c), \vec q-\vec d)\\
&=&
\sum_{(\vec a,\vec b)+(\vec c, \vec d)=(\vec p,2\vec q)}
\rho(\vec a, \vec b)
\sigma(\vec c, \vec d).
\end{eqnarray*}

\end{proof}

\begin{mproof}[Proof of Theorem \ref{thm:Clif_test_2}]

We only need to prove that 
\begin{align*}
    (1-\epsilon)^6
    \leq \Tr{(\boxtimes_3J_U)^2}
    \leq 1-6\epsilon+O(\epsilon^2).
\end{align*}
Since $\max_{V\in Cl_n}|\inner{V}{U}|^2=1-\epsilon$, then there exist some Clifford unitary $V$ such that 
$|\inner{V}{U}|^2=1-\epsilon$.  By the
commutativity of $\boxtimes_H$ with Clifford unitaries in Theorem~\ref{thm:commu_clif}, we have
\begin{eqnarray*}
\Tr{(\boxtimes_3J_U)^2}
=\Tr{(\boxtimes_3J_{V^\dag U})^2}.
\end{eqnarray*}

Same as in the proof of Theorem \ref{thm:Clif_test_d}, 
we denote
\begin{eqnarray*}
\ket{w(\vec p, \vec q)}
=w(\vec p, \vec q)\ot I\ket{\text{Bell}}, \forall (\vec p, \vec q)\in V^n.
\end{eqnarray*}
They form an orthonormal basis of the $2n$-qubit systems.
For any $2n$-qubit state $\rho$ we denote
\begin{eqnarray*}
\rho(\vec p, \vec q)
=\bra{w(\vec p, \vec q)}\rho\ket{w(\vec p, \vec q)}.
\end{eqnarray*}
The Choi state $J_{V^\dag U}=V^\dag U\ot I\proj{\text{Bell}}U^\dag V \ot I$.
Let us denote $\ket{V^\dag U}=V^\dag U\ot I\ket{\text{Bell}}$ for simplicity, then $J_{V^\dag U}=\proj{V^\dag U}$. 
Therefore,
\begin{eqnarray*}
|\iinner{w(\vec 0, \vec 0)}{V^\dag U}|^2
=\bra{\text{Bell}}J_{V^\dag U}\ket{\text{Bell}}
=\bra{\text{Bell}}V^\dag U\ket{\text{Bell}}
\bra{\text{Bell}}U^\dag V\ket{\text{Bell}}
=|\inner{V}{U}|^2=1-\epsilon.
\end{eqnarray*}

For the lower bound, we have
\begin{eqnarray*}
\Tr{(\boxtimes_3J_{V^\dag U})^2}
&\geq& \bra{w(\vec 0, \vec 0)}\boxtimes_3J_{V^\dag U} \ket{w(\vec 0, \vec 0)}^2\\
&=& \left(\sum_{\sum_i(\vec p_i,\vec q_i)=(\vec 0, \vec 0)}J_{V^\dag U}(\vec p_1, \vec q_1)J_{V^\dag U}(\vec p_2, \vec q_2) J_{V^\dag U}(\vec p_3, \vec q_3)\right)^2\\
&\geq& \left(J_{V^\dag U}(\vec 0, \vec 0)^3\right)^2\\
&\geq& (1-\epsilon)^6,
\end{eqnarray*}
where the first inequality comes from the Cauchy-Schwarz inequality, and the equality comes from Lemma~\ref{lem:qubit_uni}.

For the upper bound, let us rewrite  $\ket{V^\dag U}$  as
\begin{align*}
    \ket{V^\dag U} =\sqrt{1-\epsilon}\ket{w(\vec 0,\vec 0)}
    +\sum_{(\vec p, \vec q)\neq (\vec 0, \vec 0)}
    \iinner{w(\vec p, \vec q)}{V^\dag U}\ket{w(\vec p, \vec q)},
\end{align*}
where $\sum_{(\vec p, \vec q)\neq (\vec 0, \vec 0)}| \iinner{w(\vec p, \vec q)}{V^\dag U}|^2=\epsilon$.
Let us take the stabilizer group of $\ket{w(\vec 0, \vec 0)}$, which is 
$G=\set{((\vec p, \vec q), (\vec p, \vec q): \vec p, \vec q\in \mathbb{Z}^n_2)}$. 
Hence, $((\vec p, \vec q), (\vec u, \vec v))\notin G$ iff 
$(\vec p, \vec q)\neq (\vec u, \vec v)$. 
In this case, it can be written as 
$((\vec p, \vec q), (\vec u, \vec v))=((\vec u, \vec v), (\vec u, \vec v))+((\vec p-\vec u, \vec q- \vec v), (\vec 0, \vec 0))$.
That is,
\begin{eqnarray*}
w(\vec p, \vec q)\ot w(\vec u, \vec v)
=i^{-\inner{(\vec u, \vec v)}{(\vec p-\vec u, \vec q-\vec v)}}[w(\vec u, \vec v)\ot w(\vec u, \vec v)][w(\vec p-\vec u, \vec q-\vec v)\ot w(\vec 0, \vec 0)].
\end{eqnarray*}
Since
\begin{eqnarray*}
(-1)^{\vec p\cdot \vec q}w(\vec p, \vec q)\ot w(\vec p, \vec q)
=\sum_{\vec x, \vec y}\proj{w(\vec x, \vec y)}(-1)^{\inner{(\vec x, \vec y)}{(\vec p, \vec q))}_s}
\end{eqnarray*}
then 
\begin{eqnarray*}
&&w(\vec p, \vec q)\ot w(\vec u, \vec v)\\
&=&i^{-\inner{(\vec u, \vec v)}{(\vec p-\vec u, \vec q-\vec v)}}(-1)^{\vec p\cdot \vec q}
\sum_{\vec x, \vec y}\ket{w(\vec x, \vec y)}\bra{w(\vec x-\vec p+\vec u, \vec y-\vec q+\vec v)}(-1)^{\inner{(\vec x, \vec y)}{(\vec u, \vec v))}_s}i^{-\inner{(\vec x, \vec y)}{(\vec p-\vec u, \vec q-\vec v)}}.
\end{eqnarray*}

Then 
\begin{align*}
 &\left|\Tr{J_{V^\dag U}w(\vec p, \vec q)\ot w(\vec u, \vec v)}\right|\\
 \leq& \sum_{\vec x, \vec y}\left|\Tr{J_{V^\dag U}\ket{w(\vec x, \vec y)}\bra{w(\vec x-\vec p+\vec u,  \vec y-\vec q+\vec v)}}\right|\\
 =&  \sum_{\vec x, \vec y}\left|\iinner{w(\vec x, \vec y)}{V^\dag U}\right|
 \left|\iinner{w(\vec x-\vec p+\vec u,  \vec y-\vec q+\vec v)}{V^\dag U}\right|\\
  =&2\left|\iinner{w(\vec 0, \vec 0)}{V^\dag U}\right|
 \left|\iinner{w(-\vec p+\vec u,  -\vec q+\vec v)}{V^\dag U}\right|
 +\sum_{(\vec x, \vec y)\neq (\vec 0,\vec 0), (\vec p-\vec u, \vec q-\vec v)}\left|\iinner{w(\vec x, \vec y)}{V^\dag U}\right|
 \left|\iinner{w(\vec x-\vec p+\vec u,  \vec y-\vec q+\vec v)}{V^\dag U}\right|
 \\
 \leq&2\sqrt{(1-\epsilon)\epsilon}+\epsilon.
\end{align*}
That is, 
\begin{eqnarray*}
\max_{((\vec p, \vec q), (\vec u,\vec v))\notin G}
\left|\Xi_{J_{V^\dag U}}((\vec p, \vec q), (\vec u, \vec v))\right|\leq 
2\sqrt{(1-\epsilon)\epsilon}+\epsilon.
\end{eqnarray*}

Moreover, 
\begin{align*}
1=&\Tr{J^2_{V^\dag U}}
=\frac{1}{2^{2n}}
\sum_{(\vec p, \vec q), (\vec u, \vec v)}
\left|\Xi_{J_{V^\dag U}}((\vec p, \vec q), (\vec u, \vec v))\right|^2\\
=&\frac{1}{2^{2n}}
\sum_{((\vec p, \vec q), (\vec u, \vec v))\in G}
\left|\Xi_{J_{V^\dag U}}((\vec p, \vec q), (\vec u, \vec v))\right|^2
+\frac{1}{2^{2n}}
\sum_{((\vec p, \vec q), (\vec u, \vec v))\notin G}
\left|\Xi_{J_{V^\dag U}}((\vec p, \vec q), (\vec u, \vec v))\right|^2,
\end{align*}
where 
\begin{eqnarray*}
&&\frac{1}{2^{2n}}
\sum_{((\vec p, \vec q), (\vec u, \vec v))\in G}
\left|\Xi_{J_{V^\dag U}}((\vec p, \vec q), (\vec u, \vec v))\right|^2\\
&=&\frac{1}{2^{2n}}
\sum_{(\vec p, \vec q)\in V^n}
\left|\Xi_{J_{V^\dag U}}((\vec p, \vec q), (\vec p, \vec q))\right|^2\\
&=&\frac{1}{2^{2n}}
\sum_{(\vec p, \vec q)\in V^n}
\left|\sum_{\vec x, \vec y} J_{V^\dag U}(\vec x, \vec y)(-1)^{\inner{(\vec x, \vec y)}{(\vec p, \vec q))}_s}\right|^2\\
&=&(1-\epsilon)^2+\sum_{(\vec p, \vec q)\neq(\vec 0,\vec 0)}
J_{V^\dag U}(\vec p, \vec q)^2.
\end{eqnarray*}
Hence, we have 
\begin{eqnarray*}
\frac{1}{2^{2n}}
\sum_{((\vec p, \vec q), (\vec u, \vec v))\notin G}
\left|\Xi_{J_{V^\dag U}}(\vec p, \vec q), (\vec u, \vec v))\right|^2
\leq 1-(1-\epsilon)^2.
\end{eqnarray*}
Then 
\begin{eqnarray*}
&&\Tr{(\boxtimes_3 J_{V^\dag U})^2}\\
&=&\frac{1}{2^{2n}}
\sum_{(\vec p, \vec q), (\vec u, \vec v)}
\left|\Xi_{J_{V^\dag U}}((\vec p, \vec q), (\vec u, \vec v))\right|^6\\
&=&\frac{1}{2^{2n}}
\sum_{((\vec p, \vec q), (\vec u, \vec v))\in G}
\left|\Xi_{J_{V^\dag U}}((\vec p, \vec q), (\vec u, \vec v))\right|^6
+\frac{1}{2^{2n}}
\sum_{((\vec p, \vec q), (\vec u, \vec v))\notin G}
\left|\Xi_{J_{V^\dag U}}((\vec p, \vec q), (\vec u, \vec v))\right|^6,
\end{eqnarray*}
where 
\begin{eqnarray*}
\frac{1}{2^{2n}}
\sum_{((\vec p, \vec q), (\vec u, \vec v))\in G}
\left|\Xi_{J_{V^\dag U}}((\vec p, \vec q), (\vec u, \vec v))\right|^6
\leq 1-6\epsilon(1-\epsilon)^5,
\end{eqnarray*}
and 
\begin{eqnarray*}
&&\frac{1}{d^{2n}}
\sum_{((\vec p, \vec q), (\vec u, \vec v))\notin G}
\left|\Xi_{J_{V^\dag U}}((\vec p, \vec q), (\vec u, \vec v))\right|^6\\
&\leq& \left[2\sqrt{(1-\epsilon)\epsilon}+\epsilon\right]^4
\frac{1}{2^{2n}}
\sum_{((\vec p, \vec q), (\vec u, \vec v))\notin G}
\left|\Xi_{J_{V^\dag U}}((\vec p, \vec q), (\vec u, \vec v))\right|^2\\
&\leq& \left[2\sqrt{(1-\epsilon)\epsilon}+\epsilon\right]^4
\cdot\left[1-(1-\epsilon)^2\right]\\
&=&32\epsilon^3+o(\epsilon^2).
\end{eqnarray*}
Hence 
\begin{eqnarray*}
\Tr{(\boxtimes_3 J_{V^\dag U})^2}
\leq 1-6\epsilon+O(\epsilon^2).
\end{eqnarray*}
\end{mproof}

\begin{lem}\label{lem:qubit_uni}
For any $2n$-qubit states $\set{\rho_i}^K_i$ with odd $K$, we have
\begin{eqnarray}
\boxtimes_K\left(\ot^K_{i=1}\rho_i\right)(\vec p, \vec q)
=\sum_{(\vec p_i, \vec q_i)^K_{i=1}:\sum_i(\vec p_i, \vec q_i)=(\vec p, \vec q)}
\prod^K_{i=1}\rho_i (\vec p_i, \vec q_i).
\end{eqnarray}
\end{lem}
\begin{proof}
\begin{align*}
\boxtimes_K\left(\ot^K_{i=1}\rho_i\right)(\vec p, \vec q)
=&\Tr{\boxtimes_K(\ot^K_{i=1}\rho_i)\proj{w(\vec p, \vec q)}}
=\Tr{\ot^K_{i=1}\rho_iV^\dag(\proj{w(\vec p, \vec q)}\ot^K_{i=2}I)V}\\
=&\sum_{(\vec p_i, \vec q_i)^K_{i=2}}\Tr{\ot^K_{i=1}\rho_iV^\dag\proj{w(\vec p, \vec q)}\ot^K_{i=2}\proj{w(\vec p_i, \vec q_i)}V}\\
=& \sum_{(\vec p_i, \vec q_i)^K_{i=2}}
\rho_1\left(\vec p+\sum^K_{j=2}\vec p_j, \vec q+\sum^K_{j=2}\vec q_j\right)
\ot^K_{i=2}\rho_i\left(\vec p-\vec p_i, \vec q+\sum^K_{j=2}\vec q_j-\vec q_i\right) \\
=&\sum_{(\vec p_i, \vec q_i)^K_{i=1}:\sum_i(\vec p_i, \vec q_i)=(\vec p, \vec q)}
\prod^K_{i=1}\rho_i (\vec p_i, \vec q_i),
\end{align*}
where the fourth equality comes from Proposition~\ref{prop:action_key} and the fact that 
\begin{eqnarray*}
V^\dag\ket{\text{Bell}}^{\ot K}
=\frac{1}{d^{nK/2}}V^\dag \sum_{\set{\vec x_i}^K_{i=1}}
\ot^K_{i=1}\ket{\vec x_i}\ket{\vec x_i}
=\frac{1}{d^{nK/2}} \sum_{\set{\vec x_i}^K_{i=1}}\ket{\sum^K_{i=1}\vec x_i}\ket{\sum^K_{i=1}\vec x_i}\ot^K_{i=2}
\ket{\sum_{j\neq i} 
\vec x_j}\ket{\sum_{j\neq i} 
\vec x_j}
=\ket{\text{Bell}}^{\ot K}.
\end{eqnarray*}
\end{proof}

\section*{Data Availability}
Data sharing is not applicable to this article as no datasets were generated or analysed during the current study.

\section*{Conflicts of Interest}
The authors have no relevant financial or non-financial interests to disclose.

\bibliography{reference}{}

\begin{thebibliography}{131}%
\makeatletter
\providecommand \@ifxundefined [1]{%
 \@ifx{#1\undefined}
}%
\providecommand \@ifnum [1]{%
 \ifnum #1\expandafter \@firstoftwo
 \else \expandafter \@secondoftwo
 \fi
}%
\providecommand \@ifx [1]{%
 \ifx #1\expandafter \@firstoftwo
 \else \expandafter \@secondoftwo
 \fi
}%
\providecommand \natexlab [1]{#1}%
\providecommand \enquote  [1]{``#1''}%
\providecommand \bibnamefont  [1]{#1}%
\providecommand \bibfnamefont [1]{#1}%
\providecommand \citenamefont [1]{#1}%
\providecommand \href@noop [0]{\@secondoftwo}%
\providecommand \href [0]{\begingroup \@sanitize@url \@href}%
\providecommand \@href[1]{\@@startlink{#1}\@@href}%
\providecommand \@@href[1]{\endgroup#1\@@endlink}%
\providecommand \@sanitize@url [0]{\catcode `\\12\catcode `\$12\catcode
  `\&12\catcode `\#12\catcode `\^12\catcode `\_12\catcode `\%12\relax}%
\providecommand \@@startlink[1]{}%
\providecommand \@@endlink[0]{}%
\providecommand \url  [0]{\begingroup\@sanitize@url \@url }%
\providecommand \@url [1]{\endgroup\@href {#1}{\urlprefix }}%
\providecommand \urlprefix  [0]{URL }%
\providecommand \Eprint [0]{\href }%
\providecommand \doibase [0]{http://dx.doi.org/}%
\providecommand \selectlanguage [0]{\@gobble}%
\providecommand \bibinfo  [0]{\@secondoftwo}%
\providecommand \bibfield  [0]{\@secondoftwo}%
\providecommand \translation [1]{[#1]}%
\providecommand \BibitemOpen [0]{}%
\providecommand \bibitemStop [0]{}%
\providecommand \bibitemNoStop [0]{.\EOS\space}%
\providecommand \EOS [0]{\spacefactor3000\relax}%
\providecommand \BibitemShut  [1]{\csname bibitem#1\endcsname}%
\let\auto@bib@innerbib\@empty
\bibitem [{\citenamefont {Goldreich}(2017)}]{goldreich2017introduction}%
  \BibitemOpen
  \bibfield  {author} {\bibinfo {author} {\bibfnamefont {O.}~\bibnamefont
  {Goldreich}},\ }\href@noop {} {\emph {\bibinfo {title} {Introduction to
  property testing}}}\ (\bibinfo  {publisher} {Cambridge University Press},\
  \bibinfo {year} {2017})\BibitemShut {NoStop}%
\bibitem [{\citenamefont {Blum}\ \emph {et~al.}(1993)\citenamefont {Blum},
  \citenamefont {Luby},\ and\ \citenamefont {Rubinfeld}}]{blum1990self}%
  \BibitemOpen
  \bibfield  {author} {\bibinfo {author} {\bibfnamefont {M.}~\bibnamefont
  {Blum}}, \bibinfo {author} {\bibfnamefont {M.}~\bibnamefont {Luby}}, \ and\
  \bibinfo {author} {\bibfnamefont {R.}~\bibnamefont {Rubinfeld}},\ }\href
  {\doibase 10.1016/0022-0000(93)90044-W} {\bibfield  {journal} {\bibinfo
  {journal} {Journal of Computer and System Sciences}\ }\textbf {\bibinfo
  {volume} {47}},\ \bibinfo {pages} {549} (\bibinfo {year} {1993})}\BibitemShut
  {NoStop}%
\bibitem [{\citenamefont {Alon}\ \emph {et~al.}(2005)\citenamefont {Alon},
  \citenamefont {Kaufman}, \citenamefont {Krivelevich}, \citenamefont
  {Litsyn},\ and\ \citenamefont {Ron}}]{Alon05}%
  \BibitemOpen
  \bibfield  {author} {\bibinfo {author} {\bibfnamefont {N.}~\bibnamefont
  {Alon}}, \bibinfo {author} {\bibfnamefont {T.}~\bibnamefont {Kaufman}},
  \bibinfo {author} {\bibfnamefont {M.}~\bibnamefont {Krivelevich}}, \bibinfo
  {author} {\bibfnamefont {S.}~\bibnamefont {Litsyn}}, \ and\ \bibinfo {author}
  {\bibfnamefont {D.}~\bibnamefont {Ron}},\ }\href {\doibase
  10.1109/TIT.2005.856958} {\bibfield  {journal} {\bibinfo  {journal} {IEEE
  Transactions on Information Theory}\ }\textbf {\bibinfo {volume} {51}},\
  \bibinfo {pages} {4032} (\bibinfo {year} {2005})}\BibitemShut {NoStop}%
\bibitem [{\citenamefont {Bhattacharyya}\ \emph {et~al.}(2010)\citenamefont
  {Bhattacharyya}, \citenamefont {Kopparty}, \citenamefont {Schoenebeck},
  \citenamefont {Sudan},\ and\ \citenamefont
  {Zuckerman}}]{bhattacharyya2010optimal}%
  \BibitemOpen
  \bibfield  {author} {\bibinfo {author} {\bibfnamefont {A.}~\bibnamefont
  {Bhattacharyya}}, \bibinfo {author} {\bibfnamefont {S.}~\bibnamefont
  {Kopparty}}, \bibinfo {author} {\bibfnamefont {G.}~\bibnamefont
  {Schoenebeck}}, \bibinfo {author} {\bibfnamefont {M.}~\bibnamefont {Sudan}},
  \ and\ \bibinfo {author} {\bibfnamefont {D.}~\bibnamefont {Zuckerman}},\ }in\
  \href {\doibase 10.1109/FOCS.2010.54} {\emph {\bibinfo {booktitle} {2010 IEEE
  51st Annual Symposium on Foundations of Computer Science}}}\ (\bibinfo
  {organization} {IEEE},\ \bibinfo {year} {2010})\ pp.\ \bibinfo {pages}
  {488--497}\BibitemShut {NoStop}%
\bibitem [{\citenamefont {Babai}\ \emph
  {et~al.}(1991{\natexlab{a}})\citenamefont {Babai}, \citenamefont {Fortnow},
  \citenamefont {Levin},\ and\ \citenamefont {Szegedy}}]{babai1991checking}%
  \BibitemOpen
  \bibfield  {author} {\bibinfo {author} {\bibfnamefont {L.}~\bibnamefont
  {Babai}}, \bibinfo {author} {\bibfnamefont {L.}~\bibnamefont {Fortnow}},
  \bibinfo {author} {\bibfnamefont {L.~A.}\ \bibnamefont {Levin}}, \ and\
  \bibinfo {author} {\bibfnamefont {M.}~\bibnamefont {Szegedy}},\ }in\ \href
  {\doibase 10.1145/103418.103428} {\emph {\bibinfo {booktitle} {Proceedings of
  the Twenty-Third Annual ACM Symposium on Theory of Computing}}},\ \bibinfo
  {series and number} {STOC '91}\ (\bibinfo  {publisher} {Association for
  Computing Machinery},\ \bibinfo {address} {New York, NY, USA},\ \bibinfo
  {year} {1991})\ p.\ \bibinfo {pages} {21–32}\BibitemShut {NoStop}%
\bibitem [{\citenamefont {Babai}\ \emph
  {et~al.}(1991{\natexlab{b}})\citenamefont {Babai}, \citenamefont {Fortnow},\
  and\ \citenamefont {Lund}}]{babai1991non}%
  \BibitemOpen
  \bibfield  {author} {\bibinfo {author} {\bibfnamefont {L.}~\bibnamefont
  {Babai}}, \bibinfo {author} {\bibfnamefont {L.}~\bibnamefont {Fortnow}}, \
  and\ \bibinfo {author} {\bibfnamefont {C.}~\bibnamefont {Lund}},\ }\href
  {\doibase 10.1007/BF01200056} {\bibfield  {journal} {\bibinfo  {journal}
  {Computational complexity}\ }\textbf {\bibinfo {volume} {1}},\ \bibinfo
  {pages} {3} (\bibinfo {year} {1991}{\natexlab{b}})}\BibitemShut {NoStop}%
\bibitem [{\citenamefont {Feige}\ \emph {et~al.}(1991)\citenamefont {Feige},
  \citenamefont {Goldwasser}, \citenamefont {Lovasz}, \citenamefont {Safra},\
  and\ \citenamefont {Szegedy}}]{feige1991approximating}%
  \BibitemOpen
  \bibfield  {author} {\bibinfo {author} {\bibfnamefont {U.}~\bibnamefont
  {Feige}}, \bibinfo {author} {\bibfnamefont {S.}~\bibnamefont {Goldwasser}},
  \bibinfo {author} {\bibfnamefont {L.}~\bibnamefont {Lovasz}}, \bibinfo
  {author} {\bibfnamefont {S.}~\bibnamefont {Safra}}, \ and\ \bibinfo {author}
  {\bibfnamefont {M.}~\bibnamefont {Szegedy}},\ }in\ \href {\doibase
  10.1109/SFCS.1991.185341} {\emph {\bibinfo {booktitle} {[1991] Proceedings
  32nd Annual Symposium of Foundations of Computer Science}}}\ (\bibinfo {year}
  {1991})\ pp.\ \bibinfo {pages} {2--12}\BibitemShut {NoStop}%
\bibitem [{\citenamefont {Arora}\ \emph {et~al.}(1998)\citenamefont {Arora},
  \citenamefont {Lund}, \citenamefont {Motwani}, \citenamefont {Sudan},\ and\
  \citenamefont {Szegedy}}]{arora1998proof}%
  \BibitemOpen
  \bibfield  {author} {\bibinfo {author} {\bibfnamefont {S.}~\bibnamefont
  {Arora}}, \bibinfo {author} {\bibfnamefont {C.}~\bibnamefont {Lund}},
  \bibinfo {author} {\bibfnamefont {R.}~\bibnamefont {Motwani}}, \bibinfo
  {author} {\bibfnamefont {M.}~\bibnamefont {Sudan}}, \ and\ \bibinfo {author}
  {\bibfnamefont {M.}~\bibnamefont {Szegedy}},\ }\href {\doibase
  10.1145/278298.278306} {\bibfield  {journal} {\bibinfo  {journal} {J. ACM}\
  }\textbf {\bibinfo {volume} {45}},\ \bibinfo {pages} {501–555} (\bibinfo
  {year} {1998})}\BibitemShut {NoStop}%
\bibitem [{\citenamefont {Arora}\ and\ \citenamefont
  {Safra}(1998)}]{Arora98prob}%
  \BibitemOpen
  \bibfield  {author} {\bibinfo {author} {\bibfnamefont {S.}~\bibnamefont
  {Arora}}\ and\ \bibinfo {author} {\bibfnamefont {S.}~\bibnamefont {Safra}},\
  }\href {\doibase 10.1145/273865.273901} {\ \textbf {\bibinfo {volume} {45}},\
  \bibinfo {pages} {70–122} (\bibinfo {year} {1998})}\BibitemShut {NoStop}%
\bibitem [{\citenamefont {Ben-Sasson}\ \emph {et~al.}(2004)\citenamefont
  {Ben-Sasson}, \citenamefont {Goldreich}, \citenamefont {Harsha},
  \citenamefont {Sudan},\ and\ \citenamefont {Vadhan}}]{Ben04STOC}%
  \BibitemOpen
  \bibfield  {author} {\bibinfo {author} {\bibfnamefont {E.}~\bibnamefont
  {Ben-Sasson}}, \bibinfo {author} {\bibfnamefont {O.}~\bibnamefont
  {Goldreich}}, \bibinfo {author} {\bibfnamefont {P.}~\bibnamefont {Harsha}},
  \bibinfo {author} {\bibfnamefont {M.}~\bibnamefont {Sudan}}, \ and\ \bibinfo
  {author} {\bibfnamefont {S.}~\bibnamefont {Vadhan}},\ }in\ \href {\doibase
  10.1145/1007352.1007361} {\emph {\bibinfo {booktitle} {Proceedings of the
  Thirty-Sixth Annual ACM Symposium on Theory of Computing}}},\ \bibinfo
  {series and number} {STOC '04}\ (\bibinfo  {publisher} {Association for
  Computing Machinery},\ \bibinfo {address} {New York, NY, USA},\ \bibinfo
  {year} {2004})\ p.\ \bibinfo {pages} {1–10}\BibitemShut {NoStop}%
\bibitem [{\citenamefont {Goldreich}\ and\ \citenamefont
  {Sudan}(2006)}]{Goldreich06}%
  \BibitemOpen
  \bibfield  {author} {\bibinfo {author} {\bibfnamefont {O.}~\bibnamefont
  {Goldreich}}\ and\ \bibinfo {author} {\bibfnamefont {M.}~\bibnamefont
  {Sudan}},\ }\href {\doibase 10.1145/1162349.1162351} {\bibfield  {journal}
  {\bibinfo  {journal} {J. ACM}\ }\textbf {\bibinfo {volume} {53}},\ \bibinfo
  {pages} {558–655} (\bibinfo {year} {2006})}\BibitemShut {NoStop}%
\bibitem [{\citenamefont {Moshkovitz}\ and\ \citenamefont
  {Raz}(2008)}]{Moshkovitz08}%
  \BibitemOpen
  \bibfield  {author} {\bibinfo {author} {\bibfnamefont {D.}~\bibnamefont
  {Moshkovitz}}\ and\ \bibinfo {author} {\bibfnamefont {R.}~\bibnamefont
  {Raz}},\ }\href {\doibase 10.1145/1754399.1754402} {\bibfield  {journal}
  {\bibinfo  {journal} {J. ACM}\ }\textbf {\bibinfo {volume} {57}},\ \bibinfo
  {pages} {1} (\bibinfo {year} {2008})}\BibitemShut {NoStop}%
\bibitem [{\citenamefont {Dinur}\ and\ \citenamefont
  {Harsha}(2013)}]{DinurSIAM13}%
  \BibitemOpen
  \bibfield  {author} {\bibinfo {author} {\bibfnamefont {I.}~\bibnamefont
  {Dinur}}\ and\ \bibinfo {author} {\bibfnamefont {P.}~\bibnamefont {Harsha}},\
  }\href {\doibase 10.1137/100788161} {\bibfield  {journal} {\bibinfo
  {journal} {SIAM Journal on Computing}\ }\textbf {\bibinfo {volume} {42}},\
  \bibinfo {pages} {2452} (\bibinfo {year} {2013})}\BibitemShut {NoStop}%
\bibitem [{\citenamefont {Harrow}\ and\ \citenamefont
  {Montanaro}(2010)}]{Harrow13}%
  \BibitemOpen
  \bibfield  {author} {\bibinfo {author} {\bibfnamefont {A.~W.}\ \bibnamefont
  {Harrow}}\ and\ \bibinfo {author} {\bibfnamefont {A.}~\bibnamefont
  {Montanaro}},\ }in\ \href {\doibase 10.1109/FOCS.2010.66} {\emph {\bibinfo
  {booktitle} {2010 IEEE 51st Annual Symposium on Foundations of Computer
  Science}}}\ (\bibinfo {year} {2010})\ pp.\ \bibinfo {pages}
  {633--642}\BibitemShut {NoStop}%
\bibitem [{\citenamefont {Gutoski}\ \emph {et~al.}(2015)\citenamefont
  {Gutoski}, \citenamefont {Hayden}, \citenamefont {Milner},\ and\
  \citenamefont {Wilde}}]{v011a003}%
  \BibitemOpen
  \bibfield  {author} {\bibinfo {author} {\bibfnamefont {G.}~\bibnamefont
  {Gutoski}}, \bibinfo {author} {\bibfnamefont {P.}~\bibnamefont {Hayden}},
  \bibinfo {author} {\bibfnamefont {K.}~\bibnamefont {Milner}}, \ and\ \bibinfo
  {author} {\bibfnamefont {M.~M.}\ \bibnamefont {Wilde}},\ }\href {\doibase
  10.4086/toc.2015.v011a003} {\bibfield  {journal} {\bibinfo  {journal} {Theory
  of Computing}\ }\textbf {\bibinfo {volume} {11}},\ \bibinfo {pages} {59}
  (\bibinfo {year} {2015})}\BibitemShut {NoStop}%
\bibitem [{\citenamefont {Beckey}\ \emph {et~al.}(2021)\citenamefont {Beckey},
  \citenamefont {Gigena}, \citenamefont {Coles},\ and\ \citenamefont
  {Cerezo}}]{Beckey21}%
  \BibitemOpen
  \bibfield  {author} {\bibinfo {author} {\bibfnamefont {J.~L.}\ \bibnamefont
  {Beckey}}, \bibinfo {author} {\bibfnamefont {N.}~\bibnamefont {Gigena}},
  \bibinfo {author} {\bibfnamefont {P.~J.}\ \bibnamefont {Coles}}, \ and\
  \bibinfo {author} {\bibfnamefont {M.}~\bibnamefont {Cerezo}},\ }\href
  {\doibase 10.1103/PhysRevLett.127.140501} {\bibfield  {journal} {\bibinfo
  {journal} {Phys. Rev. Lett.}\ }\textbf {\bibinfo {volume} {127}},\ \bibinfo
  {pages} {140501} (\bibinfo {year} {2021})}\BibitemShut {NoStop}%
\bibitem [{\citenamefont {Montanaro}\ and\ \citenamefont {Wolf}(2016)}]{gs007}%
  \BibitemOpen
  \bibfield  {author} {\bibinfo {author} {\bibfnamefont {A.}~\bibnamefont
  {Montanaro}}\ and\ \bibinfo {author} {\bibfnamefont {R.~d.}\ \bibnamefont
  {Wolf}},\ }\href {\doibase 10.4086/toc.gs.2016.007} {\emph {\bibinfo {title}
  {A Survey of Quantum Property Testing}}},\ \bibinfo {series} {Graduate
  Surveys}\ No.~\bibinfo {number} {7}\ (\bibinfo  {publisher} {Theory of
  Computing Library},\ \bibinfo {year} {2016})\ pp.\ \bibinfo {pages}
  {1--81}\BibitemShut {NoStop}%
\bibitem [{\citenamefont {Buhrman}\ \emph
  {et~al.}(2001{\natexlab{a}})\citenamefont {Buhrman}, \citenamefont {Cleve},
  \citenamefont {Watrous},\ and\ \citenamefont {de~Wolf}}]{Buhrman01}%
  \BibitemOpen
  \bibfield  {author} {\bibinfo {author} {\bibfnamefont {H.}~\bibnamefont
  {Buhrman}}, \bibinfo {author} {\bibfnamefont {R.}~\bibnamefont {Cleve}},
  \bibinfo {author} {\bibfnamefont {J.}~\bibnamefont {Watrous}}, \ and\
  \bibinfo {author} {\bibfnamefont {R.}~\bibnamefont {de~Wolf}},\ }\href
  {\doibase 10.1103/PhysRevLett.87.167902} {\bibfield  {journal} {\bibinfo
  {journal} {Phys. Rev. Lett.}\ }\textbf {\bibinfo {volume} {87}},\ \bibinfo
  {pages} {167902} (\bibinfo {year} {2001}{\natexlab{a}})}\BibitemShut
  {NoStop}%
\bibitem [{\citenamefont {Natarajan}\ and\ \citenamefont
  {Vidick}(2017)}]{natarajan2017quantum}%
  \BibitemOpen
  \bibfield  {author} {\bibinfo {author} {\bibfnamefont {A.}~\bibnamefont
  {Natarajan}}\ and\ \bibinfo {author} {\bibfnamefont {T.}~\bibnamefont
  {Vidick}}\ }(\bibinfo  {publisher} {Association for Computing Machinery},\
  \bibinfo {address} {New York, NY, USA},\ \bibinfo {year} {2017})\ p.\
  \bibinfo {pages} {1003–1015}\BibitemShut {NoStop}%
\bibitem [{\citenamefont {Natarajan}\ and\ \citenamefont
  {Vidick}(2018)}]{Natarajanfocs18}%
  \BibitemOpen
  \bibfield  {author} {\bibinfo {author} {\bibfnamefont {A.}~\bibnamefont
  {Natarajan}}\ and\ \bibinfo {author} {\bibfnamefont {T.}~\bibnamefont
  {Vidick}},\ }in\ \href {\doibase 10.1109/FOCS.2018.00075} {\emph {\bibinfo
  {booktitle} {2018 IEEE 59th Annual Symposium on Foundations of Computer
  Science (FOCS)}}}\ (\bibinfo {year} {2018})\ pp.\ \bibinfo {pages}
  {731--742}\BibitemShut {NoStop}%
\bibitem [{\citenamefont {Ji}\ \emph {et~al.}(2022)\citenamefont {Ji},
  \citenamefont {Natarajan}, \citenamefont {Vidick}, \citenamefont {Wright},\
  and\ \citenamefont {Yuen}}]{ji2022mipre}%
  \BibitemOpen
  \bibfield  {author} {\bibinfo {author} {\bibfnamefont {Z.}~\bibnamefont
  {Ji}}, \bibinfo {author} {\bibfnamefont {A.}~\bibnamefont {Natarajan}},
  \bibinfo {author} {\bibfnamefont {T.}~\bibnamefont {Vidick}}, \bibinfo
  {author} {\bibfnamefont {J.}~\bibnamefont {Wright}}, \ and\ \bibinfo {author}
  {\bibfnamefont {H.}~\bibnamefont {Yuen}},\ }\href@noop {} {\  (\bibinfo
  {year} {2022})},\ \Eprint {http://arxiv.org/abs/2001.04383} {arXiv:2001.04383
  [quant-ph]} \BibitemShut {NoStop}%
\bibitem [{\citenamefont {Vidal}\ \emph {et~al.}(2003)\citenamefont {Vidal},
  \citenamefont {Latorre}, \citenamefont {Rico},\ and\ \citenamefont
  {Kitaev}}]{VidalPRL03}%
  \BibitemOpen
  \bibfield  {author} {\bibinfo {author} {\bibfnamefont {G.}~\bibnamefont
  {Vidal}}, \bibinfo {author} {\bibfnamefont {J.~I.}\ \bibnamefont {Latorre}},
  \bibinfo {author} {\bibfnamefont {E.}~\bibnamefont {Rico}}, \ and\ \bibinfo
  {author} {\bibfnamefont {A.}~\bibnamefont {Kitaev}},\ }\href {\doibase
  10.1103/PhysRevLett.90.227902} {\bibfield  {journal} {\bibinfo  {journal}
  {Phys. Rev. Lett.}\ }\textbf {\bibinfo {volume} {90}},\ \bibinfo {pages}
  {227902} (\bibinfo {year} {2003})}\BibitemShut {NoStop}%
\bibitem [{\citenamefont {Calabrese}\ and\ \citenamefont
  {Cardy}(2009)}]{Calabrese09}%
  \BibitemOpen
  \bibfield  {author} {\bibinfo {author} {\bibfnamefont {P.}~\bibnamefont
  {Calabrese}}\ and\ \bibinfo {author} {\bibfnamefont {J.}~\bibnamefont
  {Cardy}},\ }\href {\doibase 10.1088/1751-8113/42/50/504005} {\bibfield
  {journal} {\bibinfo  {journal} {Journal of Physics A: Mathematical and
  Theoretical}\ }\textbf {\bibinfo {volume} {42}},\ \bibinfo {pages} {504005}
  (\bibinfo {year} {2009})}\BibitemShut {NoStop}%
\bibitem [{\citenamefont {Nishioka}\ \emph {et~al.}(2009)\citenamefont
  {Nishioka}, \citenamefont {Ryu},\ and\ \citenamefont
  {Takayanagi}}]{Nishioka09}%
  \BibitemOpen
  \bibfield  {author} {\bibinfo {author} {\bibfnamefont {T.}~\bibnamefont
  {Nishioka}}, \bibinfo {author} {\bibfnamefont {S.}~\bibnamefont {Ryu}}, \
  and\ \bibinfo {author} {\bibfnamefont {T.}~\bibnamefont {Takayanagi}},\
  }\href {\doibase 10.1088/1751-8113/42/50/504008} {\bibfield  {journal}
  {\bibinfo  {journal} {Journal of Physics A: Mathematical and Theoretical}\
  }\textbf {\bibinfo {volume} {42}},\ \bibinfo {pages} {504008} (\bibinfo
  {year} {2009})}\BibitemShut {NoStop}%
\bibitem [{\citenamefont {Islam}\ \emph {et~al.}(2015)\citenamefont {Islam},
  \citenamefont {Ma}, \citenamefont {Preiss}, \citenamefont {Eric~Tai},
  \citenamefont {Lukin}, \citenamefont {Rispoli},\ and\ \citenamefont
  {Greiner}}]{Islam15}%
  \BibitemOpen
  \bibfield  {author} {\bibinfo {author} {\bibfnamefont {R.}~\bibnamefont
  {Islam}}, \bibinfo {author} {\bibfnamefont {R.}~\bibnamefont {Ma}}, \bibinfo
  {author} {\bibfnamefont {P.~M.}\ \bibnamefont {Preiss}}, \bibinfo {author}
  {\bibfnamefont {M.}~\bibnamefont {Eric~Tai}}, \bibinfo {author}
  {\bibfnamefont {A.}~\bibnamefont {Lukin}}, \bibinfo {author} {\bibfnamefont
  {M.}~\bibnamefont {Rispoli}}, \ and\ \bibinfo {author} {\bibfnamefont
  {M.}~\bibnamefont {Greiner}},\ }\href {\doibase 10.1038/nature15750}
  {\bibfield  {journal} {\bibinfo  {journal} {Nature}\ }\textbf {\bibinfo
  {volume} {528}},\ \bibinfo {pages} {77} (\bibinfo {year} {2015})}\BibitemShut
  {NoStop}%
\bibitem [{\citenamefont {Gottesman}(1997)}]{Gottesman97}%
  \BibitemOpen
  \bibfield  {author} {\bibinfo {author} {\bibfnamefont {D.}~\bibnamefont
  {Gottesman}},\ }\href {https://arxiv.org/abs/quant-ph/9705052} {\bibfield
  {journal} {\bibinfo  {journal} {arXiv:quant-ph/9705052}\ } (\bibinfo {year}
  {1997})}\BibitemShut {NoStop}%
\bibitem [{\citenamefont {Gottesman}(1998)}]{gottesman1998heisenberg}%
  \BibitemOpen
  \bibfield  {author} {\bibinfo {author} {\bibfnamefont {D.}~\bibnamefont
  {Gottesman}},\ }in\ \href@noop {} {\emph {\bibinfo {booktitle} {Proc. XXII
  International Colloquium on Group Theoretical Methods in Physics, 1998}}}\
  (\bibinfo {year} {1998})\ pp.\ \bibinfo {pages} {32--43}\BibitemShut
  {NoStop}%
\bibitem [{\citenamefont {Low}(2009)}]{LowPRA09}%
  \BibitemOpen
  \bibfield  {author} {\bibinfo {author} {\bibfnamefont {R.~A.}\ \bibnamefont
  {Low}},\ }\href {\doibase 10.1103/PhysRevA.80.052314} {\bibfield  {journal}
  {\bibinfo  {journal} {Phys. Rev. A}\ }\textbf {\bibinfo {volume} {80}},\
  \bibinfo {pages} {052314} (\bibinfo {year} {2009})}\BibitemShut {NoStop}%
\bibitem [{\citenamefont {Wang}(2011)}]{WangPRA11}%
  \BibitemOpen
  \bibfield  {author} {\bibinfo {author} {\bibfnamefont {G.}~\bibnamefont
  {Wang}},\ }\href {\doibase 10.1103/PhysRevA.84.052328} {\bibfield  {journal}
  {\bibinfo  {journal} {Phys. Rev. A}\ }\textbf {\bibinfo {volume} {84}},\
  \bibinfo {pages} {052328} (\bibinfo {year} {2011})}\BibitemShut {NoStop}%
\bibitem [{\citenamefont {Rocchetto}(2018)}]{Rocchetto18}%
  \BibitemOpen
  \bibfield  {author} {\bibinfo {author} {\bibfnamefont {A.}~\bibnamefont
  {Rocchetto}},\ }\href {\doibase 10.26421/QIC18.7-8-1} {\bibfield  {journal}
  {\bibinfo  {journal} {Quantum Information and Computation}\ }\textbf
  {\bibinfo {volume} {18}},\ \bibinfo {pages} {541} (\bibinfo {year}
  {2018})}\BibitemShut {NoStop}%
\bibitem [{\citenamefont {Montanaro}(2017)}]{montanaro2017learning}%
  \BibitemOpen
  \bibfield  {author} {\bibinfo {author} {\bibfnamefont {A.}~\bibnamefont
  {Montanaro}},\ }\href@noop {} {\enquote {\bibinfo {title} {Learning
  stabilizer states by bell sampling},}\ } (\bibinfo {year} {2017}),\ \Eprint
  {http://arxiv.org/abs/1707.04012} {arXiv:1707.04012 [quant-ph]} \BibitemShut
  {NoStop}%
\bibitem [{\citenamefont {Gross}\ \emph {et~al.}(2021)\citenamefont {Gross},
  \citenamefont {Nezami},\ and\ \citenamefont {Walter}}]{Gross21}%
  \BibitemOpen
  \bibfield  {author} {\bibinfo {author} {\bibfnamefont {D.}~\bibnamefont
  {Gross}}, \bibinfo {author} {\bibfnamefont {S.}~\bibnamefont {Nezami}}, \
  and\ \bibinfo {author} {\bibfnamefont {M.}~\bibnamefont {Walter}},\ }\href
  {\doibase 10.1007/s00220-021-04118-7} {\bibfield  {journal} {\bibinfo
  {journal} {Communications in Mathematical Physics}\ }\textbf {\bibinfo
  {volume} {385}},\ \bibinfo {pages} {1325} (\bibinfo {year}
  {2021})}\BibitemShut {NoStop}%
\bibitem [{\citenamefont {Lai}\ and\ \citenamefont {Cheng}(2022)}]{LaiIEEE22}%
  \BibitemOpen
  \bibfield  {author} {\bibinfo {author} {\bibfnamefont {C.-Y.}\ \bibnamefont
  {Lai}}\ and\ \bibinfo {author} {\bibfnamefont {H.-C.}\ \bibnamefont
  {Cheng}},\ }\href {\doibase 10.1109/TIT.2022.3151760} {\bibfield  {journal}
  {\bibinfo  {journal} {IEEE Transactions on Information Theory}\ }\textbf
  {\bibinfo {volume} {68}},\ \bibinfo {pages} {3951} (\bibinfo {year}
  {2022})}\BibitemShut {NoStop}%
\bibitem [{\citenamefont {Grewal}\ \emph
  {et~al.}(2023{\natexlab{a}})\citenamefont {Grewal}, \citenamefont {Iyer},
  \citenamefont {Kretschmer},\ and\ \citenamefont {Liang}}]{grewalLIPICs23}%
  \BibitemOpen
  \bibfield  {author} {\bibinfo {author} {\bibfnamefont {S.}~\bibnamefont
  {Grewal}}, \bibinfo {author} {\bibfnamefont {V.}~\bibnamefont {Iyer}},
  \bibinfo {author} {\bibfnamefont {W.}~\bibnamefont {Kretschmer}}, \ and\
  \bibinfo {author} {\bibfnamefont {D.}~\bibnamefont {Liang}},\ }in\ \href
  {\doibase 10.4230/LIPIcs.ITCS.2023.64} {\emph {\bibinfo {booktitle} {14th
  Innovations in Theoretical Computer Science Conference (ITCS 2023)}}},\
  \bibinfo {series} {Leibniz International Proceedings in Informatics
  (LIPIcs)}, Vol.\ \bibinfo {volume} {251},\ \bibinfo {editor} {edited by\
  \bibinfo {editor} {\bibfnamefont {Y.}~\bibnamefont {Tauman~Kalai}}}\
  (\bibinfo  {publisher} {Schloss Dagstuhl -- Leibniz-Zentrum f{\"u}r
  Informatik},\ \bibinfo {address} {Dagstuhl, Germany},\ \bibinfo {year}
  {2023})\ pp.\ \bibinfo {pages} {64:1--64:20}\BibitemShut {NoStop}%
\bibitem [{\citenamefont {Grewal}\ \emph
  {et~al.}(2023{\natexlab{b}})\citenamefont {Grewal}, \citenamefont {Iyer},
  \citenamefont {Kretschmer},\ and\ \citenamefont
  {Liang}}]{grewal2023improved}%
  \BibitemOpen
  \bibfield  {author} {\bibinfo {author} {\bibfnamefont {S.}~\bibnamefont
  {Grewal}}, \bibinfo {author} {\bibfnamefont {V.}~\bibnamefont {Iyer}},
  \bibinfo {author} {\bibfnamefont {W.}~\bibnamefont {Kretschmer}}, \ and\
  \bibinfo {author} {\bibfnamefont {D.}~\bibnamefont {Liang}},\ }\href
  {https://arxiv.org/abs/2304.13915} {\bibfield  {journal} {\bibinfo  {journal}
  {arXiv:2304.13915}\ } (\bibinfo {year} {2023}{\natexlab{b}})}\BibitemShut
  {NoStop}%
\bibitem [{\citenamefont {Haug}\ and\ \citenamefont {Kim}(2023)}]{Haug23}%
  \BibitemOpen
  \bibfield  {author} {\bibinfo {author} {\bibfnamefont {T.}~\bibnamefont
  {Haug}}\ and\ \bibinfo {author} {\bibfnamefont {M.}~\bibnamefont {Kim}},\
  }\href {\doibase 10.1103/PRXQuantum.4.010301} {\bibfield  {journal} {\bibinfo
   {journal} {PRX Quantum}\ }\textbf {\bibinfo {volume} {4}},\ \bibinfo {pages}
  {010301} (\bibinfo {year} {2023})}\BibitemShut {NoStop}%
\bibitem [{\citenamefont {Veitch}\ \emph {et~al.}(2012)\citenamefont {Veitch},
  \citenamefont {Ferrie}, \citenamefont {Gross},\ and\ \citenamefont
  {Emerson}}]{Veitch12mag}%
  \BibitemOpen
  \bibfield  {author} {\bibinfo {author} {\bibfnamefont {V.}~\bibnamefont
  {Veitch}}, \bibinfo {author} {\bibfnamefont {C.}~\bibnamefont {Ferrie}},
  \bibinfo {author} {\bibfnamefont {D.}~\bibnamefont {Gross}}, \ and\ \bibinfo
  {author} {\bibfnamefont {J.}~\bibnamefont {Emerson}},\ }\href {\doibase
  0.1088/1367-2630/14/11/113011} {\bibfield  {journal} {\bibinfo  {journal}
  {New J. Phys.}\ }\textbf {\bibinfo {volume} {14}},\ \bibinfo {pages} {113011}
  (\bibinfo {year} {2012})}\BibitemShut {NoStop}%
\bibitem [{\citenamefont {Veitch}\ \emph {et~al.}(2014)\citenamefont {Veitch},
  \citenamefont {Mousavian}, \citenamefont {Gottesman},\ and\ \citenamefont
  {Emerson}}]{Veitch14}%
  \BibitemOpen
  \bibfield  {author} {\bibinfo {author} {\bibfnamefont {V.}~\bibnamefont
  {Veitch}}, \bibinfo {author} {\bibfnamefont {S.~A.~H.}\ \bibnamefont
  {Mousavian}}, \bibinfo {author} {\bibfnamefont {D.}~\bibnamefont
  {Gottesman}}, \ and\ \bibinfo {author} {\bibfnamefont {J.}~\bibnamefont
  {Emerson}},\ }\href {\doibase 10.1088/1367-2630/16/1/013009} {\bibfield
  {journal} {\bibinfo  {journal} {New Journal of Physics}\ }\textbf {\bibinfo
  {volume} {16}},\ \bibinfo {pages} {013009} (\bibinfo {year}
  {2014})}\BibitemShut {NoStop}%
\bibitem [{\citenamefont {Howard}\ and\ \citenamefont
  {Campbell}(2017)}]{HowardPRL17}%
  \BibitemOpen
  \bibfield  {author} {\bibinfo {author} {\bibfnamefont {M.}~\bibnamefont
  {Howard}}\ and\ \bibinfo {author} {\bibfnamefont {E.}~\bibnamefont
  {Campbell}},\ }\href {\doibase 10.1103/PhysRevLett.118.090501} {\bibfield
  {journal} {\bibinfo  {journal} {Phys. Rev. Lett.}\ }\textbf {\bibinfo
  {volume} {118}},\ \bibinfo {pages} {090501} (\bibinfo {year}
  {2017})}\BibitemShut {NoStop}%
\bibitem [{\citenamefont {Beverland}\ \emph {et~al.}(2020)\citenamefont
  {Beverland}, \citenamefont {Campbell}, \citenamefont {Howard},\ and\
  \citenamefont {Kliuchnikov}}]{BeverlandQST20}%
  \BibitemOpen
  \bibfield  {author} {\bibinfo {author} {\bibfnamefont {M.}~\bibnamefont
  {Beverland}}, \bibinfo {author} {\bibfnamefont {E.}~\bibnamefont {Campbell}},
  \bibinfo {author} {\bibfnamefont {M.}~\bibnamefont {Howard}}, \ and\ \bibinfo
  {author} {\bibfnamefont {V.}~\bibnamefont {Kliuchnikov}},\ }\href {\doibase
  10.1088/2058-9565/ab8963} {\bibfield  {journal} {\bibinfo  {journal} {Quantum
  Sci. Technol.}\ }\textbf {\bibinfo {volume} {5}},\ \bibinfo {pages} {035009}
  (\bibinfo {year} {2020})}\BibitemShut {NoStop}%
\bibitem [{\citenamefont {Seddon}\ \emph {et~al.}(2021)\citenamefont {Seddon},
  \citenamefont {Regula}, \citenamefont {Pashayan}, \citenamefont {Ouyang},\
  and\ \citenamefont {Campbell}}]{SeddonPRXQ21}%
  \BibitemOpen
  \bibfield  {author} {\bibinfo {author} {\bibfnamefont {J.~R.}\ \bibnamefont
  {Seddon}}, \bibinfo {author} {\bibfnamefont {B.}~\bibnamefont {Regula}},
  \bibinfo {author} {\bibfnamefont {H.}~\bibnamefont {Pashayan}}, \bibinfo
  {author} {\bibfnamefont {Y.}~\bibnamefont {Ouyang}}, \ and\ \bibinfo {author}
  {\bibfnamefont {E.~T.}\ \bibnamefont {Campbell}},\ }\href {\doibase
  10.1103/PRXQuantum.2.010345} {\bibfield  {journal} {\bibinfo  {journal} {PRX
  Quantum}\ }\textbf {\bibinfo {volume} {2}},\ \bibinfo {pages} {010345}
  (\bibinfo {year} {2021})}\BibitemShut {NoStop}%
\bibitem [{\citenamefont {Bravyi}\ and\ \citenamefont
  {Gosset}(2016)}]{BravyiPRL16}%
  \BibitemOpen
  \bibfield  {author} {\bibinfo {author} {\bibfnamefont {S.}~\bibnamefont
  {Bravyi}}\ and\ \bibinfo {author} {\bibfnamefont {D.}~\bibnamefont
  {Gosset}},\ }\href {\doibase 10.1103/PhysRevLett.116.250501} {\bibfield
  {journal} {\bibinfo  {journal} {Phys. Rev. Lett.}\ }\textbf {\bibinfo
  {volume} {116}},\ \bibinfo {pages} {250501} (\bibinfo {year}
  {2016})}\BibitemShut {NoStop}%
\bibitem [{\citenamefont {Bravyi}\ \emph {et~al.}(2016)\citenamefont {Bravyi},
  \citenamefont {Smith},\ and\ \citenamefont {Smolin}}]{BravyiPRX16}%
  \BibitemOpen
  \bibfield  {author} {\bibinfo {author} {\bibfnamefont {S.}~\bibnamefont
  {Bravyi}}, \bibinfo {author} {\bibfnamefont {G.}~\bibnamefont {Smith}}, \
  and\ \bibinfo {author} {\bibfnamefont {J.~A.}\ \bibnamefont {Smolin}},\
  }\href {\doibase 10.1103/PhysRevX.6.021043} {\bibfield  {journal} {\bibinfo
  {journal} {Phys. Rev. X}\ }\textbf {\bibinfo {volume} {6}},\ \bibinfo {pages}
  {021043} (\bibinfo {year} {2016})}\BibitemShut {NoStop}%
\bibitem [{\citenamefont {Bravyi}\ \emph {et~al.}(2019)\citenamefont {Bravyi},
  \citenamefont {Browne}, \citenamefont {Calpin}, \citenamefont {Campbell},
  \citenamefont {Gosset},\ and\ \citenamefont {Howard}}]{bravyi2019simulation}%
  \BibitemOpen
  \bibfield  {author} {\bibinfo {author} {\bibfnamefont {S.}~\bibnamefont
  {Bravyi}}, \bibinfo {author} {\bibfnamefont {D.}~\bibnamefont {Browne}},
  \bibinfo {author} {\bibfnamefont {P.}~\bibnamefont {Calpin}}, \bibinfo
  {author} {\bibfnamefont {E.}~\bibnamefont {Campbell}}, \bibinfo {author}
  {\bibfnamefont {D.}~\bibnamefont {Gosset}}, \ and\ \bibinfo {author}
  {\bibfnamefont {M.}~\bibnamefont {Howard}},\ }\href {\doibase
  10.22331/q-2019-09-02-181} {\bibfield  {journal} {\bibinfo  {journal}
  {{Quantum}}\ }\textbf {\bibinfo {volume} {3}},\ \bibinfo {pages} {181}
  (\bibinfo {year} {2019})}\BibitemShut {NoStop}%
\bibitem [{\citenamefont {Bu}\ and\ \citenamefont {Koh}(2019)}]{Bu19}%
  \BibitemOpen
  \bibfield  {author} {\bibinfo {author} {\bibfnamefont {K.}~\bibnamefont
  {Bu}}\ and\ \bibinfo {author} {\bibfnamefont {D.~E.}\ \bibnamefont {Koh}},\
  }\href {\doibase 10.1103/PhysRevLett.123.170502} {\bibfield  {journal}
  {\bibinfo  {journal} {Phys. Rev. Lett.}\ }\textbf {\bibinfo {volume} {123}},\
  \bibinfo {pages} {170502} (\bibinfo {year} {2019})}\BibitemShut {NoStop}%
\bibitem [{\citenamefont {Bu}\ \emph {et~al.}(2024)\citenamefont {Bu},
  \citenamefont {Garcia}, \citenamefont {Jaffe}, \citenamefont {Koh},\ and\
  \citenamefont {Li}}]{Bucomplexity22}%
  \BibitemOpen
  \bibfield  {author} {\bibinfo {author} {\bibfnamefont {K.}~\bibnamefont
  {Bu}}, \bibinfo {author} {\bibfnamefont {R.~J.}\ \bibnamefont {Garcia}},
  \bibinfo {author} {\bibfnamefont {A.}~\bibnamefont {Jaffe}}, \bibinfo
  {author} {\bibfnamefont {D.~E.}\ \bibnamefont {Koh}}, \ and\ \bibinfo
  {author} {\bibfnamefont {L.}~\bibnamefont {Li}},\ }\href
  {https://doi.org/10.1007/s00220-024-05030-6} {\bibfield  {journal} {\bibinfo
  {journal} {Communications in Mathematical Physics}\ }\textbf {\bibinfo
  {volume} {405}},\ \bibinfo {pages} {161} (\bibinfo {year}
  {2024})}\BibitemShut {NoStop}%
\bibitem [{\citenamefont {Bu}\ \emph {et~al.}(2022)\citenamefont {Bu},
  \citenamefont {Koh}, \citenamefont {Li}, \citenamefont {Luo},\ and\
  \citenamefont {Zhang}}]{BuPRA19_stat}%
  \BibitemOpen
  \bibfield  {author} {\bibinfo {author} {\bibfnamefont {K.}~\bibnamefont
  {Bu}}, \bibinfo {author} {\bibfnamefont {D.~E.}\ \bibnamefont {Koh}},
  \bibinfo {author} {\bibfnamefont {L.}~\bibnamefont {Li}}, \bibinfo {author}
  {\bibfnamefont {Q.}~\bibnamefont {Luo}}, \ and\ \bibinfo {author}
  {\bibfnamefont {Y.}~\bibnamefont {Zhang}},\ }\href {\doibase
  10.1103/PhysRevA.105.062431} {\bibfield  {journal} {\bibinfo  {journal}
  {Phys. Rev. A}\ }\textbf {\bibinfo {volume} {105}},\ \bibinfo {pages}
  {062431} (\bibinfo {year} {2022})}\BibitemShut {NoStop}%
\bibitem [{\citenamefont {Rall}\ \emph {et~al.}(2019)\citenamefont {Rall},
  \citenamefont {Liang}, \citenamefont {Cook},\ and\ \citenamefont
  {Kretschmer}}]{RallPRA19}%
  \BibitemOpen
  \bibfield  {author} {\bibinfo {author} {\bibfnamefont {P.}~\bibnamefont
  {Rall}}, \bibinfo {author} {\bibfnamefont {D.}~\bibnamefont {Liang}},
  \bibinfo {author} {\bibfnamefont {J.}~\bibnamefont {Cook}}, \ and\ \bibinfo
  {author} {\bibfnamefont {W.}~\bibnamefont {Kretschmer}},\ }\href {\doibase
  10.1103/PhysRevA.99.062337} {\bibfield  {journal} {\bibinfo  {journal} {Phys.
  Rev. A}\ }\textbf {\bibinfo {volume} {99}},\ \bibinfo {pages} {062337}
  (\bibinfo {year} {2019})}\BibitemShut {NoStop}%
\bibitem [{\citenamefont {Wang}\ \emph {et~al.}(2019)\citenamefont {Wang},
  \citenamefont {Wilde},\ and\ \citenamefont {Su}}]{WangNJP19}%
  \BibitemOpen
  \bibfield  {author} {\bibinfo {author} {\bibfnamefont {X.}~\bibnamefont
  {Wang}}, \bibinfo {author} {\bibfnamefont {M.~M.}\ \bibnamefont {Wilde}}, \
  and\ \bibinfo {author} {\bibfnamefont {Y.}~\bibnamefont {Su}},\ }\href
  {\doibase 10.1088/1367-2630/ab451d} {\bibfield  {journal} {\bibinfo
  {journal} {New Journal of Physics}\ }\textbf {\bibinfo {volume} {21}},\
  \bibinfo {pages} {103002} (\bibinfo {year} {2019})}\BibitemShut {NoStop}%
\bibitem [{\citenamefont {Leone}\ \emph {et~al.}(2022)\citenamefont {Leone},
  \citenamefont {Oliviero},\ and\ \citenamefont {Hamma}}]{LeonePRL22}%
  \BibitemOpen
  \bibfield  {author} {\bibinfo {author} {\bibfnamefont {L.}~\bibnamefont
  {Leone}}, \bibinfo {author} {\bibfnamefont {S.~F.~E.}\ \bibnamefont
  {Oliviero}}, \ and\ \bibinfo {author} {\bibfnamefont {A.}~\bibnamefont
  {Hamma}},\ }\href {\doibase 10.1103/PhysRevLett.128.050402} {\bibfield
  {journal} {\bibinfo  {journal} {Phys. Rev. Lett.}\ }\textbf {\bibinfo
  {volume} {128}},\ \bibinfo {pages} {050402} (\bibinfo {year}
  {2022})}\BibitemShut {NoStop}%
\bibitem [{\citenamefont {Haug}\ \emph {et~al.}(2023)\citenamefont {Haug},
  \citenamefont {Lee},\ and\ \citenamefont {Kim}}]{haug2023efficient}%
  \BibitemOpen
  \bibfield  {author} {\bibinfo {author} {\bibfnamefont {T.}~\bibnamefont
  {Haug}}, \bibinfo {author} {\bibfnamefont {S.}~\bibnamefont {Lee}}, \ and\
  \bibinfo {author} {\bibfnamefont {M.~S.}\ \bibnamefont {Kim}},\ }\href@noop
  {} {\  (\bibinfo {year} {2023})},\ \Eprint {http://arxiv.org/abs/2305.19152}
  {arXiv:2305.19152 [quant-ph]} \BibitemShut {NoStop}%
\bibitem [{\citenamefont {Haug}\ and\ \citenamefont
  {Piroli}(2023{\natexlab{a}})}]{HaugPRB23}%
  \BibitemOpen
  \bibfield  {author} {\bibinfo {author} {\bibfnamefont {T.}~\bibnamefont
  {Haug}}\ and\ \bibinfo {author} {\bibfnamefont {L.}~\bibnamefont {Piroli}},\
  }\href {\doibase 10.1103/PhysRevB.107.035148} {\bibfield  {journal} {\bibinfo
   {journal} {Phys. Rev. B}\ }\textbf {\bibinfo {volume} {107}},\ \bibinfo
  {pages} {035148} (\bibinfo {year} {2023}{\natexlab{a}})}\BibitemShut
  {NoStop}%
\bibitem [{\citenamefont {Haug}\ and\ \citenamefont
  {Piroli}(2023{\natexlab{b}})}]{HL2023stabilizer}%
  \BibitemOpen
  \bibfield  {author} {\bibinfo {author} {\bibfnamefont {T.}~\bibnamefont
  {Haug}}\ and\ \bibinfo {author} {\bibfnamefont {L.}~\bibnamefont {Piroli}},\
  }\href@noop {} {\  (\bibinfo {year} {2023}{\natexlab{b}})},\ \Eprint
  {http://arxiv.org/abs/2303.10152} {arXiv:2303.10152 [quant-ph]} \BibitemShut
  {NoStop}%
\bibitem [{\citenamefont {Jiang}\ and\ \citenamefont {Wang}(2023)}]{WangPRA23}%
  \BibitemOpen
  \bibfield  {author} {\bibinfo {author} {\bibfnamefont {J.}~\bibnamefont
  {Jiang}}\ and\ \bibinfo {author} {\bibfnamefont {X.}~\bibnamefont {Wang}},\
  }\href {\doibase 10.1103/PhysRevApplied.19.034052} {\bibfield  {journal}
  {\bibinfo  {journal} {Phys. Rev. Appl.}\ }\textbf {\bibinfo {volume} {19}},\
  \bibinfo {pages} {034052} (\bibinfo {year} {2023})}\BibitemShut {NoStop}%
\bibitem [{\citenamefont {Bu}\ \emph {et~al.}(2023{\natexlab{a}})\citenamefont
  {Bu}, \citenamefont {Gu},\ and\ \citenamefont {Jaffe}}]{BGJ23a}%
  \BibitemOpen
  \bibfield  {author} {\bibinfo {author} {\bibfnamefont {K.}~\bibnamefont
  {Bu}}, \bibinfo {author} {\bibfnamefont {W.}~\bibnamefont {Gu}}, \ and\
  \bibinfo {author} {\bibfnamefont {A.}~\bibnamefont {Jaffe}},\ }\href
  {\doibase 10.1073/pnas.2304589120} {\bibfield  {journal} {\bibinfo  {journal}
  {Proceedings of the National Academy of Sciences}\ }\textbf {\bibinfo
  {volume} {120}},\ \bibinfo {pages} {e2304589120} (\bibinfo {year}
  {2023}{\natexlab{a}})}\BibitemShut {NoStop}%
\bibitem [{\citenamefont {Bu}\ \emph {et~al.}(2023{\natexlab{b}})\citenamefont
  {Bu}, \citenamefont {Gu},\ and\ \citenamefont {Jaffe}}]{BGJ23b}%
  \BibitemOpen
  \bibfield  {author} {\bibinfo {author} {\bibfnamefont {K.}~\bibnamefont
  {Bu}}, \bibinfo {author} {\bibfnamefont {W.}~\bibnamefont {Gu}}, \ and\
  \bibinfo {author} {\bibfnamefont {A.}~\bibnamefont {Jaffe}},\ }\href
  {https://arxiv.org/abs/2302.08423} {\bibfield  {journal} {\bibinfo  {journal}
  {arXiv:2302.08423}\ } (\bibinfo {year} {2023}{\natexlab{b}})}\BibitemShut
  {NoStop}%
\bibitem [{\citenamefont {Bu}\ and\ \citenamefont {Jaffe}(2025)}]{BJ24a}%
  \BibitemOpen
  \bibfield  {author} {\bibinfo {author} {\bibfnamefont {K.}~\bibnamefont
  {Bu}}\ and\ \bibinfo {author} {\bibfnamefont {A.}~\bibnamefont {Jaffe}},\
  }\href {\doibase 10.1103/PhysRevLett.134.050202} {\bibfield  {journal}
  {\bibinfo  {journal} {Phys. Rev. Lett.}\ }\textbf {\bibinfo {volume} {134}},\
  \bibinfo {pages} {050202} (\bibinfo {year} {2025})}\BibitemShut {NoStop}%
\bibitem [{\citenamefont {Bu}\ \emph {et~al.}(2025{\natexlab{a}})\citenamefont
  {Bu}, \citenamefont {Gu},\ and\ \citenamefont {Jaffe}}]{BGJ24a}%
  \BibitemOpen
  \bibfield  {author} {\bibinfo {author} {\bibfnamefont {K.}~\bibnamefont
  {Bu}}, \bibinfo {author} {\bibfnamefont {W.}~\bibnamefont {Gu}}, \ and\
  \bibinfo {author} {\bibfnamefont {A.}~\bibnamefont {Jaffe}},\ }\href
  {\doibase 10.1109/TIT.2025.3543276} {\bibfield  {journal} {\bibinfo
  {journal} {IEEE Transactions on Information Theory}\ }\textbf {\bibinfo
  {volume} {71}},\ \bibinfo {pages} {2726} (\bibinfo {year}
  {2025}{\natexlab{a}})}\BibitemShut {NoStop}%
\bibitem [{\citenamefont {Bu}\ \emph {et~al.}(2025{\natexlab{b}})\citenamefont
  {Bu}, \citenamefont {Gu},\ and\ \citenamefont {Jaffe}}]{BGJ25a}%
  \BibitemOpen
  \bibfield  {author} {\bibinfo {author} {\bibfnamefont {K.}~\bibnamefont
  {Bu}}, \bibinfo {author} {\bibfnamefont {W.}~\bibnamefont {Gu}}, \ and\
  \bibinfo {author} {\bibfnamefont {A.}~\bibnamefont {Jaffe}},\ }\href@noop {}
  {\enquote {\bibinfo {title} {Quantum higher order {F}ourier analysis and the
  {C}lifford hierarchy},}\ } (\bibinfo {year} {work in progress,
  2025}{\natexlab{b}})\BibitemShut {NoStop}%
\bibitem [{\citenamefont {Valiant}(2001)}]{valiant2001quantum}%
  \BibitemOpen
  \bibfield  {author} {\bibinfo {author} {\bibfnamefont {L.~G.}\ \bibnamefont
  {Valiant}},\ }in\ \href {\doibase 10.1145/380752.380785} {\emph {\bibinfo
  {booktitle} {Proceedings of the Thirty-Third Annual ACM Symposium on Theory
  of Computing}}},\ \bibinfo {series and number} {STOC '01}\ (\bibinfo
  {publisher} {Association for Computing Machinery},\ \bibinfo {address} {New
  York, NY, USA},\ \bibinfo {year} {2001})\ p.\ \bibinfo {pages}
  {114–123}\BibitemShut {NoStop}%
\bibitem [{\citenamefont {Valiant}(2002)}]{valiant2002quantum}%
  \BibitemOpen
  \bibfield  {author} {\bibinfo {author} {\bibfnamefont {L.~G.}\ \bibnamefont
  {Valiant}},\ }\href@noop {} {\bibfield  {journal} {\bibinfo  {journal} {SIAM
  Journal on Computing}\ }\textbf {\bibinfo {volume} {31}},\ \bibinfo {pages}
  {1229} (\bibinfo {year} {2002})}\BibitemShut {NoStop}%
\bibitem [{\citenamefont {Bravyi}\ and\ \citenamefont
  {Kitaev}(2002)}]{bravyi2002fermionic}%
  \BibitemOpen
  \bibfield  {author} {\bibinfo {author} {\bibfnamefont {S.~B.}\ \bibnamefont
  {Bravyi}}\ and\ \bibinfo {author} {\bibfnamefont {A.~Y.}\ \bibnamefont
  {Kitaev}},\ }\href {\doibase 10.1006/aphy.2002.6254} {\bibfield  {journal}
  {\bibinfo  {journal} {Annals of Physics}\ }\textbf {\bibinfo {volume}
  {298}},\ \bibinfo {pages} {210} (\bibinfo {year} {2002})}\BibitemShut
  {NoStop}%
\bibitem [{\citenamefont {Terhal}\ and\ \citenamefont
  {DiVincenzo}(2002)}]{terhal2002classical}%
  \BibitemOpen
  \bibfield  {author} {\bibinfo {author} {\bibfnamefont {B.~M.}\ \bibnamefont
  {Terhal}}\ and\ \bibinfo {author} {\bibfnamefont {D.~P.}\ \bibnamefont
  {DiVincenzo}},\ }\href {\doibase 10.1103/PhysRevA.65.032325} {\bibfield
  {journal} {\bibinfo  {journal} {Phys. Rev. A}\ }\textbf {\bibinfo {volume}
  {65}},\ \bibinfo {pages} {032325} (\bibinfo {year} {2002})}\BibitemShut
  {NoStop}%
\bibitem [{\citenamefont {DiVincenzo}\ and\ \citenamefont
  {Terhal}(2004)}]{divincenzo2004fermionic}%
  \BibitemOpen
  \bibfield  {author} {\bibinfo {author} {\bibfnamefont {D.~P.}\ \bibnamefont
  {DiVincenzo}}\ and\ \bibinfo {author} {\bibfnamefont {B.~M.}\ \bibnamefont
  {Terhal}},\ }\href {https://doi.org/10.1007/s10701-005-8657-0} {\bibfield
  {journal} {\bibinfo  {journal} {Found. Phys.}\ }\textbf {\bibinfo {volume}
  {35}},\ \bibinfo {pages} {1967} (\bibinfo {year} {2004})}\BibitemShut
  {NoStop}%
\bibitem [{\citenamefont {Bartlett}\ and\ \citenamefont
  {Sanders}(2002)}]{Bartlett02}%
  \BibitemOpen
  \bibfield  {author} {\bibinfo {author} {\bibfnamefont {S.~D.}\ \bibnamefont
  {Bartlett}}\ and\ \bibinfo {author} {\bibfnamefont {B.~C.}\ \bibnamefont
  {Sanders}},\ }\href {\doibase 10.1103/PhysRevA.65.042304} {\bibfield
  {journal} {\bibinfo  {journal} {Phys. Rev. A}\ }\textbf {\bibinfo {volume}
  {65}},\ \bibinfo {pages} {042304} (\bibinfo {year} {2002})}\BibitemShut
  {NoStop}%
\bibitem [{\citenamefont {Mari}\ and\ \citenamefont {Eisert}(2012)}]{Mari12}%
  \BibitemOpen
  \bibfield  {author} {\bibinfo {author} {\bibfnamefont {A.}~\bibnamefont
  {Mari}}\ and\ \bibinfo {author} {\bibfnamefont {J.}~\bibnamefont {Eisert}},\
  }\href {\doibase 10.1103/PhysRevLett.109.230503} {\bibfield  {journal}
  {\bibinfo  {journal} {Phys. Rev. Lett.}\ }\textbf {\bibinfo {volume} {109}},\
  \bibinfo {pages} {230503} (\bibinfo {year} {2012})}\BibitemShut {NoStop}%
\bibitem [{\citenamefont {Veitch}\ \emph {et~al.}(2013)\citenamefont {Veitch},
  \citenamefont {Wiebe}, \citenamefont {Ferrie},\ and\ \citenamefont
  {Emerson}}]{Veitch_2013}%
  \BibitemOpen
  \bibfield  {author} {\bibinfo {author} {\bibfnamefont {V.}~\bibnamefont
  {Veitch}}, \bibinfo {author} {\bibfnamefont {N.}~\bibnamefont {Wiebe}},
  \bibinfo {author} {\bibfnamefont {C.}~\bibnamefont {Ferrie}}, \ and\ \bibinfo
  {author} {\bibfnamefont {J.}~\bibnamefont {Emerson}},\ }\href {\doibase
  10.1088/1367-2630/15/1/013037} {\bibfield  {journal} {\bibinfo  {journal}
  {New J. Phys.}\ }\textbf {\bibinfo {volume} {15}},\ \bibinfo {pages} {013037}
  (\bibinfo {year} {2013})}\BibitemShut {NoStop}%
\bibitem [{\citenamefont {Lyu}\ and\ \citenamefont
  {Bu}(2024{\natexlab{a}})}]{lyu2024fermionic}%
  \BibitemOpen
  \bibfield  {author} {\bibinfo {author} {\bibfnamefont {N.}~\bibnamefont
  {Lyu}}\ and\ \bibinfo {author} {\bibfnamefont {K.}~\bibnamefont {Bu}},\
  }\href {https://arxiv.org/abs/2409.08180} {\bibfield  {journal} {\bibinfo
  {journal} {arXiv preprint arXiv:2409.08180}\ } (\bibinfo {year}
  {2024}{\natexlab{a}})}\BibitemShut {NoStop}%
\bibitem [{\citenamefont {Lyu}\ and\ \citenamefont
  {Bu}(2024{\natexlab{b}})}]{lyu2024fermionic_G}%
  \BibitemOpen
  \bibfield  {author} {\bibinfo {author} {\bibfnamefont {N.}~\bibnamefont
  {Lyu}}\ and\ \bibinfo {author} {\bibfnamefont {K.}~\bibnamefont {Bu}},\
  }\href {https://arxiv.org/abs/2411.18517} {\bibfield  {journal} {\bibinfo
  {journal} {arXiv preprint arXiv:2411.18517}\ } (\bibinfo {year}
  {2024}{\natexlab{b}})}\BibitemShut {NoStop}%
\bibitem [{\citenamefont {Bu}\ and\ \citenamefont {Li}(2025)}]{Bu2025CV}%
  \BibitemOpen
  \bibfield  {author} {\bibinfo {author} {\bibfnamefont {K.}~\bibnamefont
  {Bu}}\ and\ \bibinfo {author} {\bibfnamefont {B.}~\bibnamefont {Li}},\ }\href
  {https://arxiv.org/abs/2507.10272} {\bibfield  {journal} {\bibinfo  {journal}
  {arXiv preprint arXiv:2507.10272}\ } (\bibinfo {year} {2025})}\BibitemShut
  {NoStop}%
\bibitem [{\citenamefont {Montanaro}\ and\ \citenamefont
  {Osborne}(2010)}]{montanaro2010quantum}%
  \BibitemOpen
  \bibfield  {author} {\bibinfo {author} {\bibfnamefont {A.}~\bibnamefont
  {Montanaro}}\ and\ \bibinfo {author} {\bibfnamefont {T.~J.}\ \bibnamefont
  {Osborne}},\ }\href@noop {} {\bibfield  {journal} {\bibinfo  {journal}
  {Chicago Journal of Theoretical Computer Science}\ }\textbf {\bibinfo
  {volume} {2010}} (\bibinfo {year} {2010})}\BibitemShut {NoStop}%
\bibitem [{\citenamefont {Garcia}\ \emph {et~al.}(2023)\citenamefont {Garcia},
  \citenamefont {Bu},\ and\ \citenamefont {Jaffe}}]{GBJPNAS23}%
  \BibitemOpen
  \bibfield  {author} {\bibinfo {author} {\bibfnamefont {R.~J.}\ \bibnamefont
  {Garcia}}, \bibinfo {author} {\bibfnamefont {K.}~\bibnamefont {Bu}}, \ and\
  \bibinfo {author} {\bibfnamefont {A.}~\bibnamefont {Jaffe}},\ }\href
  {\doibase 10.1073/pnas.2217031120} {\bibfield  {journal} {\bibinfo  {journal}
  {Proceedings of the National Academy of Sciences}\ }\textbf {\bibinfo
  {volume} {120}},\ \bibinfo {pages} {e2217031120} (\bibinfo {year}
  {2023})}\BibitemShut {NoStop}%
\bibitem [{\citenamefont {Jaffe}\ \emph {et~al.}(2020)\citenamefont {Jaffe},
  \citenamefont {Jiang}, \citenamefont {Liu}, \citenamefont {Ren},\ and\
  \citenamefont {Wu}}]{JaffePNAS20}%
  \BibitemOpen
  \bibfield  {author} {\bibinfo {author} {\bibfnamefont {A.}~\bibnamefont
  {Jaffe}}, \bibinfo {author} {\bibfnamefont {C.}~\bibnamefont {Jiang}},
  \bibinfo {author} {\bibfnamefont {Z.}~\bibnamefont {Liu}}, \bibinfo {author}
  {\bibfnamefont {Y.}~\bibnamefont {Ren}}, \ and\ \bibinfo {author}
  {\bibfnamefont {J.}~\bibnamefont {Wu}},\ }\href {\doibase
  10.1073/pnas.2002813117} {\bibfield  {journal} {\bibinfo  {journal}
  {Proceedings of the National Academy of Sciences}\ }\textbf {\bibinfo
  {volume} {117}},\ \bibinfo {pages} {10715} (\bibinfo {year}
  {2020})}\BibitemShut {NoStop}%
\bibitem [{\citenamefont {Gottesman}(1996)}]{Gottesman96}%
  \BibitemOpen
  \bibfield  {author} {\bibinfo {author} {\bibfnamefont {D.}~\bibnamefont
  {Gottesman}},\ }\href {\doibase 10.1103/PhysRevA.54.1862} {\bibfield
  {journal} {\bibinfo  {journal} {Phys. Rev. A}\ }\textbf {\bibinfo {volume}
  {54}},\ \bibinfo {pages} {1862} (\bibinfo {year} {1996})}\BibitemShut
  {NoStop}%
\bibitem [{\citenamefont {Marshall}\ \emph {et~al.}(1979)\citenamefont
  {Marshall}, \citenamefont {Olkin},\ and\ \citenamefont {Arnold}}]{MOA79}%
  \BibitemOpen
  \bibfield  {author} {\bibinfo {author} {\bibfnamefont {A.~W.}\ \bibnamefont
  {Marshall}}, \bibinfo {author} {\bibfnamefont {I.}~\bibnamefont {Olkin}}, \
  and\ \bibinfo {author} {\bibfnamefont {B.~C.}\ \bibnamefont {Arnold}},\
  }\href@noop {} {\emph {\bibinfo {title} {Inequalities: theory of majorization
  and its applications}}}\ (\bibinfo  {publisher} {Academic press},\ \bibinfo
  {address} {New York},\ \bibinfo {year} {1979})\BibitemShut {NoStop}%
\bibitem [{\citenamefont {Brandão}\ \emph {et~al.}(2015)\citenamefont
  {Brandão}, \citenamefont {Horodecki}, \citenamefont {Ng}, \citenamefont
  {Oppenheim},\ and\ \citenamefont {Wehner}}]{BHOW15}%
  \BibitemOpen
  \bibfield  {author} {\bibinfo {author} {\bibfnamefont {F.}~\bibnamefont
  {Brandão}}, \bibinfo {author} {\bibfnamefont {M.}~\bibnamefont {Horodecki}},
  \bibinfo {author} {\bibfnamefont {N.}~\bibnamefont {Ng}}, \bibinfo {author}
  {\bibfnamefont {J.}~\bibnamefont {Oppenheim}}, \ and\ \bibinfo {author}
  {\bibfnamefont {S.}~\bibnamefont {Wehner}},\ }\href {\doibase
  10.1073/pnas.1411728112} {\bibfield  {journal} {\bibinfo  {journal}
  {Proceedings of the National Academy of Sciences}\ }\textbf {\bibinfo
  {volume} {112}},\ \bibinfo {pages} {3275} (\bibinfo {year}
  {2015})}\BibitemShut {NoStop}%
\bibitem [{\citenamefont {Barenco}\ \emph {et~al.}(1997)\citenamefont
  {Barenco}, \citenamefont {Berthiaume}, \citenamefont {Deutsch}, \citenamefont
  {Ekert}, \citenamefont {Jozsa},\ and\ \citenamefont
  {Macchiavello}}]{barenco1997stabilization}%
  \BibitemOpen
  \bibfield  {author} {\bibinfo {author} {\bibfnamefont {A.}~\bibnamefont
  {Barenco}}, \bibinfo {author} {\bibfnamefont {A.}~\bibnamefont {Berthiaume}},
  \bibinfo {author} {\bibfnamefont {D.}~\bibnamefont {Deutsch}}, \bibinfo
  {author} {\bibfnamefont {A.}~\bibnamefont {Ekert}}, \bibinfo {author}
  {\bibfnamefont {R.}~\bibnamefont {Jozsa}}, \ and\ \bibinfo {author}
  {\bibfnamefont {C.}~\bibnamefont {Macchiavello}},\ }\href
  {https://doi.org/10.1137/S0097539796302452} {\bibfield  {journal} {\bibinfo
  {journal} {SIAM Journal on Computing}\ }\textbf {\bibinfo {volume} {26}},\
  \bibinfo {pages} {1541} (\bibinfo {year} {1997})}\BibitemShut {NoStop}%
\bibitem [{\citenamefont {Buhrman}\ \emph
  {et~al.}(2001{\natexlab{b}})\citenamefont {Buhrman}, \citenamefont {Cleve},
  \citenamefont {Watrous},\ and\ \citenamefont {de~Wolf}}]{BuhrmanPRL01}%
  \BibitemOpen
  \bibfield  {author} {\bibinfo {author} {\bibfnamefont {H.}~\bibnamefont
  {Buhrman}}, \bibinfo {author} {\bibfnamefont {R.}~\bibnamefont {Cleve}},
  \bibinfo {author} {\bibfnamefont {J.}~\bibnamefont {Watrous}}, \ and\
  \bibinfo {author} {\bibfnamefont {R.}~\bibnamefont {de~Wolf}},\ }\href
  {\doibase 10.1103/PhysRevLett.87.167902} {\bibfield  {journal} {\bibinfo
  {journal} {Phys. Rev. Lett.}\ }\textbf {\bibinfo {volume} {87}},\ \bibinfo
  {pages} {167902} (\bibinfo {year} {2001}{\natexlab{b}})}\BibitemShut
  {NoStop}%
\bibitem [{\citenamefont {De~Wolf}(2019)}]{de2019quantum}%
  \BibitemOpen
  \bibfield  {author} {\bibinfo {author} {\bibfnamefont {R.}~\bibnamefont
  {De~Wolf}},\ }\href {https://arxiv.org/abs/1907.09415} {\bibfield  {journal}
  {\bibinfo  {journal} {arXiv:1907.09415}\ } (\bibinfo {year}
  {2019})}\BibitemShut {NoStop}%
\bibitem [{\citenamefont {Evered}\ \emph {et~al.}(2023)\citenamefont {Evered},
  \citenamefont {Bluvstein}, \citenamefont {Kalinowski}, \citenamefont {Ebadi},
  \citenamefont {Manovitz}, \citenamefont {Zhou}, \citenamefont {Li},
  \citenamefont {Geim}, \citenamefont {Wang}, \citenamefont {Maskara} \emph
  {et~al.}}]{evered2023highfidelity}%
  \BibitemOpen
  \bibfield  {author} {\bibinfo {author} {\bibfnamefont {S.~J.}\ \bibnamefont
  {Evered}}, \bibinfo {author} {\bibfnamefont {D.}~\bibnamefont {Bluvstein}},
  \bibinfo {author} {\bibfnamefont {M.}~\bibnamefont {Kalinowski}}, \bibinfo
  {author} {\bibfnamefont {S.}~\bibnamefont {Ebadi}}, \bibinfo {author}
  {\bibfnamefont {T.}~\bibnamefont {Manovitz}}, \bibinfo {author}
  {\bibfnamefont {H.}~\bibnamefont {Zhou}}, \bibinfo {author} {\bibfnamefont
  {S.~H.}\ \bibnamefont {Li}}, \bibinfo {author} {\bibfnamefont {A.~A.}\
  \bibnamefont {Geim}}, \bibinfo {author} {\bibfnamefont {T.~T.}\ \bibnamefont
  {Wang}}, \bibinfo {author} {\bibfnamefont {N.}~\bibnamefont {Maskara}},
  \emph {et~al.},\ }\href {https://doi.org/10.1038/s41586-023-06481-y}
  {\bibfield  {journal} {\bibinfo  {journal} {Nature}\ }\textbf {\bibinfo
  {volume} {622}},\ \bibinfo {pages} {268} (\bibinfo {year}
  {2023})}\BibitemShut {NoStop}%
\bibitem [{\citenamefont {Araki}\ and\ \citenamefont
  {Lieb}(1970)}]{araki1970entropy}%
  \BibitemOpen
  \bibfield  {author} {\bibinfo {author} {\bibfnamefont {H.}~\bibnamefont
  {Araki}}\ and\ \bibinfo {author} {\bibfnamefont {E.~H.}\ \bibnamefont
  {Lieb}},\ }\href {\doibase 10.1007/BF01646092} {\bibfield  {journal}
  {\bibinfo  {journal} {Communications in Mathematical Physics}\ }\textbf
  {\bibinfo {volume} {18}},\ \bibinfo {pages} {160} (\bibinfo {year}
  {1970})}\BibitemShut {NoStop}%
\bibitem [{\citenamefont {Hiai}\ \emph {et~al.}(2011)\citenamefont {Hiai},
  \citenamefont {Mosonyi}, \citenamefont {Petz},\ and\ \citenamefont
  {B\'{e}ny}}]{Hiai11}%
  \BibitemOpen
  \bibfield  {author} {\bibinfo {author} {\bibfnamefont {F.}~\bibnamefont
  {Hiai}}, \bibinfo {author} {\bibfnamefont {M.}~\bibnamefont {Mosonyi}},
  \bibinfo {author} {\bibfnamefont {D.}~\bibnamefont {Petz}}, \ and\ \bibinfo
  {author} {\bibfnamefont {C.}~\bibnamefont {B\'{e}ny}},\ }\href {\doibase
  10.1142/S0129055X11004412} {\bibfield  {journal} {\bibinfo  {journal} {Rev.
  Math. Phys.}\ }\textbf {\bibinfo {volume} {23}},\ \bibinfo {pages} {691}
  (\bibinfo {year} {2011})}\BibitemShut {NoStop}%
\bibitem [{\citenamefont {Müller-Lennert}\ \emph {et~al.}(2013)\citenamefont
  {Müller-Lennert}, \citenamefont {Dupuis}, \citenamefont {Szehr},
  \citenamefont {Fehr},\ and\ \citenamefont {Tomamichel}}]{Martin13}%
  \BibitemOpen
  \bibfield  {author} {\bibinfo {author} {\bibfnamefont {M.}~\bibnamefont
  {Müller-Lennert}}, \bibinfo {author} {\bibfnamefont {F.}~\bibnamefont
  {Dupuis}}, \bibinfo {author} {\bibfnamefont {O.}~\bibnamefont {Szehr}},
  \bibinfo {author} {\bibfnamefont {S.}~\bibnamefont {Fehr}}, \ and\ \bibinfo
  {author} {\bibfnamefont {M.}~\bibnamefont {Tomamichel}},\ }\href {\doibase
  10.1063/1.4838856} {\bibfield  {journal} {\bibinfo  {journal} {J. Math.
  Phys.}\ }\textbf {\bibinfo {volume} {54}},\ \bibinfo {pages} {122203}
  (\bibinfo {year} {2013})}\BibitemShut {NoStop}%
\bibitem [{\citenamefont {Araki}(1976)}]{araki1976relative}%
  \BibitemOpen
  \bibfield  {author} {\bibinfo {author} {\bibfnamefont {H.}~\bibnamefont
  {Araki}},\ }\href {\doibase 10.2977/prims/1195191148} {\bibfield  {journal}
  {\bibinfo  {journal} {Publications of the Research Institute for Mathematical
  Sciences}\ }\textbf {\bibinfo {volume} {11}},\ \bibinfo {pages} {809}
  (\bibinfo {year} {1976})}\BibitemShut {NoStop}%
\bibitem [{\citenamefont {Tomamichel}(2015)}]{Tomamichel2015quantum}%
  \BibitemOpen
  \bibfield  {author} {\bibinfo {author} {\bibfnamefont {M.}~\bibnamefont
  {Tomamichel}},\ }\href {\doibase 10.1007/978-3-319-21891-5} {\emph {\bibinfo
  {title} {Quantum information processing with finite resources: mathematical
  foundations}}},\ Vol.~\bibinfo {volume} {5}\ (\bibinfo  {publisher}
  {Springer},\ \bibinfo {year} {2015})\BibitemShut {NoStop}%
\bibitem [{\citenamefont {Vedral}\ and\ \citenamefont
  {Plenio}(1998)}]{VedralPRA98}%
  \BibitemOpen
  \bibfield  {author} {\bibinfo {author} {\bibfnamefont {V.}~\bibnamefont
  {Vedral}}\ and\ \bibinfo {author} {\bibfnamefont {M.~B.}\ \bibnamefont
  {Plenio}},\ }\href {\doibase 10.1103/PhysRevA.57.1619} {\bibfield  {journal}
  {\bibinfo  {journal} {Phys. Rev. A}\ }\textbf {\bibinfo {volume} {57}},\
  \bibinfo {pages} {1619} (\bibinfo {year} {1998})}\BibitemShut {NoStop}%
\bibitem [{\citenamefont {Datta}(2009)}]{DattaIEEE09}%
  \BibitemOpen
  \bibfield  {author} {\bibinfo {author} {\bibfnamefont {N.}~\bibnamefont
  {Datta}},\ }\href {\doibase 10.1109/TIT.2009.2018325} {\bibfield  {journal}
  {\bibinfo  {journal} {IEEE Transactions on Information Theory}\ }\textbf
  {\bibinfo {volume} {55}},\ \bibinfo {pages} {2816} (\bibinfo {year}
  {2009})}\BibitemShut {NoStop}%
\bibitem [{\citenamefont {Arunachalam}\ \emph {et~al.}(2022)\citenamefont
  {Arunachalam}, \citenamefont {Bravyi}, \citenamefont {Nirkhe},\ and\
  \citenamefont {O'Gorman}}]{arunachalam2022parameterized}%
  \BibitemOpen
  \bibfield  {author} {\bibinfo {author} {\bibfnamefont {S.}~\bibnamefont
  {Arunachalam}}, \bibinfo {author} {\bibfnamefont {S.}~\bibnamefont {Bravyi}},
  \bibinfo {author} {\bibfnamefont {C.}~\bibnamefont {Nirkhe}}, \ and\ \bibinfo
  {author} {\bibfnamefont {B.}~\bibnamefont {O'Gorman}},\ }\href@noop {}
  {\enquote {\bibinfo {title} {The parameterized complexity of quantum
  verification},}\ } (\bibinfo {year} {2022}),\ \Eprint
  {http://arxiv.org/abs/2202.08119} {arXiv:2202.08119 [quant-ph]} \BibitemShut
  {NoStop}%
\bibitem [{\citenamefont {Impagliazzo}\ and\ \citenamefont
  {Paturi}(2001)}]{impagliazzo2001complexity}%
  \BibitemOpen
  \bibfield  {author} {\bibinfo {author} {\bibfnamefont {R.}~\bibnamefont
  {Impagliazzo}}\ and\ \bibinfo {author} {\bibfnamefont {R.}~\bibnamefont
  {Paturi}},\ }\href {\doibase 10.1006/jcss.2000.1727} {\bibfield  {journal}
  {\bibinfo  {journal} {Journal of Computer and System Sciences}\ }\textbf
  {\bibinfo {volume} {62}},\ \bibinfo {pages} {367} (\bibinfo {year}
  {2001})}\BibitemShut {NoStop}%
\bibitem [{\citenamefont {Choi}(1975)}]{Choi75}%
  \BibitemOpen
  \bibfield  {author} {\bibinfo {author} {\bibfnamefont {M.-D.}\ \bibnamefont
  {Choi}},\ }\href {\doibase 10.1016/0024-3795(75)90075-0} {\bibfield
  {journal} {\bibinfo  {journal} {Linear Algebra and its Application}\ }\textbf
  {\bibinfo {volume} {10}},\ \bibinfo {pages} {285} (\bibinfo {year}
  {1975})}\BibitemShut {NoStop}%
\bibitem [{\citenamefont {Jamiołkowski}(1972)}]{Jamio72}%
  \BibitemOpen
  \bibfield  {author} {\bibinfo {author} {\bibfnamefont {A.}~\bibnamefont
  {Jamiołkowski}},\ }\href {\doibase 10.1016/0034-4877(72)90011-0} {\bibfield
  {journal} {\bibinfo  {journal} {Rep. Math. Phys.}\ }\textbf {\bibinfo
  {volume} {3}},\ \bibinfo {pages} {275–278} (\bibinfo {year}
  {1972})}\BibitemShut {NoStop}%
\bibitem [{\citenamefont {Aharonov}\ and\ \citenamefont
  {Eldar}(2015)}]{AharonovQLTC15}%
  \BibitemOpen
  \bibfield  {author} {\bibinfo {author} {\bibfnamefont {D.}~\bibnamefont
  {Aharonov}}\ and\ \bibinfo {author} {\bibfnamefont {L.}~\bibnamefont
  {Eldar}},\ }\href {\doibase 10.1137/140975498} {\bibfield  {journal}
  {\bibinfo  {journal} {SIAM Journal on Computing}\ }\textbf {\bibinfo {volume}
  {44}},\ \bibinfo {pages} {1230} (\bibinfo {year} {2015})}\BibitemShut
  {NoStop}%
\bibitem [{\citenamefont {Eldar}\ and\ \citenamefont {Harrow}(2017)}]{Eldar17}%
  \BibitemOpen
  \bibfield  {author} {\bibinfo {author} {\bibfnamefont {L.}~\bibnamefont
  {Eldar}}\ and\ \bibinfo {author} {\bibfnamefont {A.~W.}\ \bibnamefont
  {Harrow}},\ }in\ \href {\doibase 10.1109/FOCS.2017.46} {\emph {\bibinfo
  {booktitle} {2017 IEEE 58th Annual Symposium on Foundations of Computer
  Science (FOCS)}}}\ (\bibinfo {year} {2017})\ pp.\ \bibinfo {pages}
  {427--438}\BibitemShut {NoStop}%
\bibitem [{\citenamefont {Aharonov}\ \emph {et~al.}(2013)\citenamefont
  {Aharonov}, \citenamefont {Arad},\ and\ \citenamefont
  {Vidick}}]{AharonovQPCP13}%
  \BibitemOpen
  \bibfield  {author} {\bibinfo {author} {\bibfnamefont {D.}~\bibnamefont
  {Aharonov}}, \bibinfo {author} {\bibfnamefont {I.}~\bibnamefont {Arad}}, \
  and\ \bibinfo {author} {\bibfnamefont {T.}~\bibnamefont {Vidick}},\ }\href
  {\doibase 10.1145/2491533.2491549} {\bibfield  {journal} {\bibinfo  {journal}
  {SIGACT News}\ }\textbf {\bibinfo {volume} {44}},\ \bibinfo {pages} {47–79}
  (\bibinfo {year} {2013})}\BibitemShut {NoStop}%
\bibitem [{\citenamefont {Appleby}(2005)}]{Appleby05}%
  \BibitemOpen
  \bibfield  {author} {\bibinfo {author} {\bibfnamefont {D.~M.}\ \bibnamefont
  {Appleby}},\ }\href {\doibase 10.1063/1.1896384} {\bibfield  {journal}
  {\bibinfo  {journal} {Journal of Mathematical Physics}\ }\textbf {\bibinfo
  {volume} {46}},\ \bibinfo {pages} {052107} (\bibinfo {year}
  {2005})}\BibitemShut {NoStop}%
\bibitem [{\citenamefont {Gross}(2006)}]{Gross06}%
  \BibitemOpen
  \bibfield  {author} {\bibinfo {author} {\bibfnamefont {D.}~\bibnamefont
  {Gross}},\ }\href {\doibase 10.1063/1.2393152} {\bibfield  {journal}
  {\bibinfo  {journal} {J. Math. Phys.}\ }\textbf {\bibinfo {volume} {47}},\
  \bibinfo {pages} {122107} (\bibinfo {year} {2006})}\BibitemShut {NoStop}%
\bibitem [{\citenamefont {Zhu}(2017)}]{ZhuPRA17}%
  \BibitemOpen
  \bibfield  {author} {\bibinfo {author} {\bibfnamefont {H.}~\bibnamefont
  {Zhu}},\ }\href {\doibase 10.1103/PhysRevA.96.062336} {\bibfield  {journal}
  {\bibinfo  {journal} {Phys. Rev. A}\ }\textbf {\bibinfo {volume} {96}},\
  \bibinfo {pages} {062336} (\bibinfo {year} {2017})}\BibitemShut {NoStop}%
\bibitem [{\citenamefont {de~Beaudrap}(2013)}]{DB13}%
  \BibitemOpen
  \bibfield  {author} {\bibinfo {author} {\bibfnamefont {N.}~\bibnamefont
  {de~Beaudrap}},\ }\href {\doibase 10.26421/QIC13.1-2-6} {\bibfield  {journal}
  {\bibinfo  {journal} {Quantum Information and Computation}\ }\textbf
  {\bibinfo {volume} {13}},\ \bibinfo {pages} {0073} (\bibinfo {year}
  {2013})}\BibitemShut {NoStop}%
\bibitem [{\citenamefont {Tsallis}(1988)}]{tsallis1988possible}%
  \BibitemOpen
  \bibfield  {author} {\bibinfo {author} {\bibfnamefont {C.}~\bibnamefont
  {Tsallis}},\ }\href {https://doi.org/10.1007/BF01016429} {\bibfield
  {journal} {\bibinfo  {journal} {Journal of Statistical Physics}\ }\textbf
  {\bibinfo {volume} {52}},\ \bibinfo {pages} {479} (\bibinfo {year}
  {1988})}\BibitemShut {NoStop}%
\bibitem [{\citenamefont {Datta}\ \emph {et~al.}(2014)\citenamefont {Datta},
  \citenamefont {Dorlas}, \citenamefont {Jozsa},\ and\ \citenamefont
  {Benatti}}]{datta2014suben}%
  \BibitemOpen
  \bibfield  {author} {\bibinfo {author} {\bibfnamefont {N.}~\bibnamefont
  {Datta}}, \bibinfo {author} {\bibfnamefont {T.}~\bibnamefont {Dorlas}},
  \bibinfo {author} {\bibfnamefont {R.}~\bibnamefont {Jozsa}}, \ and\ \bibinfo
  {author} {\bibfnamefont {F.}~\bibnamefont {Benatti}},\ }\href {\doibase
  10.1063/1.4882935} {\bibfield  {journal} {\bibinfo  {journal} {Journal of
  Mathematical Physics}\ }\textbf {\bibinfo {volume} {55}},\ \bibinfo {pages}
  {062203} (\bibinfo {year} {2014})}\BibitemShut {NoStop}%
\bibitem [{\citenamefont {Szarek}\ and\ \citenamefont
  {Voiculescu}(1996)}]{szarek1996volumes}%
  \BibitemOpen
  \bibfield  {author} {\bibinfo {author} {\bibfnamefont {S.~J.}\ \bibnamefont
  {Szarek}}\ and\ \bibinfo {author} {\bibfnamefont {D.}~\bibnamefont
  {Voiculescu}},\ }\href {\doibase 10.1007/BF02108815} {\bibfield  {journal}
  {\bibinfo  {journal} {Communications in Mathematical Physics}\ }\textbf
  {\bibinfo {volume} {178}},\ \bibinfo {pages} {563} (\bibinfo {year}
  {1996})}\BibitemShut {NoStop}%
\bibitem [{\citenamefont {Shlyakhtenko}\ and\ \citenamefont
  {Schultz}(2007)}]{shlyakhtenko2007shannon}%
  \BibitemOpen
  \bibfield  {author} {\bibinfo {author} {\bibfnamefont {D.}~\bibnamefont
  {Shlyakhtenko}}\ and\ \bibinfo {author} {\bibfnamefont {H.}~\bibnamefont
  {Schultz}},\ }\href {\doibase 10.1073/pnas.0706451104} {\bibfield  {journal}
  {\bibinfo  {journal} {Proceedings of the National Academy of Sciences}\
  }\textbf {\bibinfo {volume} {104}},\ \bibinfo {pages} {15254} (\bibinfo
  {year} {2007})}\BibitemShut {NoStop}%
\bibitem [{\citenamefont {Shlyakhtenko}(2007)}]{shlyakhtenko2007free}%
  \BibitemOpen
  \bibfield  {author} {\bibinfo {author} {\bibfnamefont {D.}~\bibnamefont
  {Shlyakhtenko}},\ }\href {\doibase 10.1016/j.aim.2006.03.014} {\bibfield
  {journal} {\bibinfo  {journal} {Advances in Mathematics}\ }\textbf {\bibinfo
  {volume} {208}},\ \bibinfo {pages} {824} (\bibinfo {year}
  {2007})}\BibitemShut {NoStop}%
\bibitem [{\citenamefont {König}\ and\ \citenamefont {Smith}(2014)}]{Konig14}%
  \BibitemOpen
  \bibfield  {author} {\bibinfo {author} {\bibfnamefont {R.}~\bibnamefont
  {König}}\ and\ \bibinfo {author} {\bibfnamefont {G.}~\bibnamefont {Smith}},\
  }\href {\doibase 10.1109/TIT.2014.2298436} {\bibfield  {journal} {\bibinfo
  {journal} {IEEE Trans. Inform. Theory}\ }\textbf {\bibinfo {volume} {60}},\
  \bibinfo {pages} {1536} (\bibinfo {year} {2014})}\BibitemShut {NoStop}%
\bibitem [{\citenamefont {De~Palma}\ \emph {et~al.}(2014)\citenamefont
  {De~Palma}, \citenamefont {Mari},\ and\ \citenamefont
  {Giovannetti}}]{Palma14}%
  \BibitemOpen
  \bibfield  {author} {\bibinfo {author} {\bibfnamefont {G.}~\bibnamefont
  {De~Palma}}, \bibinfo {author} {\bibfnamefont {A.}~\bibnamefont {Mari}}, \
  and\ \bibinfo {author} {\bibfnamefont {V.}~\bibnamefont {Giovannetti}},\
  }\href {\doibase https://doi.org/10.1038/nphoton.2014.252} {\bibfield
  {journal} {\bibinfo  {journal} {Nature Photon}\ }\textbf {\bibinfo {volume}
  {8}},\ \bibinfo {pages} {958–964} (\bibinfo {year} {2014})}\BibitemShut
  {NoStop}%
\bibitem [{\citenamefont {Huang}\ \emph {et~al.}(2022)\citenamefont {Huang},
  \citenamefont {Liu},\ and\ \citenamefont {Wu}}]{HuangLiuWu22}%
  \BibitemOpen
  \bibfield  {author} {\bibinfo {author} {\bibfnamefont {L.}~\bibnamefont
  {Huang}}, \bibinfo {author} {\bibfnamefont {Z.}~\bibnamefont {Liu}}, \ and\
  \bibinfo {author} {\bibfnamefont {J.}~\bibnamefont {Wu}},\ }\href
  {https://doi.org/10.48550/arXiv.2204.04401} {\bibfield  {journal} {\bibinfo
  {journal} {arXiv:2204.04401}\ } (\bibinfo {year} {2022})}\BibitemShut
  {NoStop}%
\bibitem [{\citenamefont {Audenaert}\ \emph {et~al.}(2016)\citenamefont
  {Audenaert}, \citenamefont {Datta},\ and\ \citenamefont
  {Ozols}}]{Audenaert16}%
  \BibitemOpen
  \bibfield  {author} {\bibinfo {author} {\bibfnamefont {K.}~\bibnamefont
  {Audenaert}}, \bibinfo {author} {\bibfnamefont {N.}~\bibnamefont {Datta}}, \
  and\ \bibinfo {author} {\bibfnamefont {M.}~\bibnamefont {Ozols}},\ }\href
  {\doibase 10.1063/1.4950785} {\bibfield  {journal} {\bibinfo  {journal} {J.
  Math. Phys.}\ }\textbf {\bibinfo {volume} {57}},\ \bibinfo {pages} {052202}
  (\bibinfo {year} {2016})}\BibitemShut {NoStop}%
\bibitem [{\citenamefont {Carlen}\ \emph {et~al.}(2016)\citenamefont {Carlen},
  \citenamefont {Lieb},\ and\ \citenamefont {Loss}}]{Carlen16}%
  \BibitemOpen
  \bibfield  {author} {\bibinfo {author} {\bibfnamefont {E.~A.}\ \bibnamefont
  {Carlen}}, \bibinfo {author} {\bibfnamefont {E.~H.}\ \bibnamefont {Lieb}}, \
  and\ \bibinfo {author} {\bibfnamefont {M.}~\bibnamefont {Loss}},\ }\href
  {\doibase 10.1063/1.4953638} {\bibfield  {journal} {\bibinfo  {journal} {J.
  Math. Phys.}\ }\textbf {\bibinfo {volume} {57}},\ \bibinfo {pages} {062203}
  (\bibinfo {year} {2016})}\BibitemShut {NoStop}%
\bibitem [{\citenamefont {Cushen}\ and\ \citenamefont
  {Hudson}(1971)}]{Cushen71}%
  \BibitemOpen
  \bibfield  {author} {\bibinfo {author} {\bibfnamefont {C.}~\bibnamefont
  {Cushen}}\ and\ \bibinfo {author} {\bibfnamefont {R.}~\bibnamefont
  {Hudson}},\ }\href {\doibase doi:10.2307/3212170} {\bibfield  {journal}
  {\bibinfo  {journal} {J. Appl. Probab.}\ }\textbf {\bibinfo {volume} {8}},\
  \bibinfo {pages} {454} (\bibinfo {year} {1971})}\BibitemShut {NoStop}%
\bibitem [{\citenamefont {Hepp}\ and\ \citenamefont
  {Lieb}(1973{\natexlab{a}})}]{Lieb73}%
  \BibitemOpen
  \bibfield  {author} {\bibinfo {author} {\bibfnamefont {K.}~\bibnamefont
  {Hepp}}\ and\ \bibinfo {author} {\bibfnamefont {E.}~\bibnamefont {Lieb}},\
  }\href {\doibase 10.5169/seals-114496} {\bibfield  {journal} {\bibinfo
  {journal} {Helv. Phys. Acta}\ }\textbf {\bibinfo {volume} {46}},\ \bibinfo
  {pages} {573–603} (\bibinfo {year} {1973}{\natexlab{a}})}\BibitemShut
  {NoStop}%
\bibitem [{\citenamefont {Hepp}\ and\ \citenamefont
  {Lieb}(1973{\natexlab{b}})}]{Lieb1973}%
  \BibitemOpen
  \bibfield  {author} {\bibinfo {author} {\bibfnamefont {K.}~\bibnamefont
  {Hepp}}\ and\ \bibinfo {author} {\bibfnamefont {E.}~\bibnamefont {Lieb}},\
  }\href {\doibase 10.1016/0003-4916(73)90039-0} {\bibfield  {journal}
  {\bibinfo  {journal} {Ann. Phys.}\ }\textbf {\bibinfo {volume} {76}},\
  \bibinfo {pages} {360} (\bibinfo {year} {1973}{\natexlab{b}})}\BibitemShut
  {NoStop}%
\bibitem [{\citenamefont {Giri}\ and\ \citenamefont {von
  Waldenfels}(1978)}]{Giri78}%
  \BibitemOpen
  \bibfield  {author} {\bibinfo {author} {\bibfnamefont {N.}~\bibnamefont
  {Giri}}\ and\ \bibinfo {author} {\bibfnamefont {W.}~\bibnamefont {von
  Waldenfels}},\ }\href {\doibase 10.1007/BF00536048} {\bibfield  {journal}
  {\bibinfo  {journal} {Probab. Theory Relat. Fields}\ }\textbf {\bibinfo
  {volume} {42}},\ \bibinfo {pages} {129} (\bibinfo {year} {1978})}\BibitemShut
  {NoStop}%
\bibitem [{\citenamefont {Goderis}\ and\ \citenamefont
  {Vets}(1978)}]{Goderis89}%
  \BibitemOpen
  \bibfield  {author} {\bibinfo {author} {\bibfnamefont {D.}~\bibnamefont
  {Goderis}}\ and\ \bibinfo {author} {\bibfnamefont {P.}~\bibnamefont {Vets}},\
  }\href {\doibase 10.1007/BF01257415} {\bibfield  {journal} {\bibinfo
  {journal} {Commun. Math. Phys.}\ }\textbf {\bibinfo {volume} {122}},\
  \bibinfo {pages} {249} (\bibinfo {year} {1978})}\BibitemShut {NoStop}%
\bibitem [{\citenamefont {Matsui}(2002)}]{Matsui02}%
  \BibitemOpen
  \bibfield  {author} {\bibinfo {author} {\bibfnamefont {T.}~\bibnamefont
  {Matsui}},\ }\href {\doibase 10.1142/S0129055X02001272} {\bibfield  {journal}
  {\bibinfo  {journal} {Rev. Math. Phys.}\ }\textbf {\bibinfo {volume} {14}},\
  \bibinfo {pages} {675–700} (\bibinfo {year} {2002})}\BibitemShut {NoStop}%
\bibitem [{\citenamefont {Cramer}\ and\ \citenamefont
  {Eisert}(2010)}]{Cramer10}%
  \BibitemOpen
  \bibfield  {author} {\bibinfo {author} {\bibfnamefont {M.}~\bibnamefont
  {Cramer}}\ and\ \bibinfo {author} {\bibfnamefont {J.}~\bibnamefont
  {Eisert}},\ }\href {\doibase 10.1088/1367-2630/12/5/055020} {\bibfield
  {journal} {\bibinfo  {journal} {New J. Phys.}\ }\textbf {\bibinfo {volume}
  {12}},\ \bibinfo {pages} {055020} (\bibinfo {year} {2010})}\BibitemShut
  {NoStop}%
\bibitem [{\citenamefont {Jaksic}\ \emph {et~al.}(2009)\citenamefont {Jaksic},
  \citenamefont {Pautrat},\ and\ \citenamefont {Pille}}]{Jaksic09}%
  \BibitemOpen
  \bibfield  {author} {\bibinfo {author} {\bibfnamefont {V.}~\bibnamefont
  {Jaksic}}, \bibinfo {author} {\bibfnamefont {Y.}~\bibnamefont {Pautrat}}, \
  and\ \bibinfo {author} {\bibfnamefont {C.-A.}\ \bibnamefont {Pille}},\ }\href
  {\doibase 10.1007/s00220-008-0610-6} {\bibfield  {journal} {\bibinfo
  {journal} {Commun. Math. Phys.}\ }\textbf {\bibinfo {volume} {285}},\
  \bibinfo {pages} {175} (\bibinfo {year} {2009})}\BibitemShut {NoStop}%
\bibitem [{\citenamefont {Arous}\ \emph {et~al.}(2013)\citenamefont {Arous},
  \citenamefont {Kirkpatrick},\ and\ \citenamefont {Schlein}}]{Arous13}%
  \BibitemOpen
  \bibfield  {author} {\bibinfo {author} {\bibfnamefont {G.~B.}\ \bibnamefont
  {Arous}}, \bibinfo {author} {\bibfnamefont {K.}~\bibnamefont {Kirkpatrick}},
  \ and\ \bibinfo {author} {\bibfnamefont {B.}~\bibnamefont {Schlein}},\ }\href
  {\doibase 10.1007/s00220-013-1722-1} {\bibfield  {journal} {\bibinfo
  {journal} {Commun. Math. Phys.}\ }\textbf {\bibinfo {volume} {321}},\
  \bibinfo {pages} {371} (\bibinfo {year} {2013})}\BibitemShut {NoStop}%
\bibitem [{\citenamefont {Michoel}\ and\ \citenamefont
  {Nachtergaele}(2004)}]{Michoel04}%
  \BibitemOpen
  \bibfield  {author} {\bibinfo {author} {\bibfnamefont {T.}~\bibnamefont
  {Michoel}}\ and\ \bibinfo {author} {\bibfnamefont {B.}~\bibnamefont
  {Nachtergaele}},\ }\href {\doibase 10.1007/s00440-004-0364-9} {\bibfield
  {journal} {\bibinfo  {journal} {Probab. Theory Relat. Fields}\ }\textbf
  {\bibinfo {volume} {130}},\ \bibinfo {pages} {493} (\bibinfo {year}
  {2004})}\BibitemShut {NoStop}%
\bibitem [{\citenamefont {Goderis}\ \emph {et~al.}(1989)\citenamefont
  {Goderis}, \citenamefont {Verbeure},\ and\ \citenamefont
  {Vets}}]{GoderisPTRT89}%
  \BibitemOpen
  \bibfield  {author} {\bibinfo {author} {\bibfnamefont {D.}~\bibnamefont
  {Goderis}}, \bibinfo {author} {\bibfnamefont {A.}~\bibnamefont {Verbeure}}, \
  and\ \bibinfo {author} {\bibfnamefont {P.}~\bibnamefont {Vets}},\ }\href
  {\doibase 10.1007/BF00341282} {\bibfield  {journal} {\bibinfo  {journal}
  {Probab. Theory Relat. Fields}\ }\textbf {\bibinfo {volume} {82}},\ \bibinfo
  {pages} {527} (\bibinfo {year} {1989})}\BibitemShut {NoStop}%
\bibitem [{\citenamefont {Jakšić}\ \emph {et~al.}(2010)\citenamefont
  {Jakšić}, \citenamefont {Pautrat},\ and\ \citenamefont
  {Pillet}}]{JaksicJMP10}%
  \BibitemOpen
  \bibfield  {author} {\bibinfo {author} {\bibfnamefont {V.}~\bibnamefont
  {Jakšić}}, \bibinfo {author} {\bibfnamefont {Y.}~\bibnamefont {Pautrat}}, \
  and\ \bibinfo {author} {\bibfnamefont {C.-A.}\ \bibnamefont {Pillet}},\
  }\href {\doibase 10.1063/1.3285287} {\bibfield  {journal} {\bibinfo
  {journal} {J. Math. Phys.}\ }\textbf {\bibinfo {volume} {51}},\ \bibinfo
  {pages} {015208} (\bibinfo {year} {2010})}\BibitemShut {NoStop}%
\bibitem [{\citenamefont {Accardi}\ and\ \citenamefont {Lu}(1994)}]{Accardi94}%
  \BibitemOpen
  \bibfield  {author} {\bibinfo {author} {\bibfnamefont {L.}~\bibnamefont
  {Accardi}}\ and\ \bibinfo {author} {\bibfnamefont {Y.~G.}\ \bibnamefont
  {Lu}},\ }\href {\doibase 10.1007/BF01874132} {\bibfield  {journal} {\bibinfo
  {journal} {Acta Math. Hung.}\ }\textbf {\bibinfo {volume} {63}},\ \bibinfo
  {pages} {249} (\bibinfo {year} {1994})}\BibitemShut {NoStop}%
\bibitem [{\citenamefont {Liu}(2016)}]{Liu16}%
  \BibitemOpen
  \bibfield  {author} {\bibinfo {author} {\bibfnamefont {Z.}~\bibnamefont
  {Liu}},\ }\href {\doibase 10.1090/tran/6582} {\bibfield  {journal} {\bibinfo
  {journal} {Transactions of the American Mathematical Society}\ }\textbf
  {\bibinfo {volume} {368}},\ \bibinfo {pages} {8303} (\bibinfo {year}
  {2016})}\BibitemShut {NoStop}%
\bibitem [{\citenamefont {Jiang}\ \emph {et~al.}(2019)\citenamefont {Jiang},
  \citenamefont {Liu},\ and\ \citenamefont {Wu}}]{JiangLiuWu19}%
  \BibitemOpen
  \bibfield  {author} {\bibinfo {author} {\bibfnamefont {C.}~\bibnamefont
  {Jiang}}, \bibinfo {author} {\bibfnamefont {Z.}~\bibnamefont {Liu}}, \ and\
  \bibinfo {author} {\bibfnamefont {J.}~\bibnamefont {Wu}},\ }\href {\doibase
  10.1007/s11425-017-9263-7} {\bibfield  {journal} {\bibinfo  {journal}
  {Science China Mathematics}\ }\textbf {\bibinfo {volume} {62}},\ \bibinfo
  {pages} {1585} (\bibinfo {year} {2019})}\BibitemShut {NoStop}%
\bibitem [{\citenamefont {Hayashi}(2009)}]{Hayashi09}%
  \BibitemOpen
  \bibfield  {author} {\bibinfo {author} {\bibfnamefont {M.}~\bibnamefont
  {Hayashi}},\ }\href {\doibase 10.1090/trans2/227/05} {\bibfield  {journal}
  {\bibinfo  {journal} {Am. Math. Soc. Trans. Ser.}\ }\textbf {\bibinfo
  {volume} {2}},\ \bibinfo {pages} {95–123} (\bibinfo {year}
  {2009})}\BibitemShut {NoStop}%
\bibitem [{\citenamefont {Campbell}\ \emph {et~al.}(2013)\citenamefont
  {Campbell}, \citenamefont {Genoni},\ and\ \citenamefont
  {Eisert}}]{CampbellPRA13}%
  \BibitemOpen
  \bibfield  {author} {\bibinfo {author} {\bibfnamefont {E.~T.}\ \bibnamefont
  {Campbell}}, \bibinfo {author} {\bibfnamefont {M.~G.}\ \bibnamefont
  {Genoni}}, \ and\ \bibinfo {author} {\bibfnamefont {J.}~\bibnamefont
  {Eisert}},\ }\href {\doibase 10.1103/PhysRevA.87.042330} {\bibfield
  {journal} {\bibinfo  {journal} {Phys. Rev. A}\ }\textbf {\bibinfo {volume}
  {87}},\ \bibinfo {pages} {042330} (\bibinfo {year} {2013})}\BibitemShut
  {NoStop}%
\bibitem [{\citenamefont {Becker}\ \emph {et~al.}(2021)\citenamefont {Becker},
  \citenamefont {Datta}, \citenamefont {Lami},\ and\ \citenamefont
  {Rouzé}}]{BekerCMP21}%
  \BibitemOpen
  \bibfield  {author} {\bibinfo {author} {\bibfnamefont {S.}~\bibnamefont
  {Becker}}, \bibinfo {author} {\bibfnamefont {N.}~\bibnamefont {Datta}},
  \bibinfo {author} {\bibfnamefont {L.}~\bibnamefont {Lami}}, \ and\ \bibinfo
  {author} {\bibfnamefont {C.}~\bibnamefont {Rouzé}},\ }\href {\doibase
  10.1007/s00220-021-03988-1} {\bibfield  {journal} {\bibinfo  {journal}
  {Commun. Math. Phys.}\ }\textbf {\bibinfo {volume} {383}},\ \bibinfo {pages}
  {223} (\bibinfo {year} {2021})}\BibitemShut {NoStop}%
\bibitem [{\citenamefont {Carbone}\ \emph {et~al.}(2022)\citenamefont
  {Carbone}, \citenamefont {Girotti},\ and\ \citenamefont
  {Melchor~Hernandez}}]{Carbone22}%
  \BibitemOpen
  \bibfield  {author} {\bibinfo {author} {\bibfnamefont {R.}~\bibnamefont
  {Carbone}}, \bibinfo {author} {\bibfnamefont {F.}~\bibnamefont {Girotti}}, \
  and\ \bibinfo {author} {\bibfnamefont {A.}~\bibnamefont
  {Melchor~Hernandez}},\ }\href {\doibase 10.1007/s10955-022-02938-y}
  {\bibfield  {journal} {\bibinfo  {journal} {Journal of Statistical Physics}\
  }\textbf {\bibinfo {volume} {188}},\ \bibinfo {pages} {8} (\bibinfo {year}
  {2022})}\BibitemShut {NoStop}%
\bibitem [{\citenamefont {Voiculescu}(1986)}]{voiculescu1986addition}%
  \BibitemOpen
  \bibfield  {author} {\bibinfo {author} {\bibfnamefont {D.}~\bibnamefont
  {Voiculescu}},\ }\href {\doibase 10.1016/0022-1236(86)90062-5} {\bibfield
  {journal} {\bibinfo  {journal} {Journal of functional analysis}\ }\textbf
  {\bibinfo {volume} {66}},\ \bibinfo {pages} {323} (\bibinfo {year}
  {1986})}\BibitemShut {NoStop}%
\bibitem [{\citenamefont {Voiculescu}(1987)}]{voiculescu1987multiplication}%
  \BibitemOpen
  \bibfield  {author} {\bibinfo {author} {\bibfnamefont {D.}~\bibnamefont
  {Voiculescu}},\ }\href {https://www.jstor.org/stable/24714784} {\bibfield
  {journal} {\bibinfo  {journal} {Journal of Operator Theory}\ ,\ \bibinfo
  {pages} {223}} (\bibinfo {year} {1987})}\BibitemShut {NoStop}%
\bibitem [{\citenamefont {Aharonov}\ \emph {et~al.}(1998)\citenamefont
  {Aharonov}, \citenamefont {Kitaev},\ and\ \citenamefont
  {Nisan}}]{Aharonov98}%
  \BibitemOpen
  \bibfield  {author} {\bibinfo {author} {\bibfnamefont {D.}~\bibnamefont
  {Aharonov}}, \bibinfo {author} {\bibfnamefont {A.}~\bibnamefont {Kitaev}}, \
  and\ \bibinfo {author} {\bibfnamefont {N.}~\bibnamefont {Nisan}},\ }in\ \href
  {\doibase 10.1145/276698.276708} {\emph {\bibinfo {booktitle} {Proceedings of
  the Thirtieth Annual ACM Symposium on Theory of Computing}}},\ \bibinfo
  {series and number} {STOC '98}\ (\bibinfo  {publisher} {Association for
  Computing Machinery},\ \bibinfo {address} {New York, NY, USA},\ \bibinfo
  {year} {1998})\ p.\ \bibinfo {pages} {20–30}\BibitemShut {NoStop}%
\bibitem [{\citenamefont {Watrous}(2018)}]{Watrous18}%
  \BibitemOpen
  \bibfield  {author} {\bibinfo {author} {\bibfnamefont {J.}~\bibnamefont
  {Watrous}},\ }\href@noop {} {\emph {\bibinfo {title} {The theory of quantum
  information}}}\ (\bibinfo  {publisher} {Cambridge university press},\
  \bibinfo {address} {New York},\ \bibinfo {year} {2018})\BibitemShut {NoStop}%
\end{thebibliography}%
\end{document}